\documentclass[a4paper,USenglish,cleveref, autoref, thm-restate]{lipics-v2021}

\usepackage{cite}
\usepackage{booktabs}
\usepackage{algorithm}
\usepackage[noend]{algpseudocode}
\usepackage{algorithmicx}

\usepackage{subcaption}
\usepackage{times}
\usepackage{multicol}
\usepackage{xcolor}
\usepackage{wrapfig}

\algdef{SE}[SUBALG]{Indent}{EndIndent}{}{\algorithmicend\ }%
\algtext*{Indent}
\algtext*{EndIndent}

\newcommand{\ignore}[1]{}

\title{EEMARQ: Efficient Lock-Free Range Queries with Memory Reclamation} 

\author{Gali Sheffi}{Department of Computer Science, Technion, Haifa, Israel}{sheffiga@gmail.com}{}{}

\author{Pedro Ramalhete}{Cisco Systems, Switzerland}{pramalhe@gmail.com}{}{}

\author{Erez Petrank}{Department of Computer Science, Technion, Haifa, Israel}{erez@cs.technion.ac.il}{}{}

\authorrunning{G. Sheffi, P. Ramalhete and E. Petrank} 

\Copyright{G. Sheffi, P. Ramalhete and E. Petrank} 



\category{} 

\relatedversion{} 

\EventEditors{John Q. Open and Joan R. Access}
\EventNoEds{2}
\EventLongTitle{42nd Conference on Very Important Topics (CVIT 2016)}
\EventShortTitle{CVIT 2016}
\EventAcronym{CVIT}
\EventYear{2016}
\EventDate{December 24--27, 2016}
\EventLocation{Little Whinging, United Kingdom}
\EventLogo{}
\SeriesVolume{42}
\ArticleNo{23}

\begin{document}
\nolinenumbers
\begin{titlepage}

\maketitle

\begin{abstract}
Multi-Version Concurrency Control (MVCC) is a common mechanism for achieving linearizable range queries in database systems and concurrent data-structures.
The core idea is to keep previous versions of nodes to serve range queries, while still providing atomic reads and updates.
Existing concurrent data-structure implementations, that support linearizable range queries, are either slow, use locks, or rely on blocking reclamation schemes. 
We present EEMARQ, the first scheme that uses MVCC with lock-free memory reclamation to obtain a fully lock-free data-structure supporting linearizable inserts, deletes, contains, and range queries.
Evaluation shows that EEMARQ outperforms existing solutions across most workloads, with lower space overhead and while providing full lock freedom.
\end{abstract}

\end{titlepage}

\newpage

\section{Introduction} \label{sec-intro}

Online Analytical Processing (OLAP) transactions
are typically long and may read data from a large subset of the records in a database~\cite{plattner2009common,reddy2010data}. As such, analytical workloads pose a significant challenge in the design and implementation of efficient concurrency controls for database management systems (DBMS).
%
Two-Phase Locking (2PL)~\cite{thomasian1991performance} is sometimes used, but locking each record before it is read
implies a high synchronization cost and, moreover, the inability to modify these records over long periods.
%
Another way to deal with OLAP queries is to use Optimistic Concurrency Controls~\cite{harder1984observations}, where the records are not locked, but they need to be validated at commit time to guarantee serializability~\cite{papadimitriou1979serializability}. If during the time that the analytical transaction executes, there is any modification to one of these records, the analytical query will have to abort and restart. Aborting can prevent long read-only queries from ever completing. 

DBMS designers typically address these obstacles using Multi-Version Concurrency Control (MVCC).
MVCC's core idea is to keep previous versions of a record, allowing transactions to read data from a fixed point in time.
For managing the older versions, each record is associated with its list of older records. Each version record contains a copy of an older version, and its respective time stamp, indicating its commit time.
Each update of the record's values adds a new version record to the top of the version list, and every read of an older version is done by traversing the version list, until the relevant timestamp is reached.

Keeping the version lists relatively short is fundamental for high performance and low memory footprint~\cite{bottcher2019scalable,lee2016hybrid,diaconu2013hekaton,kim2020long,kim2021rethink}.
However, the version lists should be carefully pruned, as a missing version record can be harmful to an ongoing range query.
In the Database setting, the common approach is to garbage collect old versions once the version list lengths exceed a certain threshold~\cite{bottcher2019scalable,lee2016hybrid,postgresql1996postgresql}.
During garbage collection, the entire database is scanned for record instances with versions that are not required for future scanners. 
However, garbage collecting old versions with update-intensive workloads considerably slows down the entire system.
An alternative approach is to use transactions and associate each scan operation with an undo-log~\cite{bartholomew2014mariadb,mysql2001mysql}.
But this requires allocating memory for storing the undo logs and a safe memory reclamation scheme to recycle log entries. 

In contrast to the DBMS approach, many concurrent in-memory data-structures implementations do not provide an MVCC mechanism, and simply give up on range queries. 
Many map-based data-structures provide linearizable~\cite{herlihy1990linearizability} insertions, deletions and membership queries of single items. 
Most data-structures that do provide range queries are blocking. They either use blocking MVCC mechanisms~\cite{nelson2021bundled,arbel2018harnessing,bronson2010practical}, or rely on the garbage collector of managed programming languages~\cite{petrank2013lock,basin2017kiwi,fatourou2019persistent,winblad2021lock} which is blocking.
It may seem like lock-free data-structures can simply employ a lock-free reclamation scheme with an MVCC mechanism to obtain full lock-freedom, but interestingly this poses a whole new challenge. 

Safe manual reclamation (SMR)~\cite{DBLP:journals/corr/abs-2107-13843,singh2021nbr,ramalhete2017brief,wen2018interval,michael2004hazard} 
algorithms rely on \textit{retire()} invocations by the program, announcing that a certain object has been unlinked from a data-structure. 
The task of the SMR mechanism is to decide which retired objects can be safely reclaimed, making their memory space available for re-allocation.
Most SMR techniques~\cite{singh2021nbr,ramalhete2017brief,wen2018interval,michael2004hazard} heavily rely on the fact that retired objects are no longer reachable for threads that concurrently traverse the data structure.
Typically, objects are retired when deleted from the data-structure.
However, when using version lists, if we retire objects when they are deleted from the current version of the data-structure, they would still be reachable by range queries via the version list links.
Therefore, it is not safe to recycle old objects with existing memory reclamation schemes.
Epoch-based reclamation (EBR)~\cite{fraser2004practical,brown2015reclaiming} is an exception to the rule because it only requires that an operation does not access nodes that were retired before the operation started. Namely, an existing range query prohibits reclamation of any deleted node, and subsequent range queries do not access these nodes, so they can be deleted safely. 
Therefore, EBR can be used without making any MVCC-specific adjustments, and indeed  EBR is used in many in-memory solutions~\cite{wei2021constant,arbel2018harnessing,nelson2021bundled}. However, EBR is not robust~\cite{singh2021nbr,wen2018interval,DBLP:journals/corr/abs-2107-13843}. I.e., a slow executing thread may prevent the reclamation of an unbounded number of retired objects, which may affect performance, and theoretically block all new allocations.

One lock-free memory reclamation scheme that can be adopted to provide lock-free support for MVCC is the VBR optimistic memory reclamation scheme~\cite{DBLP:journals/corr/abs-2107-13843,sheffi2021vbr}. 
VBR allows a program to access reclaimed space, but it raises a warning when accessed data has been re-allocated. This allows the program to retry the access with refreshed data. 
The main advantages of VBR are that it is lock-free, it is fast, and it has very low memory footprint. Any retired object can be immediately reclaimed safely. The main drawback\footnote{VBR also necessitates type-preservation. However, this does not constitute a problem in our setting, as all allocated memory objects are of the same type. For more details, see Section~\ref{sec-algoritm}.} is that reclaimed memory cannot be returned to the operating system, but must be kept for subsequent node allocations, as program threads may still access reclaimed nodes. %
Similarly to other schemes, the correctness of VBR depends on the inability of operations to visit previously retired objects. In data structures that do not use multi versions, the deleting thread adequately retires a node after disconnecting it from the data structure. But in the MVCC setting, disconnected nodes remain connected via the version list, foiling correctness of VBR. 

In this paper we modify VBR to work correctly in the MVCC setting. 
It turns out that for the specific case of old nodes in the version list, correctness can be obtained. VBR keeps a slow-ticking epoch clock and it maintains a birth epoch field for each node. Interestingly, this birth epoch can be used to tell whether a retired node in the version list has been re-allocated. As we show in this paper, an invariant of nodes in the version list is that they have non-increasing birth-epoch numbers. Moreover, if one of the nodes in the version list is re-allocated, then this node must foil the invariant. Therefore, when a range query traverses a version list to locate the version to use for its traversal, it can easily detect a node that has been re-allocated and whose data is irrelevant. When a range query detects such re-allocation, it restarts. As shown in the evaluation, restarting happens infrequently with VBR and the obtained performance is the best among existing schemes in the literature.
Using the modified VBR in the MVCC setting, a thread can delete a node and retire it after disconnecting it from the data structure (and while it is still reachable via version list pointers).

 

Recently, two novel papers~\cite{wei2021constant,nelson2021bundled} presented efficient MVCC-based key-value stores. The main new idea is to track and keep a version list of modified fields only and not of entire nodes. For many data-structures, all fields are immutable except for one or two pointer fields.
For such data-structures, it is enough to keep a version list of pointers only. 
The first paper proposed a lock-free mechanism, based on {\em versioned CAS objects} (vCAS)~\cite{wei2021constant}, and the second proposed a blocking mechanism, based on {\em Bundle objects} (Bundles)~\cite{nelson2021bundled}.
While copying one field instead of the entire node reduces the space overhead, the resulting indirection is harmful for performance: to dereference a pointer during a traversal, one must first move to the top node of the version list, and then use its pointer to continue with the traversal. As we show in Section~\ref{sec-evaluation}, this indirection suffers from high overheads in read-intensive workloads. Bundles ameliorate this overhead by caching the most recent pointer in the node to allow quick access for traversals that do not use older versions.
However, range queries still need to dereference twice as many references during a traversal.
The vCAS approach presents a more complicated optimization that completely eliminates indirection
(which we further discuss in Section~\ref{sec-algoritm}). However, its applicability depends on assumptions on the original data-structure that many data structures do not satisfy. Therefore, it is unclear how it can be integrated into general state-of-the-art concurrent data-structures.
%
In terms of memory reclamation, both schemes use the EBR technique (that may block in theory, due to allocation heap exhaustion). This means that even vCAS does not provide a lock-free range query mechanism. A subsequent paper on MVCC-specific garbage collectors~\cite{ben2021space}, provides a robust memory reclamation method for collecting the version records, but it does not deal with the actual data-structure nodes.

In this paper we present EEMARQ (End-to-End lock-free MAp with Range Queries), a design for a lock-free in-memory map with MVCC and a robust lock-free memory management scheme. 
EEMARQ provides linearizable and high-performance inserts, deletes, searches, and range queries. 

The design starts by applying vCAS to a linked-list. The linked-list is simple, and so it allows applying vCAS easily. 
However, the linked-list does not satisfy the optimization assumptions required by~\cite{wei2021constant}, and so a (non-trivial) extension is required to fit the linked-list. 
In the unoptimized variant, there is a version node for each of the updates (insert or delete) and this node is used to route the scans in the adequate timestamp. But 
this is costly, because traversals end up accessing twice the number of the original nodes. A natural extension 
is to associate the list nodes with the data that is originally kept in the version nodes. 
We augment the linked-list construction to allow the optimization of~\cite{wei2021constant} for it. They use a clever mapping from the original version nodes to existing list nodes, that allows moving the data from version nodes to list nodes, and then elide the version nodes. Now traversals need no extra memory accesses to version nodes. 
This method 
is explained in Section~\ref{sec-range-queries}.

Second, we extend VBR by adding support for reachable retired nodes on the version list. 
The extended VBR allows keeping retired reachable version nodes in the data structure (which the original VBR forbids) while maintaining high performance, lock-freedom, and robustness. 

Finally, we deal with the inefficiency of a linked-list by adding a fast indexing to the linked-list nodes. A fast index can be obtained from a binary search tree or a skip list. But the advantage we get from separating the linked-list from the indexing mechanism is that we do not have to maintain versions for the index (e.g., for the binary search tree or the skip list), but only for the underlying linked-list. 
This separation between the design of the versioned linked-list and the non-versioned index, simplifies each of the sub-designs, and also obtains high performance, because operations on the index do not need to maintain versions. 
Previous work uses this separation idea between the lower level of the skip list or leaves of the tree from the rest for various goals (e.g.,~\cite{zuriel2019efficient,nelson2021bundled}). 

The combination of all these three ideas, i.e., optimized versioned linked-list, extended VBR, and independent indexing, yields a highly performant data structure design with range queries. 
%
Evaluation shows that EEMARQ outperforms both vCAS and Bundles (while providing full lock-freedom).
A proof of correctness is provided in the appendix. 

\ignore{
Like previous solutions~\cite{wei2021constant,nelson2021bundled,arbel2018harnessing}, our map uses a global epoch counter for marking data-structure updates with their respective timestamp. It requires neither complex helping mechanisms nor any form of version records. 
I.e., as opposed to~\cite{wei2021constant}, a traversal of the data-structure does not involve any form of indirection.
Moreover, although the strong assumptions, presented in~\cite{wei2021constant}, do not hold for our implementation, we prove that a similar technique can be used for avoiding indirection.
Our baseline is a new implementation of a linearizable and lock-free linked-list (presented in Section~\ref{sec-algoritm}), which offers efficient updates, single-point reads and multi-point range queries.
For enhancing searches and updates, we allow a simple integration of faster data-structures, serving as our fast-access index.
As opposed to the vCAS technique, we do not keep old versions of nodes and pointers, required only for fast indexing (e.g., the internal nodes in a binary search tree or the upper levels pointers in a skiplist).
Finally, memory is managed using an enhanced VBR variant, providing robustness and maintaining both efficiency and lock-freedom. As VBR cannot be integrated as is, we suggest the needed changes and prove that they preserve correctness.
}

\section{Related Work} \label{sec-related}

Existing work on linearizable range queries in the shared-memory setting includes many solutions which are not based on MVCC techniques. 
Some data-structure interfaces originally include a tailor-made range query operation. E.g., there are trees~\cite{bronson2010practical,brown2012range,fatourou2019persistent,winblad2021lock}, hash tries~\cite{prokopec2012concurrent}, queues~\cite{nikolakopoulos2015consistency,nikolakopoulos2015concurrent,prokopec2015snapqueue}, skip lists~\cite{avni2013leaplist} and graphs~\cite{kallimanis2016wait} with a built-in range query mechanism.
Other and more general solutions execute range queries by explicitly taking a snapshot of the whole data-structure, followed by collecting the set of keys in the given range~\cite{afek1993atomic,attiya2008partial}. 
The Snapcollector~\cite{petrank2013lock} forces the cooperation of all executing threads while a certain thread is scanning the data-structure. Despite being lock-free and general, the Snapcollector's memory and performance overheads are high. The Snapcollector was enhanced to support range queries that do not force a snapshot of the whole data-structure~\cite{chatterjee2017lock}. However, this solution still suffers from major time and space overheads.

Another way to implement range queries is to use transactional memory ~\cite{fernandes2011lock,keidar2015multi,perelman2011smv,perelman2010maintaining}. Transactions can either be implemented in software or in hardware, allowing range queries to take effect atomically. Although transactions may seem as ideal candidates for long queries, software transactions incur high performance overheads and hardware transactions frequently abort when accessing big memory chunks.
Read-log-update (RLU)~\cite{matveev2015read} borrows software transaction techniques and extends read-copy-update (RCU)~\cite{mckenney1998read} to support multiple updates. RLU yields relatively simple implementations, but its integration involves re-designing the entire data-structure. In addition, similarly to RCU, it suffers from high overheads in write-extensive workloads.

Arbel-Raviv and Brown exploited the EBR manual reclamation to support range queries~\cite{arbel2018harnessing}. 
Their technique uses EBR's global epoch clock for  associating data-structure items with their insertion and deletion timestamps. EEMARQ uses a similar technique for associating data-structure modifications with respective timestamps.
They also took advantage of EBR's retire-lists, in order to locate nodes that were deleted during a range query scan. However, while their solution indeed avoids extensive scanning helping when deleting an item from the data-structure (as imposed by the Snapcollector~\cite{petrank2013lock}), it may still impose significant performance overheads (as shown in Section~\ref{sec-evaluation}). EEMARQ minimizes these overheads by keeping a reachable path to deleted nodes (for more details, see Section~\ref{sec-range-queries}).

\subparagraph{Multi-Version Concurrency Control}

MVCC easily provides isolation between concurrent reads and updates. I.e., range queries can work on a consistent view of the data, while not interfering with update operations. This powerful technique is widely used in commercial systems~\cite{farber2012sap,diaconu2013hekaton}, as well as in research-oriented DBMS~\cite{kemper2011hyper,postgresql1996postgresql}, in-memory shared environments~\cite{wei2021constant,nelson2021bundled} and transactional memory~\cite{fernandes2011lock,keidar2015multi,kumar2014timestamp,perelman2011smv,perelman2010maintaining}.
MVCC has been investigated in the DBMS setting for the last four decades, both from a theoretical~\cite{bernstein1982concurrency,bernstein1983multiversion,papadimitriou1984concurrency,wu2017empirical,kim2021rethink} and a practical~\cite{mysql2001mysql,diaconu2013hekaton,lee2016hybrid,postgresql1996postgresql,kemper2011hyper} point of view.
A lot of effort has been put in designing MVCC implementations and addressing the unwanted side-effects of long version lists. This issue is crucial both for accelerating version list scans and for reducing the contention between updates and garbage collection.
Accelerating DBMS version list scans (independently of the constant need to prone them) has been investigated in~\cite{kim2021rethink}.
Most DBMS-related work that focuses on this issue, tries to minimize the problem by eagerly collecting unnecessary versions~\cite{bottcher2019scalable,farber2012sap,lu2013generic,postgresql1996postgresql,mysql2001mysql,kim2020long,antonopoulos2019constant}. 


\subparagraph{Safe Memory Reclamation}

Most existing in-memory environments that enable linearizable range queries, avoid managing their memory, by relying on automatic (and blocking) garbage collection of old versions~\cite{petrank2013lock,basin2017kiwi,fatourou2019persistent,winblad2021lock}.
Solutions that do manually manage their allocated memory~\cite{wei2021constant,nelson2021bundled,arbel2018harnessing}, use EBR for safe reclamation.
In EBR, a shared epoch counter is incremented periodically, and upon each operation invocation, the threads announce their observed epoch. 
The epoch clock can be advanced when no executing thread has an announcement with a previous epoch.
During reclamation, only nodes that had been retired at least two epochs ago are reclaimed.
EBR is safe for MVCC because a running scan prevents the advance of the epoch clock, and also the reclamation of any node in the data-structure that was not deleted before the scan.


The MVCC-oriented garbage collector from~\cite{ben2021space} incorporates reference counting (RC)~\cite{correia2021orcgc,detlefs2002lock,herlihy2005nonblocking}, in order to maintain lock-freedom while safely reclaiming old versions. Any object can be immediately reclaimed once its reference count reaches zero, without the need in invoking explicit \textit{retire()} calls.
While RC simplifies reclamation, it incurs high performance overheads and does not guarantee a tight bound on unreclaimed garbage (i.e., it is not robust).

Other reclamation schemes were also considered when designing EEMARQ. NBR~\cite{singh2021nbr} was one of the strongest candidates, as it is fast and lock-free (under some hardware assumptions). However, it is not clear whether NBR can be integrated into a skip list implementation (which serves as one of our fast indexes).
Pointer–based reclamation methods (e.g., Hazard Pointers~\cite{michael2004hazard}) allow threads to protect specific objects (i.e., temporarily prevent their reclamation), by announcing their future access to these objects, or publishing an announcement indicating the protection of a bigger set of objects (e.g., Hazard Eras~\cite{ramalhete2017brief}, Interval-Based Reclamation~\cite{wen2018interval}, Margin Pointers~\cite{solomon2021efficiently}).
Although these schemes are robust (as opposed to EBR), it is unclear how they can be used in MVCC environments. They require that reaching a reclaimed node from a protected one would be impossible (even if the protected node is already retired). I.e., they require an explicit unlinking of old versions before retiring them. Besides the obvious performance overheads, it may affect robustness (as very old versions would not be reclaimed).
The garbage collector from~\cite{ben2021space} uses Hazard Eras for unlinking old versions (to be eventually collected using an RC-based strategy), but it has not been evaluated in practice. 

\section{The Algorithm} \label{sec-algoritm}



In this Section we present EEMARQ's design. We start by introducing a new lock-free linearizable linked-list implementation in Section~\ref{sec-list-implementation}.
The list implementation is based on Harris's lock-free linked-list~\cite{harris2001pragmatic}, and includes the standard \emph{insert()}, \emph{remove()} and \emph{contains()} operations. In Section~\ref{sec-range-queries} we explain how to add a linearizable and efficient range query operation. We describe the integration of the designated robust SMR algorithm in Section~\ref{sec-mvcc-vbr}, and explain how to improve performance by adding an external index in Section~\ref{sec-index}.
A full linearizability and lock-freedom proof for our implementation (including the range queries mechanism) appears in Appendix~\ref{sec-correctness}. Additional correctness proofs for the SMR and fast indexing integration appear in Appendix~\ref{sec-vbr-correctness} and~\ref{sec-index-correctness}, respectively. 

As discussed in~\cite{wei2021constant}, node-associated version lists introduce an extra level of indirection per node access. Methods that use designated version objects for recording updates, suffer from high overheads, especially in read-intensive workloads.
For avoiding this level of indirection, we introduce a new variant of Harris's linked-list. 
In our new variant, there is no need to store any update-related data in designated version records, since it can be stored directly inside nodes, in a well-defined manner.
Associating each node with a single data-structure update (i.e., an insertion or a removal) is challenging. 
%
Typically, an insert operation includes a single update, physically inserting a node into the list.
A remove operation involves a marking of the target node's {\em next} pointer (serving as its logical deletion, and the operation's linearization point in many existing implementations) and a following physical removal from the list. In other words, each node may be associated with multiple list updates.
Since the target node's physical deletion is not the linearization point of any operation, there is no need to record this update. However, each node may still be associated with either one or two updates throughout an execution (i.e., its logical insertion and deletion). 

\begin{figure}[h] 
  \centering
  \includegraphics[width=\linewidth,trim= 0in 2.6in 0in 0in]{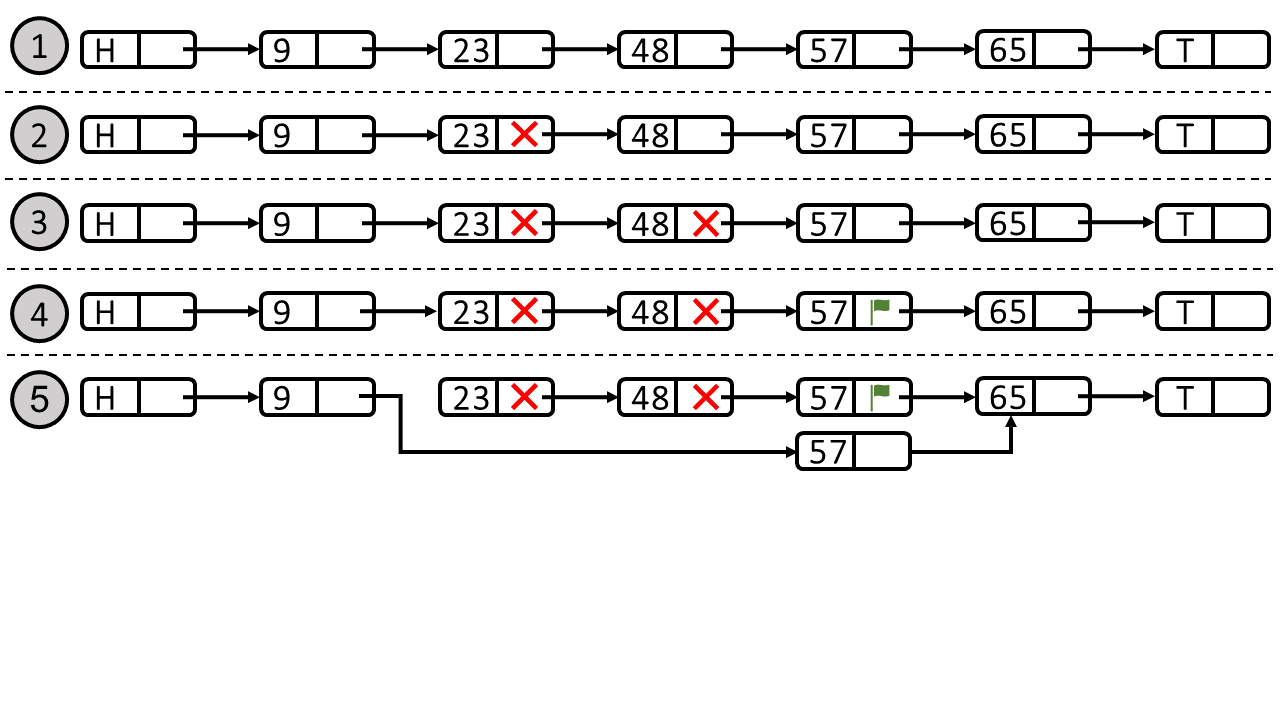}
  \caption{Removing nodes 23 and 48 from the linked-list. At stages 1-4, the list logically contains 5 nodes and at stage 5, it logically contains 3 nodes. The logical deletions of nodes 23 and 48 are executed via marking them (stages 2 and 3, respectively), flagging node 57 (stage 4) and inserting a new 57 representative instead of the three of them (stage 5). Nodes 23, 48, and the flagged 57 are then retired.} \label{fig-physical-delete}
\end{figure}

In our linked-list implementation, some new node is inserted into the list during every physical update of the list (either a physical insertion or deletion), which obviously yields the desirable association between nodes and data-structure updates (each node is associated with the update that involved its insertion into the list).
The node inserted during a physical insertion is simply the inserted node, logically inserted into the list.
The node inserted during a physical removal is a designated new node that replaces the deleted node's successor, and is physically inserted together with the physical removal of the deleted node~\footnote{When multiple nodes are physically removed together, it replaces the successor of the last node in the sequence of deleted nodes.}.

Figure~\ref{fig-physical-delete} shows an example for inserting a new node during deletion.
This list illustration shows the list layout throughout the deletion procedure of two nodes.
At the first stage, the list contains five nodes, ordered by their keys (together with the \emph{head} and \emph{tail} sentinels). Then, some thread marks node 23 for deletion. At this point, before physically removing it, some other thread marks node 48 for deletion. I.e., both nodes must be physically unlinked together, as successive marked nodes. In order to remove the marked nodes, node 57 is flagged, as it is the successor of the last marked node in the sequence. Node 57 is flagged for making sure that its \emph{next} pointer does not change. Finally, all three nodes are physically unlinked from the list, together with the physical insertion of a new node, representing the old flagged one. I.e., although all three nodes were physically removed from the list, node 57 was not logically removed, as it was replaced by a new node with the same key.
After the physical deletion at stage 5, the three nodes are retired (for more details, see Section~\ref{sec-mvcc-vbr}).
As opposed to Harris's implementation, the linearization point of both deletions is this physical removal, which atomically inserts the new representative into the list.
Our mapping from modifications to nodes, maps this deletion to the new node, inserted at stage 5 (which represents the deletion of both nodes). In Section~\ref{sec-range-queries} we explain how this mapping is used for executing range queries.


\subsection{The Linked-List Implementation} \label{sec-list-implementation}

Our linked-list implementation, together with the list node class, is presented in Algorithm~\ref{pseudo-list}. The simple pointer access methods implementation (e.g., \emph{mark()} in line~\ref{alg-remove-mark} and \emph{getRef()} in line~\ref{alg-find-curr-gets-curr-next}) appear in Appendix~\ref{sec-auxiliary}.
In a similar way to Harris's list, the API includes the \emph{insert()}, \emph{remove()} and \emph{contains()} operations~\footnote{The \emph{rangeQuery()} operation is added in Section~\ref{sec-range-queries}. In addition, lines marked in blue in Algorithm~\ref{pseudo-list} can be ignored at this point. They will also be discussed in Section~\ref{sec-range-queries}}.
The \emph{insert()} operation (lines~\ref{alg-insert-call}-\ref{alg-insert-return-no-val}) receives a key and a value. If there already exists a node with the given key in the list, it returns its value (line~\ref{alg-insert-return-curr-val}). Otherwise, is adds a new node with the given key and value to the list, and returns a designated NO\_VAL answer (line~\ref{alg-insert-return-no-val}). 
The \emph{remove()} operation (lines~\ref{alg-remove-call}-\ref{alg-remove-return-curr-val}) receives a key. If there exists a node with the given key in the list, it removes it and returns its value (line~\ref{alg-remove-return-curr-val}). Otherwise, it returns NO\_VAL (line~\ref{alg-remove-return-no-val}).
The \emph{contains()} operation (lines~\ref{alg-contains-call}-\ref{alg-contains-return-curr-val}) receives a key. If there exists a node with the given key in the list, it returns its value (line~\ref{alg-contains-return-curr-val}). Otherwise, it returns NO\_VAL (line~\ref{alg-contains-return-no-val}). 

\begin{algorithm*}[!ht]
\setlength\multicolsep{0pt}
\small
\begin{multicols}{2}
\begin{algorithmic}[1]

\State \textbf{class} Node \label{alg-class-node}
\Indent
    \State {\color{blue}  Long ts} \label{alg-class-node-ts}
    \State K key \label{alg-class-node-key}
    \State V value \label{alg-class-node-value}
    \State Node* next \label{alg-class-node-next}
    \State {\color{blue} Node* prior} \label{alg-class-node-prior}
    
\EndIndent

\Procedure{insert}{key, val} \label{alg-insert-call}
    \While{(true)}
        \State pred, curr $\leftarrow$ FIND(key) \label{alg-insert-find}
        \If{(curr $\rightarrow$ key == key)} \textbf{return} curr $\rightarrow$ val  \label{alg-insert-return-curr-val}
        \EndIf
        \State newNode := \textbf{alloc}(key, val{\color{blue}, $\bot$}) \label{alg-insert-alloc}
        \State newNode $\rightarrow$ next := curr \label{alg-insert-update-next}
        \State {\color{blue} newNode $\rightarrow$ prior := curr} \label{alg-insert-update-prior}
        \If{(CAS(\&pred $\rightarrow$ next, curr, newNode))} \label{alg-insert-update-pred} 
            \State {\color{blue} CAS(\&newNode $\rightarrow$ ts, $\bot$, getTS())} \label{alg-insert-update-ts}
            \State \textbf{return} NO\_VAL \label{alg-insert-return-no-val}  
        \EndIf        
    \EndWhile
\EndProcedure

\Procedure{remove}{key} \label{alg-remove-call}
    \While{(true)}
        \State pred, curr $\leftarrow$ FIND(key) \label{alg-remove-find}
        \If{(curr $\rightarrow$ key $\neq$ key)} \textbf{return} NO\_VAL  \label{alg-remove-return-no-val}
        \EndIf
        \If{(!\textbf{mark}(curr))} \textbf{continue} \label{alg-remove-mark} 
        \EndIf
        \State FIND(key) {\color{red} \Comment{for the physical deletion} \label{alg-remove-trim}}
        \State \textbf{return} curr $\rightarrow$ val {\color{red} \Comment{successful deletion}} \label{alg-remove-return-curr-val}
    \EndWhile
\EndProcedure

\Procedure{contains}{key} \label{alg-contains-call}
    \State pred, curr $\leftarrow$ FIND(key) \label{alg-contains-find}
    \If{(curr $\rightarrow$ key $\neq$ key)} \textbf{return} NO\_VAL  \label{alg-contains-return-no-val}
    \Else\ \textbf{return} curr $\rightarrow$ val \label{alg-contains-return-curr-val} 
    \EndIf
\EndProcedure


\Procedure{find}{key} \label{alg-find-call}
    \State \textbf{retry:} 
    \State pred := head \label{alg-find-head}
    \State predNext := pred $\rightarrow$ next \label{alg-find-head-next}
    \State curr := \textbf{getRef}(predNext) \label{alg-find-curr-gets-head-next}
        \While{(true)} \label{alg-find-while-true}
            \State \textbf{while}(\textbf{isMarkdOrFlagged}(curr $\rightarrow$ next))             \label{alg-find-while-is-marked}
            \Indent
                \If{(!\textbf{getRef}(curr $\rightarrow$ next))} \textbf{break} \label{alg-find-if-tail}
                \EndIf
                \State curr := \textbf{getRef}(curr $\rightarrow$ next) \label{alg-find-curr-gets-curr-next}
            \EndIndent
            \If{(curr $\rightarrow$ key $\geq$ key)} \textbf{break} \label{alg-find-key-bigger}
            \EndIf
            \State pred := curr \label{alg-find-pred-gets-curr}
            \State predNext := pred $\rightarrow$ next \label{alg-find-get-pred-next}
            \If{(\textbf{isMarkedOrFlagged}(predNext))}  \label{alg-find-pred-marked-flagged}
                \State \textbf{goto retry}
            \EndIf
            \State curr := \textbf{getRef}(predNext)  \label{alg-find-curr-gets-pred-next}
        \EndWhile
        \State {\color{blue} CAS(\&pred $\rightarrow$ ts, $\bot$, getTS())} \label{alg-find-update-pred-ts}
        \If{(predNext $\neq$ curr)} \label{alg-find-if-not-adjacent}
            \If{(!TRIM(pred, \textbf{getRef}(predNext)))} \label{alg-find-goto-after-trim}
                \State \textbf{goto retry} \label{alg-find-trim}
            \EndIf
            \State predNext := pred $\rightarrow$ next \label{alg-find-pred-next-after-trim}
            \If{(\textbf{isMarkedOrFlagged}(predNext))}  \label{alg-find-pred-marked-flagged-after-trim}
                \State \textbf{goto retry}
            \EndIf
            \State curr := predNext \label{alg-find-curr-gets-pred-next-after-trim}
            \If{( \textbf{isMarkedOrFlagged}(curr $\rightarrow$} \label{alg-find-key-bigger-after-trim}
            
            \Statex $\; \; \; \; \; \; \; \; \; \; \; \; \; \; \; \; \; \; \; \; \; \; \; \; \; \; \; \; \; \;$ next $\vee$ curr $\rightarrow$ key < key)
            \State \textbf{goto retry} 
            \EndIf
             
        \EndIf
        \State {\color{blue} CAS(\&curr $\rightarrow$ ts, $\bot$, getTS())} \label{alg-find-update-curr-ts}
        \State \textbf{return} pred, curr \label{alg-find-return-window}
\EndProcedure


\Procedure{trim}{pred, victim} \label{alg-trim-call}
    \State curr := victim \label{alg-trim-curr-gets-victim}
    \While{(\textbf{isMarked}(curr $\rightarrow$ next))} \label{alg-trim-while}
        \State curr := \textbf{getRef}(curr $\rightarrow$ next) \label{alg-trim-curr-gets-curr-next}
    \EndWhile
    \State {\color{blue} CAS(\&curr $\rightarrow$ ts, $\bot$, getTS())} \label{alg-trim-update-curr-ts}
    \If{(!\textbf{flag}(curr) $\wedge$ !\textbf{isFlagged}(curr $\rightarrow$ next))} \label{alg-trim-flag}
        \State \textbf{return} false \label{alg-trim-flag-return-false}
    \EndIf
    \State succ := \textbf{getRef}(curr $\rightarrow$ next) \label{alg-trim-read-succ}
    \State {\color{blue} \textbf{if} (succ)  CAS(\&succ $\rightarrow$ ts, $\bot$, getTS())} \label{alg-trim-update-succ-ts}
    \State newCurr := alloc(curr $\rightarrow$ key, curr $\rightarrow$ val{\color{blue}, $\bot$}) \label{alg-trim-alloc}
    \State newCurr $\rightarrow$ next := succ \label{alg-trim-init-next}
    \State {\color{blue}  newCurr $\rightarrow$ prior := victim} \label{alg-trim-update-prior}
    \If{(CAS(\&pred $\rightarrow$ next, victim, newCurr))} \label{alg-trim-cas} 
        \State {\color{blue}  CAS(\&newCurr $\rightarrow$ ts, $\bot$, getTS())} \label{alg-trim-update-new-curr-ts}
        \State \textbf{return} true
    \EndIf
    \State \textbf{return} false \label{alg-trim-return-false}
    
\EndProcedure

\end{algorithmic}
\end{multicols}
\caption{Our Linked-List Implementation.}
\label{pseudo-list}
\end{algorithm*}

All three API operations use the \emph{find()} auxiliary method (lines~\ref{alg-find-call}-\ref{alg-find-return-window}), which receives a key and returns pointers to two nodes, \emph{pred} and \emph{curr} (line~\ref{alg-find-return-window}). As in Harris's implementation, it is guaranteed that at some point during the method execution, both nodes are consecutive reachable nodes in the list, \emph{pred}'s key is strictly smaller than the input key, and \emph{curr}'s key is equal or bigger than the given key. I.e., if \emph{curr}'s key is strictly bigger than the input key, it is guaranteed that there is no node with the given input key in the list at this point. 
The method traverses the list, starting from the \emph{head} sentinel node (line~\ref{alg-find-head}), and until it gets to an unmarked and unflagged node with a key which is at least the input key (line~\ref{alg-find-key-bigger}).
Recall that the two output variables are guaranteed to have been reachable, adjacent, unmarked and not flagged at some point during the method execution. 
Therefore, as long as the current traversed node is either marked or flagged (checked in line~\ref{alg-find-while-is-marked}), the traversal continues, regardless of the current key (lines~\ref{alg-find-while-is-marked}-\ref{alg-find-curr-gets-curr-next}).
Once the traversal terminates (either in line~\ref{alg-find-if-tail} or~\ref{alg-find-key-bigger}), if the current two nodes, saved in the \emph{pred} and \emph{curr} variables, are adjacent (the condition checked in line~\ref{alg-find-if-not-adjacent} does not hold), then the method returns them in line~\ref{alg-find-return-window}.
Otherwise, similarly to the original implementation, the method is also in charge of physically removing marked nodes from the list.

As we are going to discuss next, our physical removal procedure, as depicted in Figure~\ref{fig-physical-delete}, is slightly different from the original one~\cite{harris2001pragmatic}.
Physical deletions are executed via the \emph{trim()} auxiliary method (lines~\ref{alg-trim-call}-\ref{alg-trim-return-false}).
Although nodes are still marked for deletion in our implementation (line~\ref{alg-remove-mark}), their successful marking does not serve as the removal linearization point. I.e., reachable marked nodes are still considered as list members.
The \emph{trim()} method receives two nodes as its input parameters, \emph{pred} and \emph{victim}. \emph{victim} is the physical removal candidate, and is assumed to already be marked. \emph{pred} is assumed to be \emph{victim}'s predecessor in the list, and to be neither marked nor flagged. 
As depicted in Figure~\ref{fig-physical-delete}, consecutive marked nodes are removed together.
Therefore, the method traverses the list, starting from \emph{victim}, for locating the first node which is not marked (lines~\ref{alg-trim-while}-\ref{alg-trim-curr-gets-curr-next}).
When such a node is found, the method tries to flag its \emph{next} pointer, for freezing it until the removal procedure is done.
In general, pointers are marked and flagged using their two least significant bits (which are practically redundant when reading node address aligned to a word). Both marked and flagged pointers are immutable, and a pointer cannot be both marked and flagged. 
Therefore, the flagging trial in line~\ref{alg-trim-flag} fails if \emph{curr}'s \emph{next} pointer is either marked or flagged. 
If the flagging trial is unsuccessful, and not because some other thread has already flagged \emph{curr}'s \emph{next} pointer, the method returns in line~\ref{alg-trim-flag-return-false}. 
Otherwise, a new node is created in order to replace the flagged one (lines~\ref{alg-trim-alloc}-\ref{alg-trim-update-prior}). Note that since this node's \emph{next} pointer is flagged (i.e., immutable), it is guaranteed that the new node points to the original one's current successor.
The actual trimming is executed in line~\ref{alg-trim-cas}. If the compare-and-swap (CAS) is successful, then the sequence of marked nodes, together with the single flagged one (at the end of the sequence), are atomically removed from the list, together with the insertion of the new copy of the flagged node (the new copy is neither flagged nor marked).

As the physical removal is necessary for linearizing the removal (as will be further discussed in Section~\ref{sec-range-queries}), a remover must physically remove the deleted node before it returns from a  \emph{remove()} call.
The marking of a node in line~\ref{alg-remove-mark} only determines the remover's identity and announces its intention to delete the marked node. Therefore, the remover must additionally ensure that the node is indeed unlinked, by calling the \emph{find()} method in line~\ref{alg-remove-trim}.
In Appendix~\ref{sec-correctness} we formally prove that the list implementation, presented in Algorithm~\ref{pseudo-list}, is linearizable and lock-free.


\subsection{Adding Range Queries} \label{sec-range-queries}

Given Algorithm~\ref{pseudo-list}, adding a linearizable range queries mechanism is relatively straight forward. We use a method which is similar to the vCAS technique~\cite{wei2021constant}.
As discussed in Section~\ref{sec-intro}, the vCAS scheme introduces an extra level of indirection for the linked-list per node access. Indeed, we show in Section~\ref{sec-evaluation} that the vCAS implementation suffers from high overheads.
The original vCAS paper provides a technique for avoiding this level of indirection. 
The suggested optimization relies on the following (very specific) assumption: 
a certain node can be the third input parameter to a successful CAS operation only once throughout the entire execution.
That successful CAS is considered as the \textit{recording} of this node, and the property is referred to as \textit{recorded-once} in~\cite{wei2021constant}.

Although the recorded-once property yields a linearizable solution, which reduces memory and time overheads, 
this assumption does not hold in the presence of physical deletions, as they usually set the deleted node's predecessor to point to the deleted node's successor~\cite{harris2001pragmatic}, or to another, already reachable node~\cite{natarajan2014fast,brown2014general}, and then, this reachable node is recorded more than once.
This makes the suggested technique inapplicable to Harris's linked-list~\cite{harris2001pragmatic} and most other concurrent data-structures (e.g.,~\cite{natarajan2014fast,bronson2010practical,herlihy2020art}). The original vCAS paper implemented a recorded-once binary search tree, based on~\cite{ellen2010non}, which we compare against in Section~\ref{sec-evaluation}.
We extend the recorded-once condition and make it fit for the linked-list and other data-structures. We claim that associating each node with the data of a single data-structure update (as provided by our list) is enough for avoiding indirection in this setting.
Given such an association, there is no need to store update-related data in designated version records, since it can be stored directly inside nodes, in a well-defined manner. 

First, in a similar way to~\cite{wei2021constant,arbel2018harnessing,nelson2021bundled}, we add a shared clock, for associating each node with a timestamp. The shared clock is read and updated using the \emph{getTS()} and \emph{fetchAddTS()} methods, respectively (see Appendix~\ref{sec-auxiliary}).
The shared clock is incremented whenever a range query is executed (e.g., see line~\ref{alg-range-query-faa} in Algorithm~\ref{pseudo-range}), and is read before setting a new node's timestamp (e.g, see lines~\ref{alg-insert-update-ts},~\ref{alg-find-update-pred-ts},~\ref{alg-find-update-curr-ts},~\ref{alg-trim-update-curr-ts},~\ref{alg-trim-update-succ-ts} and~\ref{alg-trim-update-new-curr-ts} in Algorithm~\ref{pseudo-list}).
Next, we change the nodes layout (see our node class description in Algorithm~\ref{pseudo-list}).
On top of the standard fields (i.e., key, value and \textit{next} pointer), we add two extra fields to each node. The first field is the node's timestamp (denoted as \textit{ts}), representing its insertion into the list. 
Nodes' \textit{ts} fields are always initialized with a special $\bot$ value (see lines~\ref{alg-insert-alloc} and~\ref{alg-trim-alloc} in Algorithm~\ref{pseudo-list}), to be given an actual timestamp after being inserted into the list.
The second field, \textit{prior}, points to the previous successor of this node's first predecessor in the list (its predecessor when being inserted into the list). Both fields are set once and then remain immutable. E.g., consider the new node, inserted into the list at stage 5 in Figure~\ref{fig-physical-delete}. Its \textit{prior} field points to the node whose key is 23, as this is the former successor of the node whose key is 9, which is the first predecessor of the newly inserted node.
By its specification, once the \textit{prior} field is set (see line~\ref{alg-insert-update-prior} and~\ref{alg-trim-update-prior} in Algorithm~\ref{pseudo-list}), it is immutable.
These two new fields are not used during the list operations from Algorithm~\ref{pseudo-list}, but we do specify their proper initialization, in order to support linearizable range queries.
Moreover (and in a similar way to~\cite{wei2021constant,arbel2018harnessing}), list inserts and deletes are linearized during the execution of the \emph{getTS()} method, as follows.
Let \emph{newNode} be the node inserted into the list in line~\ref{alg-insert-update-pred}, during a successful \emph{insert()} operation. \emph{newNode}'s timestamp is set at some point, not later than the CAS in line~\ref{alg-insert-update-ts} (it may be updated earlier, by a different thread). The \emph{getTS()} invocation that precedes the successful update of \emph{newNode}'s timestamp is the operation's linearization point.
In a similar way, consider a successful \emph{remove()} operation. The removed node is unlinked from the list during a successful \emph{trim()} execution. Let \emph{newCurr} be the node successfully inserted into the list in line~\ref{alg-trim-cas}, during this successful \emph{trim()} execution. \emph{newCurr}'s timestamp is set at some point, not later than the CAS in line~\ref{alg-trim-update-new-curr-ts} (it may be updated earlier, by a different thread). The \emph{getTS()} invocation that precedes the successful update of \emph{newCurr}'s timestamp is the operation's linearization point.

\begin{algorithm*}[!ht]
\begin{multicols}{2}
\begin{algorithmic}[1]

\Procedure{rangeQuery}{low, high, *arr} \label{alg-range-query-call}
    \State ts := fetchAddTS() \label{alg-range-query-faa}
    \State currKey := low
    \While{(true)} {\color{red} \label{alg-range-query-while-find} \Comment{finding a starting point}}
        \State pred, curr $\leftarrow$ FIND(currKey) \label{alg-range-query-find}
        \State currKey := pred $\rightarrow$ key
        \While{(pred $\rightarrow$ ts > ts)} \label{alg-range-query-while-pred}
            \State pred := pred $\rightarrow$ prior \label{alg-range-query-pred-gets-prior}
        \EndWhile
        \If{(pred $\rightarrow$ key $\leq$ low)}
            \State curr := pred \label{alg-range-query-curr-gets-pred}
            \State \textbf{break} \label{alg-range-query-break}
        \EndIf
        \State ts := getTS() - 1 \label{alg-range-query-new-ts}
    \EndWhile
    \While{(curr $\rightarrow$ key < low)} \label{alg-range-query-while-curr-smaller} 
        \State succ := \textbf{getRef}(curr $\rightarrow$ next) \label{alg-range-query-succ-gets-curr-next-smaller}
        \State CAS(\&succ $\rightarrow$ ts, $\bot$, getTS()) \label{alg-range-query-smaller-curr-succ-ts-update}
        \While{(succ $\rightarrow$ ts > ts)} \label{alg-range-query-smaller-curr-while-succ}
            \State succ := succ $\rightarrow$ prior \label{alg-range-query-smaller-curr-succ-gets-prior}
        \EndWhile
        \State curr := succ \label{alg-range-query-smaller-curr-gets-succ}
    \EndWhile
    \State count := 0 \label{alg-range-query-init-count}
    \While{(curr $\rightarrow$ key $\leq$ high)} \label{alg-range-query-while-curr-bigger} 
        \State arr[count] $\rightarrow$ key := curr $\rightarrow$ key \label{alg-range-query-write-key}
        \State arr[count] $\rightarrow$ value := curr $\rightarrow$ value \label{alg-range-query-write-value}
        \State count := count + 1
        \State succ := \textbf{getRef}(curr $\rightarrow$ next) 
        \label{alg-range-query-succ-gets-curr-next-bigger}
        \State CAS(\&succ $\rightarrow$ ts, $\bot$, getTS()) \label{alg-range-query-bigger-curr-succ-ts-update}
        \While{(succ $\rightarrow$ ts > ts)} \label{alg-range-query-bigger-curr-while-succ}
            \State succ := succ $\rightarrow$ prior \label{alg-range-query-bigger-curr-succ-gets-prior}
        \EndWhile
        \State curr := succ \label{alg-range-query-bigger-curr-gets-succ}
    \EndWhile
    \State \textbf{return} count \label{alg-range-query-return-count}
\EndProcedure

\end{algorithmic}
\end{multicols}
\caption{The Range Queries Mechanism.}
\label{pseudo-range}
\end{algorithm*}

Our range queries mechanism is presented in Algorithm~\ref{pseudo-range}.
The \emph{rangeQuery()} operation receives three input parameters (see line~\ref{alg-range-query-call}): the lowest and highest keys in the range, and an output array for returning the actual keys and associated values in the range. In addition to filling this array, it also returns its accumulated size in the \emph{count} variable.
The operation starts by fetching and incrementing the global timestamp counter (line~\ref{alg-range-query-faa}), which serves as the range query's linearization point. I.e., the former timestamp is the one associated with the range query. This way, the range query is indeed linearized between its invocation and response, along with guaranteeing that the respective view is immutable during the operation (as new updates will be associated with the new timestamp). 

After incrementing the global timestamp counter, the operation uses the \emph{find()} auxiliary method in order to locate the first node in range (lines~\ref{alg-range-query-while-find}-\ref{alg-range-query-new-ts}).
As opposed to the vCAS mechanism~\cite{wei2021constant}, and in a similar way to the Bundles mechanism~\cite{nelson2021bundled}, we observe that until the traversal reaches the target range, there is no need to take timestamps into consideration. This observation is crucial for performance, as there is no need to traverse nodes via the \emph{prior} fields (which produce longer traversals in practice).
In addition, it enables using the fast index (described in Section~\ref{sec-index}) for enhancing the search.
During each loop iteration, we first find a node with a key which is smaller than the lowest key in the range (saved as the \emph{pred} variable in line~\ref{alg-range-query-find}).
Then, we optimistically try to find a relatively close node, following \emph{prior} pointers, until we get to a small enough timestamp (lines~\ref{alg-range-query-while-pred}-\ref{alg-range-query-pred-gets-prior}).
Since this search may result in a node with a bigger key (e.g., see line~\ref{alg-insert-update-prior} in Algorithm~\ref{pseudo-list}), the next iteration sends a smaller key as input to the \emph{find()} execution in line~\ref{alg-range-query-find}.
Note that in the worst case scenario, the loop in lines~\ref{alg-range-query-while-find}-\ref{alg-range-query-new-ts} stops after the  \emph{find()} execution in line~\ref{alg-range-query-find} outputs the \emph{head} sentinel node (as its timestamp is necessarily smaller than ts). Therefore, it never runs infinitely.
The purpose of updating $ts$ in line~\ref{alg-range-query-new-ts} will be clarified in Section~\ref{sec-mvcc-vbr}, as it is related to the VBR mechanism. Note that in any case, this update does not foil correctness, since it is still guaranteed that the range query is linearized between the operation's invocation and response.

When \emph{pred} has a key which is smaller than the range lower bound, the operation moves on to the next step (the loop breaks in line~\ref{alg-range-query-break}). At this point, the traversal continues according to the respective timestamp\footnote{In Appendix~\ref{sec-correctness}
we prove that a node's successor at timestamp $T$ can be found by starting from its current successor and then following \textit{prior} references until reaching a node with a timestamp which is not greater than $T$.}, until getting to a node with a key which is at least the range lower bound (lines~\ref{alg-range-query-while-curr-smaller}-\ref{alg-range-query-smaller-curr-gets-succ}).
Once a node with a big enough key is found, the traversal continues in lines~\ref{alg-range-query-while-curr-bigger}-\ref{alg-range-query-bigger-curr-gets-succ}. At this stage, the \emph{count} output variable and the output array are updated according to the data accumulated during the range traversal. Finally, the \emph{count} output variable, indicating the total number of keys in range, is returned in line~\ref{alg-range-query-return-count}.
Note that throughout the traversals in lines~\ref{alg-range-query-smaller-curr-while-succ}-\ref{alg-range-query-smaller-curr-succ-gets-prior} and lines~\ref{alg-range-query-bigger-curr-while-succ}-\ref{alg-range-query-bigger-curr-succ-gets-prior}, there is no need to update \emph{succ}'s timestamp (as done in lines~\ref{alg-range-query-smaller-curr-succ-ts-update} and~\ref{alg-range-query-bigger-curr-succ-ts-update}), since it serves as a node's \emph{prior} and thus, is guaranteed to already have an updated timestamp 
(see Appendix~\ref{sec-correctness}).

\subsection{Adding A Safe Memory Reclamation Mechanism} \label{sec-mvcc-vbr}

Before integrating our list with a manual memory reclamation mechanism, we must first install \textit{retire()} invocations, for announcing that a node's memory space is available for re-allocation.
Naturally, nodes are retired after unsuccessful insertions, or after they are unlinked from the list. I.e., \emph{newNode} is retired if the CAS in line~\ref{alg-insert-update-pred} is unsuccessful, \emph{newCurr} is retired if the CAS in line~\ref{alg-trim-cas} is unsuccessful, and upon a successful trimming in line~\ref{alg-trim-cas}, the unlinked nodes are retired, starting from \emph{victim}. The last retired node is \emph{curr}, which is replaced by its new representative in the list, \emph{newCurr}.
Note that we do not handle physical removals of \emph{prior} links. Handling them is unnecessary, and might cause significant overheads, both to the list operations and to the reclamation procedure.
Therefore, retired nodes are still reachable from the list head during retirement: \emph{newCurr}'s \emph{prior} field points to \emph{victim}, making all of the unlinked nodes reachable via this pointer (and their \emph{next} pointers).

To add a safe memory reclamation mechanism to our list, we use an improved variant of Version Based Reclamation (VBR)~\cite{DBLP:journals/corr/abs-2107-13843}.
VBR cannot be integrated as is. Similarly to most safe memory reclamation techniques, it assumes that retired objects are not reachable via the data-structure links. This assumption is crucial to the correctness of VBR, as retired objects may be immediately reclaimed. 
In addition, VBR uses a slow ticking epoch clock, and ensures that the clock ticks at least once between the retirement and future re-allocation of the same node. 
During execution, the operating threads constantly check that the global epoch clock has not changed. Upon a clock tick, they conservatively treat all data read from shared memory as stale, and move control to an adequate previous point in the code in order to read a fresh value. 
As long as the clock does not tick, threads may continue executing without worrying about use-after-free issues. 
The intuition is that if a node is accessed during a certain epoch, then it must have been reachable during this epoch. I.e., even if this node has already been retired, its retirement was during the current epoch, which means that it has not been re-allocated yet (as the clock has not ticked yet).

Our list implementation poses a new challenge in this context. Suppose that the current epoch is $E$, and that a certain node, $n$, is currently in the list (i.e., it has not been unlinked using the \emph{trim()} method yet). In addition, suppose that $n$'s \emph{prior} field points to another node, $m$, that has been retired during an earlier epoch. Then $m$ may be reclaimed and re-allocated during $E$. A traversing thread may access $n$'s \emph{prior} field during $E$, without getting any indication to the fact that the referenced node is a stale value.
Another problem, which does not affect correctness, but may cause frequent thread starvations, is that the global epoch clock is likely to tick during a long range query. In the original VBR scheme, a clock tick forces the executing thread to start its traversal from scratch, even if it has not encountered any reclaimed node in practice.

In order to overcome the above problems, we made some small adjustments to the original VBR scheme.
First, we kept the global epoch clock of VBR and the timestamp clock of the range queries separated. We separated the two, as VBR works best with a (very) slow ticking clock. Read-intensive workloads (in which range queries dominate the execution) incur high overheads when combining the two clocks. The separation of the two independent clocks helps overcome the potential aborts.
The second step was to modify the nodes' layout
(The VBR-integrated node layout appears in Appendix~\ref{sec-vbr-correctness}).
Recall that the VBR scheme adds a \textit{birth epoch} to every node, along with a version per mutable field. Non-pointer mutable fields are associated with the node's birth epoch, and pointers are associated with a version which is the maximum between the birth epoch of the node and the birth epoch of its successor.
Our list nodes have two mutable fields, their timestamp $ts$ (changes only once), and their \emph{next} pointer. The \emph{prior} pointers are immutable.
Accordingly, we associated the node's timestamp with its VBR-integrated birth epoch, serving as its version (there was no need to add an extra $ts$ version), and added a designated \emph{next} pointer version. 
Writes to the mutable fields are handled exactly as in the original VBR scheme. Accordingly, upon allocation, a node's timestamp is initialized to $\bot$, along with the current VBR epoch as its associated birth epoch. When the timestamp is updated (see line~\ref{alg-insert-update-ts},~\ref{alg-find-update-pred-ts},~\ref{alg-find-update-curr-ts},~\ref{alg-trim-update-curr-ts},~\ref{alg-trim-update-succ-ts} and~\ref{alg-trim-update-new-curr-ts} in Algorithm~\ref{pseudo-list}, or lines~\ref{alg-range-query-smaller-curr-succ-ts-update} and~\ref{alg-range-query-bigger-curr-succ-ts-update} in Algorithm~\ref{pseudo-range}), the birth epoch (also serving as the timestamp's version) does not change, as the two fields are accessed together, via a wide-compare-and-swap (WCAS) instruction.
Similarly, \emph{next} pointers are associated with the maximum between the two respective birth epochs, and are also updated using WCAS.

Reads are handled in a different manner from the original VBR, as the problems we mentioned above must be treated with special care. The original VBR repeatedly reads the global epoch in order to make sure that it has not changed. In our extended VBR variant, it is read once. After reading the global epoch, and as long as the executing thread does not encounter a birth epoch or a version which is bigger then this epoch, it may continue executing its code. The motivation behind this behavior is that even if a certain node in the system has meanwhile been reclaimed, this node does not pose a problem as long as the current thread does not encounter it. 
Therefore, traversing threads follow three guidelines: (1) a node's birth epoch is read again after each read of another field, (2) after dereferencing a \emph{next} pointer, the reader additionally makes sure that the successor's birth epoch is not greater than the pointer's version, and (3) after dereferencing a \emph{prior} pointer (which is not associated with a version), the reader additionally makes sure that the successor's birth epoch is not greater than the predecessor's birth epoch. 
If any of these conditions does not hold, then the reader needs to proceed according to the original VBR's protocol.
In Appendix~\ref{sec-vbr-correctness} we prove that these three guidelines are sufficient for maintaining correctness.
%
Upon an epoch change, the original VBR enforces a rollback to a predefined checkpoint in the code. Accordingly, we install code checkpoints. Whenever a check that our guidelines impose fails, the executing thread rolls-back to the respective checkpoint. Checkpoints are installed in the beginning of each API operation (i.e., \emph{insert()}, \emph{remove()}, \emph{contains()} and \emph{rangeQuery()}). Another checkpoint is installed after a successful marking in line~\ref{alg-remove-mark} of Algorithm~\ref{pseudo-list}, as the identity of the marking thread affects linearizability (and therefore, a rollback to the beginning of the operation would foil linearizability). Note that a successful insertion in line~\ref{alg-insert-update-pred} does not force a checkpoint (although it affects linearizability, by setting the inserter identity), as it is not followed by any reads of potentially reclaimed memory.

Another issue that needs to be dealt with is the guarantee that life-cycles of nodes, allocated from the same memory address, do not overlap. The original VBR scheme does so by associating each node with a retire epoch. A node's retire epoch is set upon retirement. During re-allocation, if the current global epoch is equal to the node's retire epoch, then the global epoch is incremented before re-allocation.
We chose to optimize over the original VBR, discarding the retire epoch field, as it adds an extra field per allocated node. Instead, each retire list is associated with the epoch, recorded once it is full (right before it is returned to the global pool of nodes). Upon pulling such a list from the global pool, if its associated epoch is equal to the current one, then the global epoch is incremented.

Finally, consider the following scenario. Suppose that a thread $T_1$ is running a range query, the current epoch is $E$ and the current global timestamp is $t$. Next, suppose that another thread, $T_2$, reclaims a node $n$ that has been retired during $E$, and that is relevant for $T_1$'s range query. Starting from this point, whenever $T_1$ accesses the newly allocated node, it rolls back and starts its traversal from scratch (as the new node has a birth epoch which is greater then its predecessor through the \emph{prior} pointer, foiling guideline 3).
As long as $T_1$'s $ts$ variable does not change, $T_1$ will infinitely get to the new allocated node and then roll-back to the beginning. 
We reduce the probability of such scenarios in practice, by updating the $ts$ variable in line~\ref{alg-range-query-new-ts} of Algorithm~\ref{pseudo-range}. We further ensure that that the current global timestamp is up-to-date by incrementing it upon each re-allocation, if necessary.

\subsection{Adding A Fast Index} \label{sec-index}

Our linked-list implementation encapsulates the key-value pairs and enables the timestamps mechanism. However, when key ranges are large, the linked-list does not perform as good as other concurrent data-structures. It forces a linear traversal per operation, as opposed to skip lists~\cite{fraser2004practical,herlihy2020art} and binary search trees~\cite{natarajan2014fast,brown2014general}. We observe that the index links in such data-structures (e.g., the links connecting the upper levels in a skip list or the inner levels in a tree) are only required for fast access. The actual data exists only in the lowest level of the skip list (or the leaves of an external tree).
Therefore, we allow a simple integration of an external index, enabling fast access instead of long traversals. The index should provide an \emph{insert(key, node)} operation, receiving a key and a node pointer as its associated value. It additionally should provide a \emph{remove(key)} operation. Finally, instead of providing a \emph{contains(key)} operation, it should provide a \emph{findPred(key)} operation, receiving a key and returning a pointer to the node associated with some key which is smaller than the given one. 

The \emph{findPred(key)} can be naively implemented by calling the data-structure search method (there usually exists such method. E.g.,~\cite{harris2001pragmatic,natarajan2014fast,herlihy2020art}) with a smaller key as input. I.e., it is possible to search for key minus 2 or minus 10.
Obviously, this does not guarantee that a suitable node will indeed be returned. However, if the selected smaller key is not small enough, it is possible to start a new trial and search for a smaller key. Our experiments showed that limiting the number of such trials per search to a small constant (e.g., 5 in our experiments) is negligible in terms of performance, and is usually enough for locating a relevant node. In addition, the index is used only for fast access, so correctness is not affected even if all trials fail.
Specifically, the \emph{findPred(key)} operation can be easily implemented for a skip list, using its built-in search auxiliary method~\cite{fraser2004practical,herlihy2020art} (as it returns a predecessor with a smaller key). Examples for using a skip list as a fast index have already been introduced for linearizable data-structures~\cite{sheffi2018scalable} and in the transactional memory setting~\cite{spiegelman2016transactional}.
The \emph{findPred(key)} operation can also be implemented for some binary search trees, by traversing the left child, instead of the right child, at some point during the search path. We applied this method when implementing our tree index, based on Natarajan and Mittal's BST~\cite{natarajan2014fast}.

\begin{algorithm*}[!ht]
\begin{multicols}{2}
\begin{algorithmic}[1]

\State currKey := key \label{alg-index-init-key}
\State attempts := MAX\_ATTEMPTS \label{alg-index-max-attempts}
\While{($--$attempts $\neq 0$)} \label{alg-index-while}
    \State pred := index $\rightarrow$ \textbf{findPred}(currKey) \label{alg-index-find-pred}
    \State predTS := pred $\rightarrow$ timestamp \label{alg-index-pred-ts}
    \State predNext := pred $\rightarrow$ next \label{alg-index-pred-next}
    \State predKey := pred $\rightarrow$ key \label{alg-index-pred-key}
    \If{(pred $\rightarrow$ birth > currEpoch)} \textbf{rollback} \label{alg-index-rollback} 
    \ElsIf{(predKey $\geq$ key $\vee$ predTS == $\bot$)} \label{alg-index-continue} 
        \State \textbf{continue}
    \ElsIf{(\textbf{isMarkdOrFlagged}(predNext))} \label{alg-index-if-marked}
        \State currKey := predKey \label{alg-index-curr-gets-pred}
    \Else{\textbf{ break}} \label{alg-index-break} 
    \EndIf
\EndWhile
\If{(attempts == 0)} \label{alg-index-if-attempts} 
    \State pred := head \label{alg-index-pred-head} 
    \State predNext := pred $\rightarrow$ next \label{alg-index-head-next} 
\EndIf 

\end{algorithmic}
\end{multicols}
\caption{Starting a new traversal using the index.}
\label{pseudo-index}
\end{algorithm*}

We update the fast index as follows: New Nodes are inserted into the index after being inserted into the list (i.e., right before the \emph{insert()} operation returns in line~\ref{alg-insert-return-no-val} of Algorithm~\ref{pseudo-list}). Nodes are removed from the index after being removed from the list, and right before being retired (see Section~\ref{sec-mvcc-vbr}). Note that the \emph{curr} node, replaced by \emph{newCurr} via the CAS in line~\ref{alg-trim-cas}, should be removed from the index, followed by an insertion of  \emph{newCurr}. 
In our implemented index, we have implemented an \emph{update()} operation instead. This operation receives as input a node reference, and uses it to replace a node with the same key in the index\footnote{The implemented \emph{update()} operation also takes the node's birth epoch into account, and does not replace a node with a reclaimed node or with a node with a smaller birth epoch}. 
In case the external index does not provide an \emph{update()} operation, this also may be executed via the standard  \emph{remove()} operation, followed by a respective \emph{insert()} operation.
The index is read only during the \emph{find()} auxiliary method. Instead of starting each list traversal from the \emph{head} sentinel node (see line~\ref{alg-find-head} in Algorithm~\ref{pseudo-list}), the traversing thread tries to shorten the traversal by accessing the fast index. 
I.e., instead of executing the code in lines~\ref{alg-find-head}-\ref{alg-find-head-next} of Algorithm~\ref{pseudo-list}, each thread executes the code from Algorithm~\ref{pseudo-index}.
It starts by initializing the searched key to the input key (received as input in line~\ref{alg-find-call} of Algorithm~\ref{pseudo-list}).
Then, after finding the alleged predecessor, using the \emph{findPred()} operation (line~\ref{alg-index-find-pred}), if its birth epoch is bigger than the last recorded one, the executing thread rolls-back to its last recorded checkpoint (see Section~\ref{sec-mvcc-vbr}).
Otherwise, if \emph{pred}'s key is not smaller than the given input key, or its timestamp is not initialized yet 
(this is a possible scenario, as we show in Appendix~\ref{sec-index-correctness}), 
the thread starts another trial, with the same key. 
Otherwise, if \emph{pred} is either marked or flagged (line~\ref{alg-index-if-marked}), the thread starts another trial, with a smaller key. 
Otherwise, it is guaranteed that \emph{pred} and \emph{predNext} hold a valid node and its (unmarked and unflagged) \emph{next} pointer, respectively, and the thread may start its list traversal from line~\ref{alg-find-curr-gets-head-next} of Algorithm~\ref{pseudo-list}.
Note that the code presented in Algorithm~\ref{pseudo-index} always terminates, as in the worst case scenario, the loop breaks after a predefined number of attempts.
A full correctness proof appears in Appendix~\ref{sec-index-correctness}.

\section{Evaluation} \label{sec-evaluation}

For evaluating throughput of EEMARQ, we implemented the linked-list presented in Section~\ref{sec-algoritm}\footnote{For avoiding unnecessary accesses to the global timestamps clock, the $ts$ field updates from Algoritm~\ref{pseudo-list} and Algorithm~\ref{pseudo-range} were executed only after the $ts$ field was read, and only if it was still equal to $\bot$.}, including the extended VBR variant, as described in Section~\ref{sec-mvcc-vbr}. 
In addition, we implemented the lock-free skip list from~\cite{fraser2004practical} and the lock-free BST from~\cite{natarajan2014fast}. Both the skip list and the tree were used on top of the linked list, and served as fast indexes, as described in Section~\ref{sec-index}. Deleted nodes from both indexes were manually reclaimed, using the original VBR scheme, according to the integration guidelines from~\cite{DBLP:journals/corr/abs-2107-13843} (without the adjustments described in Section~\ref{sec-mvcc-vbr}).
Each data-structure had its own objects pool (as VBR forces type preservation).
Retire lists had 64 entries. Since VBR allows the immediate reclamation of retired nodes, retire lists were reclaimed as a whole every time they contained 64 nodes.
I.e., at most 8192 (64 retired nodes X 128 threads) objects were over-provisioned per data-structure at any given moment.

\begin{figure*}[!ht]
\centering
      \begin{subfigure}{0.63\textwidth}

         \includegraphics[width=\textwidth]{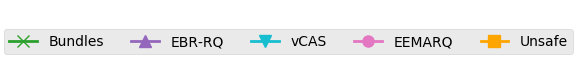}
      \end{subfigure}

      \begin{subfigure}{0.23\textwidth}
         \includegraphics[width=\textwidth]{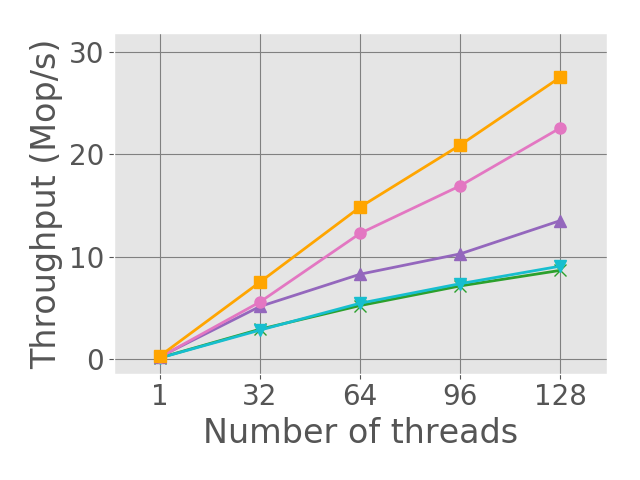}
        \caption{SL. Lookup-heavy: 0\%U-90\%C-10\%RQ}
        \label{fig:skiplist-0-10}
      \end{subfigure}
\hfill    
     \begin{subfigure}{0.23\textwidth}
    \includegraphics[width=\textwidth]{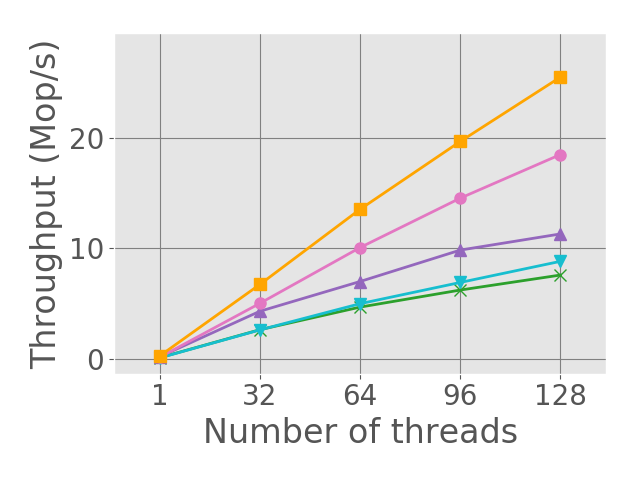}
	\caption{SL. Mixed Workload: 50\%U-40\%C-10\%RQ}
	        \label{fig:skiplist-25-10}

    \end{subfigure}
 \hfill
     \begin{subfigure}{0.23\textwidth}
    \includegraphics[width=\textwidth]{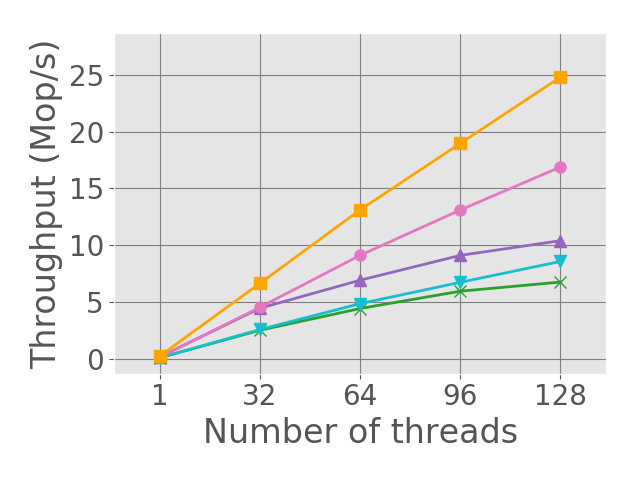}
	\caption{SL. Update-heavy: 90\%U-0\%C-10\%RQ}
	        \label{fig:skiplist-45-10}

    \end{subfigure}
 \hfill
    \begin{subfigure}{0.23\textwidth}
   \includegraphics[width=\textwidth]{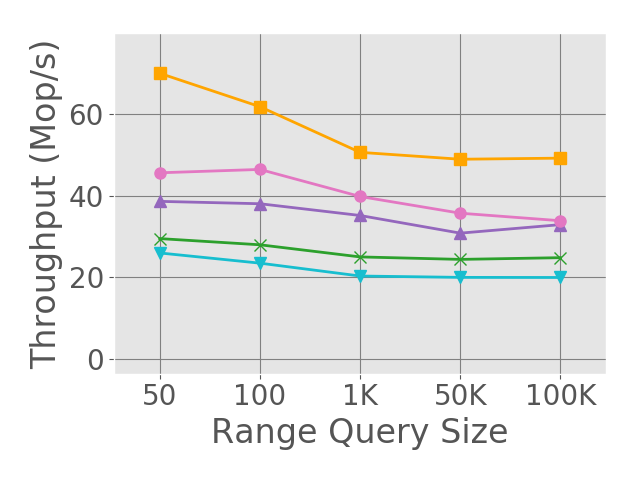}
	\caption{SL. 64 threads:50\%U-50\%C, 64 RQ threads}
	        \label{fig:skiplist-25-0}

    \end{subfigure}
    
      \begin{subfigure}{0.23\textwidth}
         \includegraphics[width=\textwidth]{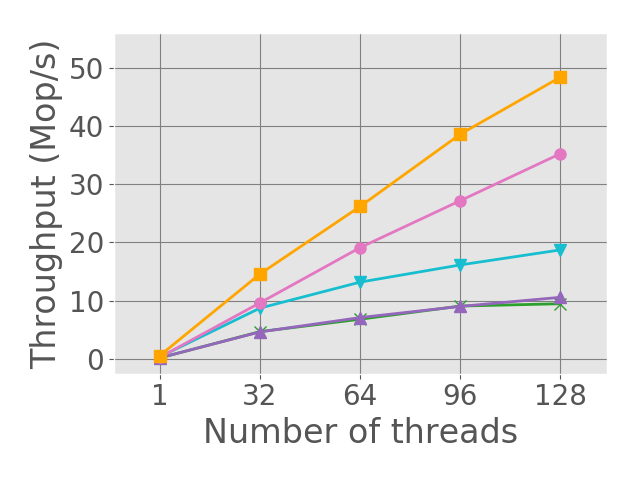}
        \caption{BST. Lookup-heavy: \\ 0\%U-90\%C-10\%RQ}
        \label{fig:tree-0-10}
      \end{subfigure}
\hfill  
     \begin{subfigure}{0.23\textwidth}
    \includegraphics[width=\textwidth]{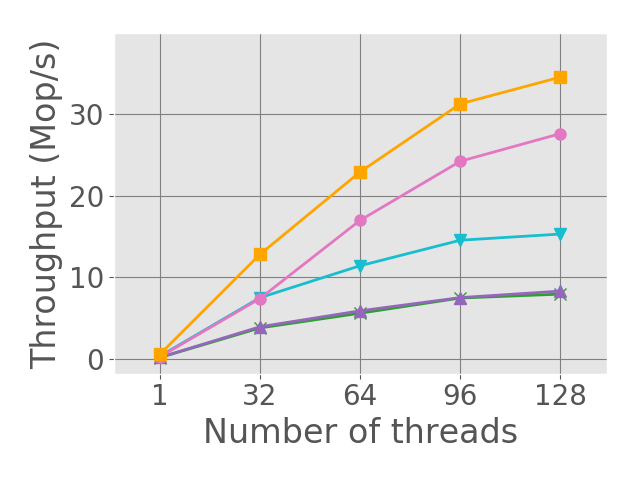}
	\caption{BST. Mixed workload: \\ 50\%U-40\%C-10\%RQ}
	        \label{fig:tree-25-10}

    \end{subfigure}
 \hfill
     \begin{subfigure}{0.23\textwidth}
    \includegraphics[width=\textwidth]{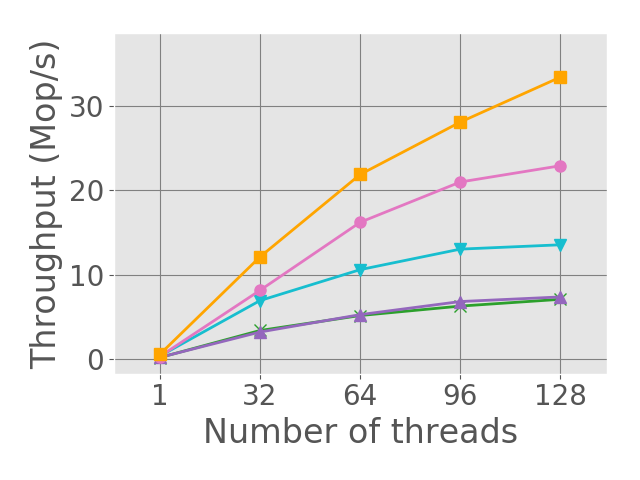}
	\caption{BST. Update-heavy: \\ 90\%U-0\%C-10\%RQ}
	        \label{fig:tree-45-10}

    \end{subfigure}
 \hfill
     \begin{subfigure}{0.23\textwidth}
    \includegraphics[width=\textwidth]{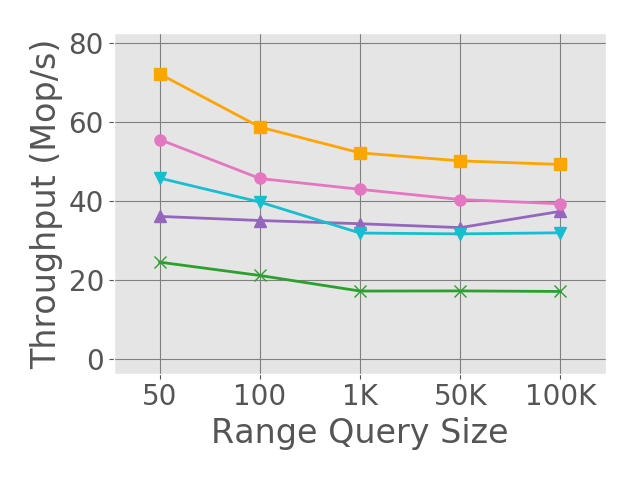}
	\caption{BST. 64 threads:50\%U-50\%C, 64 RQ threads}
	        \label{fig:tree-25-0}

   \end{subfigure}

\caption{Throughput evaluation under various workloads for the skip list (\ref{fig:skiplist-0-10}-\ref{fig:skiplist-25-0}) and the tree (\ref{fig:tree-0-10}-\ref{fig:tree-25-0}). The key range is 1M. In Figures~\ref{fig:skiplist-0-10}-\ref{fig:skiplist-45-10} and~\ref{fig:tree-0-10}-\ref{fig:tree-45-10}, the range query size is 1000. Y axis: throughput in million operations per second. X axis: \#threads in Figures~\ref{fig:skiplist-0-10}-\ref{fig:skiplist-45-10} and~\ref{fig:tree-0-10}-\ref{fig:tree-45-10}, and range query size in Figures~\ref{fig:skiplist-25-0} and~\ref{fig:tree-25-0}.} \label{fig:throughput}
\end{figure*}

We compared EEMARQ against four competitors, all using epoch-based reclamation. EBR-RQ is the lock-free epoch-based range queries technique by Arbel-Raviv and Brown~\cite{arbel2018harnessing}, vCAS is the lock-free technique by Wei et al.~\cite{wei2021constant}, Bundles is the lock-based bundled references technique by Nelson et al.~\cite{nelson2021bundled}, and Unsafe uses a naive non-linearizable scan of the nodes in the range without synchronizing with concurrent updates (used as our baseline). 
We did not compare EEMARQ against RLU~\cite{matveev2015read}, as its mechanism is not linearizable, and it was also shown to be slower than our competitors~\cite{nelson2021bundled}.
For EBR-RQ and Bundles, we used the implementation provided by the authors. The Unsafe code was provided by the EBR-RQ authors. 
The vCAS authors provided a vCAS-based lock-free BST, including their optimization for avoiding indirection (see Section~\ref{sec-range-queries} for more details). Since there does not exist any respective skip list implementation, we implemented a vCAS-based skip list according to the guidelines from~\cite{wei2021constant}. The vCAS-based skip list was not optimized, as the optimization technique, suggested in~\cite{wei2021constant}, does not fit to this data-structure.
For our competitors' memory reclamation, we used the original implementations, provided by~\cite{arbel2018harnessing,wei2021constant,nelson2021bundled}, without any code or object pools usage changes.
In particular, memory was not returned to the operating system. I.e., all implementations used pre-allocated object pools~\cite{singh2021nbr,DBLP:journals/corr/abs-2107-13843} for reclaiming memory.

\subparagraph{Setup}
We conducted our experiments on a machine running Linux (Ubuntu 20.04.4), equipped with 2 Intel Xeon Gold 6338 2.0GHz processors. Each processor had 32 cores, each capable of running 2 hyper-threads to a total of 128 threads overall.
The machine used 256GB RAM, an L1 data cache of 3MB and an L1 instruction cache of 2MB, an L2 unified cache of 80MB, and an L3 unified cache of 96MB.
The code was written in C++ and compiled using the GCC compiler version 9.4.0 with -std=c++11 -O3 -mcx16.
Each test was a fixed-time micro benchmark in which threads randomly call the \emph{insert()}, \emph{remove()}, \emph{contains()} and \emph{rangeQuery()} operations according to different workload profiles.
We ran the experiments with a range of 1M keys.
Each execution started by pre-filling the data-structure to half of its range size, and lasted 10 seconds (longer experiments showed similar results).
Each experiment was executed 10 times, and the average throughput across all executions is reported.
Figure~\ref{fig:throughput} shows the skip list and tree scalability under various workloads. 
The updates (half \emph{insert()} and half \emph{remove()}), \emph{contains()} and \emph{rangeQuery()} percentiles appear under each graph.
Figures~\ref{fig:skiplist-0-10}-\ref{fig:skiplist-45-10} and~\ref{fig:tree-0-10}-\ref{fig:tree-45-10} show the skip list and tree scalability as a function of the number of executing threads, under a variety of workloads. All queries have a fixed range of 1K keys (following~\cite{wei2021constant}).
Figures~\ref{fig:skiplist-25-0} and~\ref{fig:tree-25-0} show the effect of varying range query size on the skip list and tree performance. In these experiments, 64 threads perform 25\% \emph{insert()}, 25\% \emph{remove()} and 50\% \emph{contains()}, and 64 threads perform range queries only.
In Appendix~\ref{sec-additional_graphs} we present the respective range queries and updates throughput for Figures~\ref{fig:skiplist-25-0} and~\ref{fig:tree-25-0}, along with additional results for other workloads.

\subparagraph{Discussion}

The EEMARQ skip list surpasses the next best algorithm, EBR-RQ, by up to 65\% in the lookup-heavy workload (Figure~\ref{fig:skiplist-0-10}), by up to 50\% in the mixed workload (Figure~\ref{fig:skiplist-25-10}), and by up to 70\% in the update-heavy workload (Figure~\ref{fig:skiplist-45-10}). 
The EEMARQ tree surpasses its next best algorithm, vCAS, by up to 75\% in the lookup-heavy workload (Figure~\ref{fig:tree-0-10}), by up to 65\% in the mixed workload (Figure~\ref{fig:tree-25-10}), and by up to 70\% in the update-heavy workload (Figure~\ref{fig:tree-45-10}). 

The results show that avoiding indirection is crucial to performance.
In particular, EEMARQ outperforms its competitors in the lookup-heavy workloads (Figures~\ref{fig:skiplist-0-10} and~\ref{fig:tree-0-10}), in which memory is never reclaimed. I.e., its range query mechanism, which completely avoids traversing separate version nodes, has a significant advantage under such workloads.
It can also be seen when examining EEMARQ's competitors. 
While the vCAS-based tree (which avoids indirection) is EEMARQ's next best competitor, the vCAS-based skip list (that involves the traversal of designated version nodes) is the weakest among the skip list implementations. In addition, the Bundles technique, which employs such a level of indirection when range queries are executed, also performs worse than most competitors, under most workloads. 

Under update-dominated workloads, EEMARQ is faster also thanks to its efficient memory reclamation method. While all other algorithms use EBR as their memory reclamation scheme, EEMARQ enjoys VBR's inherent locality and fast reclamation process.
Moreover, EEMARQ avoids the original VBR's frequent accesses to the global epoch clock, as described in Section~\ref{sec-mvcc-vbr}.
Although using VBR forces rollbacks when accessing reclaimed nodes, our fast index mechanism allows a fast retry, which makes the impact of rollbacks small. Indeed, experiments with longer retire lists (i.e., fewer rollbacks) showed similar results.
This is clearly shown in Figures~\ref{fig:skiplist-25-0} and~\ref{fig:tree-25-0}: EEMARQ outperforms its competitors when the query ranges are big (10\% of the data-structure range). I.e., possible frequent rollbacks do not prevent EEMARQ from outperforming all other competitors.

\section{Conclusion} \label{sec-conclusion}

We presented EEMARQ, a design for a lock-free data-structure that supports linearizable inserts, deletes, contains, and range queries. 
Our design starts from a linked-list, which is easier to use with MVCC for fast range queries. We add lock-free memory reclamation to obtain full lock-freedom. Finally, we facilitate an easy integration of a fast external index to speed up the execution, while still providing full linearizability and lock-freedom. As the external index does not require version maintenance, it can remain simple and fast. 
We implemented the design with a skip list and a binary search tree as two possible fast indexes, and evaluated their performance against state-of-the-art solutions.
Evaluation shows that EEMARQ outperforms existing solutions across read-intensive and update-intensive workloads, and for varying range query sizes.
In addition, EEMARQ's memory footprint is relatively low, thanks to its tailored reclamation scheme, enabling the immediate reclamation of deleted objects.
EEMARQ is the only technique that provides lock-freedom, as other existing methods use blocking memory reclamation schemes.



\bibliography{lipics-v2021-sample-article}

\appendix
\newpage
\section{Auxiliary Methods and Initialization} \label{sec-auxiliary}

As described in Section~\ref{sec-list-implementation} and Algorithm~\ref{pseudo-list} and~\ref{pseudo-range}, the linked-list represents a map of key-value pairs. Each pair is represented by a node, consisting of five fields: the mutable \emph{ts} field represents its associated timestamp,
the immutable \emph{key} and \emph{value} fields represent the pair's key and value, respectively, the mutable \emph{next} field holds a pointer to the node's successor in the list, and the immutable \emph{prior} field points to the previous successor of its first predecessor in the list.

\begin{wrapfigure}{R}{0.4\textwidth}
\begin{minipage}{0.4\textwidth}
\begin{algorithm}[H]
\caption{Initializing the list.}
\label{pseudo-init}
\begin{algorithmic}

\State globalTS := 2 \label{alg-init-global-ts}
\State $\bot$ := 1
\State head := \textbf{alloc}($-\infty$, NO\_VAL, 2)
\State tail := \textbf{alloc}($\infty$, NO\_VAL, 2)
\State head $\rightarrow$ next := tail
\State head $\rightarrow$ prior := NULL
\State tail $\rightarrow$ next := NULL
\State tail $\rightarrow$ prior := NULL

\end{algorithmic}
\end{algorithm}
\end{minipage}
\end{wrapfigure}

The list initialization method appears in Algorithm~\ref{pseudo-init}.
The global timestamps clock is initialized to 2, and the $\bot$ constant, representing an uninitialized timestamp, is set to 1.
The list has a single entry point, which is a pointer to the \emph{head} sentinel node. \emph{head}'s key is the minimal key in the key range (denoted as $-\infty$). Its timestamp is set to the initial system timestamp (2 in our implementation) and its next pointer points to the \emph{tail} sentinel node. \emph{tail}'s key is the maximal key in the key range (denoted as $\infty$), its timestamp is equal to \emph{head}'s, and its next pointer points to null. Both \emph{prior} fields point to null, as \emph{head} has no predecessor, and \emph{tail}'s predecessor (which is \emph{head}) has no previous successor. 
After initialization, the list is considered as empty (the sentinel nodes do not represent map items).

\begin{algorithm*}[!ht]
\setlength\multicolsep{0pt}
\small
\begin{multicols}{2}
\begin{algorithmic}[1]


\State \textbf{getTS}() \label{methods-get-ts-call}
\Indent
    \State \textbf{return} globalTS.load() \label{methods-get-ts-return}
\EndIndent

\State \textbf{fetchAddTS}() \label{methods-fetch-add-call}
\Indent
    \State \textbf{return} globalTS.fetch\&add() \label{methods-fetch-add-return}
\EndIndent

\State MARK\_MASK := 0x1 \label{methods-mark-mask}
\State FLAG\_MASK := 0x2 \label{methods-flag-mask}
\State AUX\_MASK := 0x3 \label{methods-aux-mask}

\State \textbf{isMarked}(ptr) \label{methods-is-marked-call}
\Indent
    \If{(ptr $\wedge$ MARK\_MASK)} \textbf{return} true \label{methods-is-marked-return-true}
    \EndIf
    \State \textbf{return} false \label{methods-is-marked-return-false}
\EndIndent
\State \textbf{isFlagged}(ptr) \label{methods-is-flagged-call}
\Indent
    \If{(ptr $\wedge$ FLAG\_MASK)} \textbf{return} true \label{methods-is-flagged--return-true}
    \EndIf
    \State \textbf{return} false \label{methods-is-flagged--return-false}
\EndIndent
\State \textbf{isMarkedOrFlagged}(ptr) \label{methods-is-marked-flagged-call}
\Indent
    \If{(ptr $\wedge$ AUX\_MASK)} \textbf{return} true \label{methods-is-marked-flagged-return-true}
    \EndIf
    \State \textbf{return} false \label{methods-is-marked-flagged-return-false}
\EndIndent
\State \textbf{getRef}(ptr) \label{methods-get-ref-call}
\Indent
    \State \textbf{return} ptr $\wedge$ $\neg$ AUX\_MASK \label{methods-get-ref-return}
\EndIndent
\State \textbf{mark}(node) \label{methods-mark-call}
\Indent
    \State ptr := node $\rightarrow$ next \label{methods-mark-read}
    \If{(\textbf{isMarkedOrFlagged}(ptr))} \textbf{return} false \label{methods-mark-return-false}
    \EndIf
    \State markedPtr := ptr $\vee$ MARK\_MASK \label{methods-mark-mark}
    \State \textbf{return} CAS(\&node $\rightarrow$ next, ptr, markedPtr) \label{methods-mark-cas}
\EndIndent

\State \textbf{flag}(node) \label{methods-flag-call}
\Indent
    \State ptr := node $\rightarrow$ next \label{methods-flag-read}
    \If{(\textbf{isMarkedOrFlagged}(ptr))} \textbf{return} false \label{methods-flag-return-false}
    \EndIf
    \State flaggeddPtr := ptr $\vee$ FLAG\_MASK \label{methods-flag-flag}
    \State \textbf{return} CAS(\&node $\rightarrow$ next, ptr, flaggeddPtr) \label{methods-flag-cas}
\EndIndent

\end{algorithmic}
\end{multicols}
\caption{Our Auxiliary Methods Implementation.}
\label{pseudo-methods}
\end{algorithm*}

The pseudo code for the global timestamps clock and node pointers access auxiliary methods appears in Algorithm~\ref{pseudo-methods}.
The \emph{getTS()} method (lines~\ref{methods-get-ts-call}-\ref{methods-get-ts-return}) is used to atomically read the global timestamps clock (initialized in line~\ref{alg-init-global-ts} of Algorithm~\ref{pseudo-init}), and the \emph{fetchAddTS()} method (lines~\ref{methods-fetch-add-call}-\ref{methods-fetch-add-return}) is used to atomically update it.

We use the pointer's two least significant bits for encapsulating the mark and flag bits (see lines~\ref{methods-mark-mask}-\ref{methods-aux-mask}).
The \emph{isMarked()} (lines~\ref{methods-is-marked-call}-\ref{methods-is-marked-return-false}), \emph{isFlagged()} (lines~\ref{methods-is-flagged-call}-\ref{methods-is-flagged--return-false}), and \emph{isMarkedOrFlagged()} (lines~\ref{methods-is-marked-flagged-call}-\ref{methods-is-marked-flagged-return-false}) methods receive a pointer and return an answer using the relevant bit mask. 
The \emph{getRef()} method (lines~\ref{methods-get-ref-call}-\ref{methods-get-ref-return}) receives a (potentially marked or flagged) pointer and returns the actual reference, ignoring the mark and flag bits.
The \emph{mark()} (lines~\ref{methods-mark-call}-\ref{methods-mark-cas}) and \emph{flag()} (lines~\ref{methods-flag-call}-\ref{methods-flag-cas}) methods receive a node reference (the input pointer is assumed to be unmarked and unflagged), dereference it, and mark or flag the node's next pointer (respectively), assuming it is neither marked nor flagged.

\section{Additional Graphs} \label{sec-additional_graphs}

In this Section we present additional results. The experiments setting is described in Section~\ref{sec-evaluation}. 
In Figure~\ref{fig:throughput-diff-rq} below we present the respective updates and range queries throughput for Figure~\ref{fig:skiplist-25-0} and~\ref{fig:tree-25-0}. For convenience, Figure~\ref{fig:skiplist-throughput} is identical to Figure~\ref{fig:skiplist-25-0}, and Figure~\ref{fig:tree-throughput} is identical to Figure~\ref{fig:tree-25-0}. 
Figures~\ref{fig:skiplist-rq-throughput} and~\ref{fig:tree-rq-throughput} present the respective range queries throughput, and Figures~\ref{fig:skiplist-update-throughput} and~\ref{fig:tree-update-throughput} present the respective updates throughput.

\begin{figure*}[!ht]
\centering
      \begin{subfigure}{0.63\textwidth}

         \includegraphics[width=\textwidth]{figures/legend.png}
      \end{subfigure}
  
     \begin{subfigure}{0.3\textwidth}
    \includegraphics[width=\textwidth]{figures/skiplist-u25-rq0.png}
	\caption{Skiplist. Total throughput evaluation}
	        \label{fig:skiplist-throughput}

    \end{subfigure}
 \hfill
     \begin{subfigure}{0.3\textwidth}
    \includegraphics[width=\textwidth]{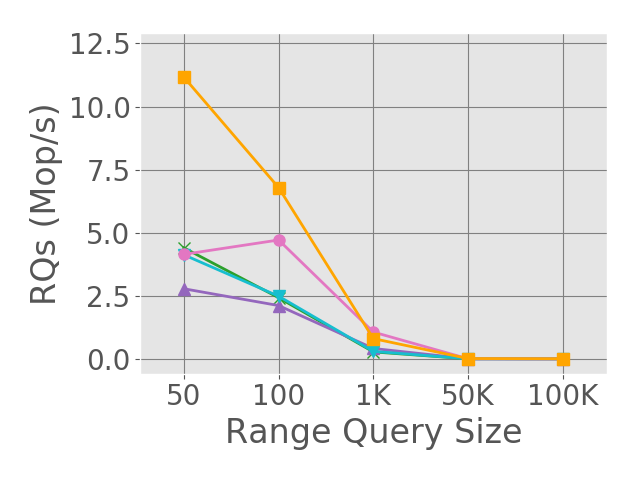}
	\caption{Skiplist. Range queries throughput evaluation}
	        \label{fig:skiplist-rq-throughput}

    \end{subfigure}
 \hfill
    \begin{subfigure}{0.3\textwidth}
   \includegraphics[width=\textwidth]{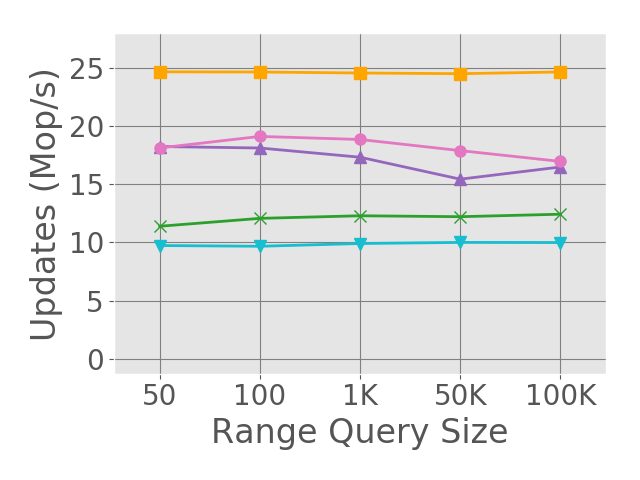}
	\caption{Skiplist. Updates throughput evaluation}
	        \label{fig:skiplist-update-throughput}

    \end{subfigure}
     
     \begin{subfigure}{0.3\textwidth}
    \includegraphics[width=\textwidth]{figures/tree-u25-rq0.png}
	\caption{BST. Total throughput evaluation}
	        \label{fig:tree-throughput}

    \end{subfigure}
 \hfill
     \begin{subfigure}{0.3\textwidth}
    \includegraphics[width=\textwidth]{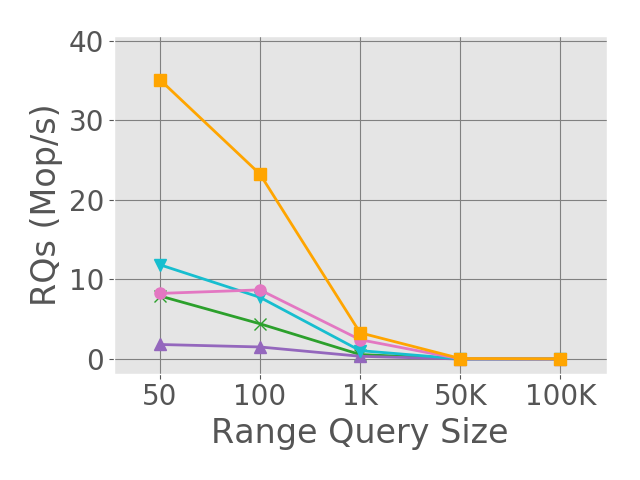}
	\caption{BST. Range queries throughput evaluation}
	        \label{fig:tree-rq-throughput}

    \end{subfigure}
 \hfill
     \begin{subfigure}{0.3\textwidth}
    \includegraphics[width=\textwidth]{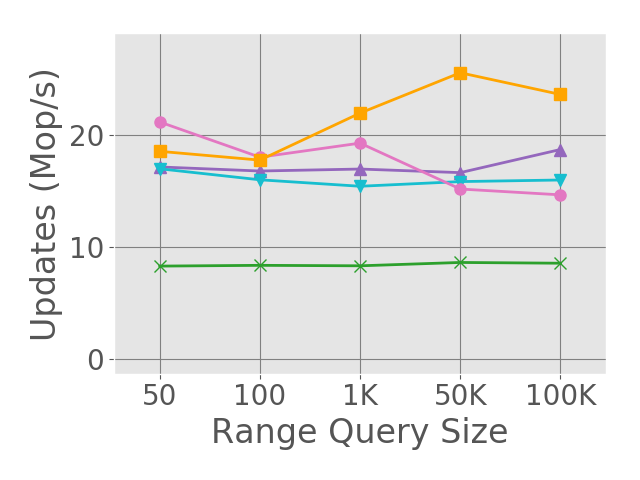}
	\caption{BST. Updates throughput evaluation}
	        \label{fig:tree-update-throughput}

   \end{subfigure}

\caption{Total throughput, range queries throughput and updates throughput for the skiplist (\ref{fig:skiplist-throughput}-\ref{fig:skiplist-update-throughput}) and the tree (\ref{fig:tree-throughput}-\ref{fig:tree-update-throughput}). 
Figures \ref{fig:skiplist-throughput}-\ref{fig:skiplist-update-throughput} and Figures~\ref{fig:tree-throughput}-\ref{fig:tree-update-throughput} present the results for the same runs, respectively.
In both runs, the key range is 1M, there are 64 threads that only perform range queries, and 64 threads that perform $\%50$ contains-$\%50$ updates.
Figures~\ref{fig:skiplist-throughput} and~\ref{fig:tree-throughput} are identical to Figures~\ref{fig:skiplist-25-0} and~\ref{fig:tree-25-0}, respectively.
X axis: range query size. 
Y axis: throughput in Figures~\ref{fig:skiplist-throughput} and~\ref{fig:tree-throughput}, range queries throughput in Figures~\ref{fig:skiplist-rq-throughput} and~\ref{fig:tree-rq-throughput}, and updates throughput in Figures~\ref{fig:skiplist-update-throughput} and~\ref{fig:tree-update-throughput}.} \label{fig:throughput-diff-rq}

\end{figure*}

The EEMARQ skip list surpasses all competitors, and is comparable to its next best scheme, EBR-RQ (see Figures~\ref{fig:skiplist-throughput}-\ref{fig:skiplist-update-throughput}). For long range queries, the EEMARQ tree surpasses its competitors (see Figure~\ref{fig:tree-throughput}) thanks to its range queries and searches high throughput, as its updates throughput deteriorates when the range queries are long (see Figure~\ref{fig:tree-update-throughput}).

\begin{figure*}[!ht]
\centering
      \begin{subfigure}{0.63\textwidth}

         \includegraphics[width=\textwidth]{figures/legend.png}
      \end{subfigure}
  
     \begin{subfigure}{0.3\textwidth}
    \includegraphics[width=\textwidth]{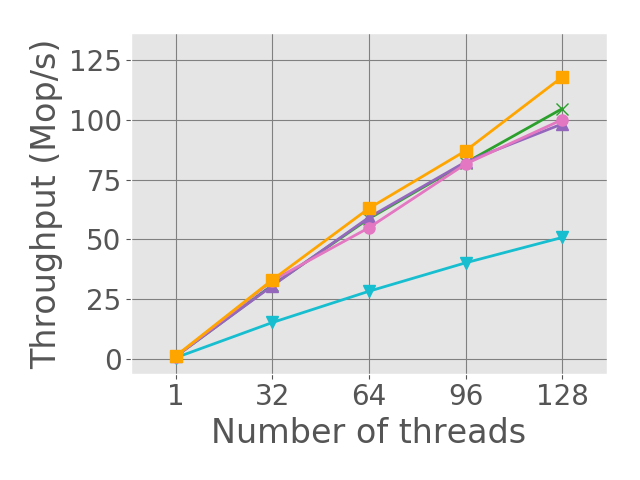}
	\caption{Skiplist. 100\% contains}
	        \label{fig:skiplist-0-0}

    \end{subfigure}
 \hfill
     \begin{subfigure}{0.3\textwidth}
    \includegraphics[width=\textwidth]{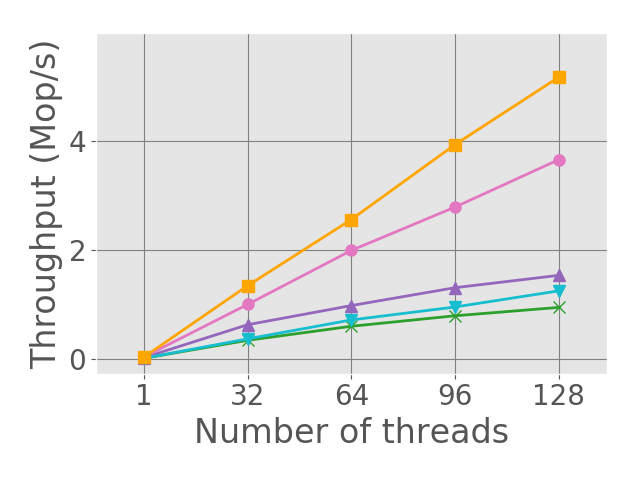}
	\caption{Skiplist. 100\% range queries}
	        \label{fig:skiplist-0-100}

    \end{subfigure}
 \hfill
    \begin{subfigure}{0.3\textwidth}
   \includegraphics[width=\textwidth]{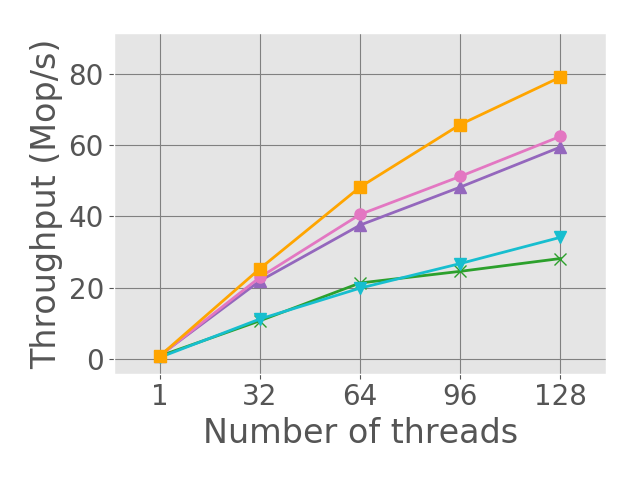}
	\caption{Skiplist. 100\% updates}
	        \label{fig:skiplist-50-0}

    \end{subfigure}
     
     \begin{subfigure}{0.3\textwidth}
    \includegraphics[width=\textwidth]{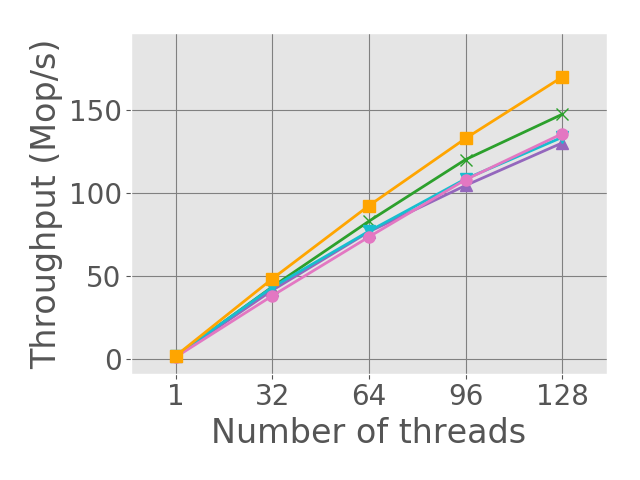}
	\caption{BST. 100\% contains}
	        \label{fig:tree-0-0}

    \end{subfigure}
 \hfill
     \begin{subfigure}{0.3\textwidth}
    \includegraphics[width=\textwidth]{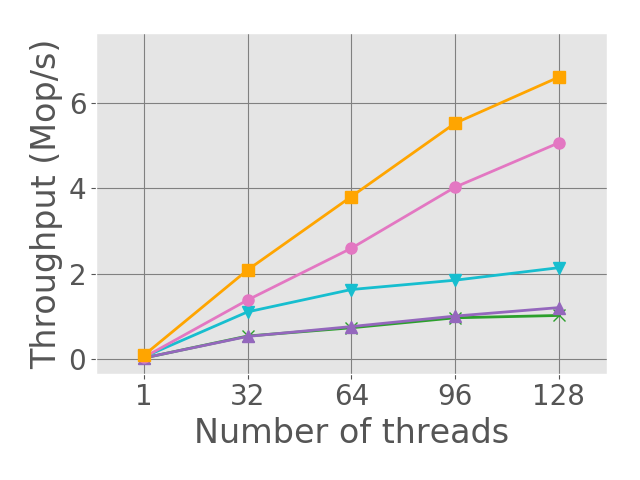}
	\caption{BST. 100\% range queries}
	        \label{fig:tree-0-100}

    \end{subfigure}
 \hfill
     \begin{subfigure}{0.3\textwidth}
    \includegraphics[width=\textwidth]{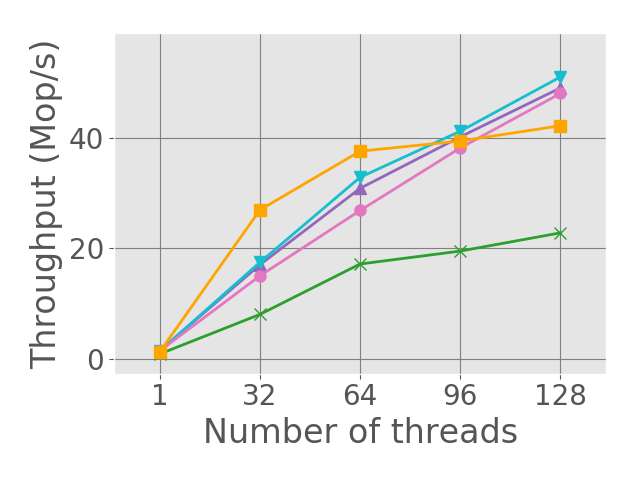}
	\caption{BST. 100\% updates}
	        \label{fig:tree-50-0}

   \end{subfigure}

\caption{Throughput evaluation under various workloads for the skip list (\ref{fig:skiplist-0-0}-\ref{fig:skiplist-50-0}) and the tree (\ref{fig:tree-0-0}-\ref{fig:tree-50-0}). The key range is 1M. The range query size for Figure~\ref{fig:skiplist-0-100} and~\ref{fig:tree-0-100} is 1000. Y axis: throughput in million operations per second. X axis: \#threads.} \label{fig:throughput-isolating}

\end{figure*}

In Figure~\ref{fig:throughput-isolating} we isolate the searches (Figure~\ref{fig:skiplist-0-0} and~\ref{fig:tree-0-0}), range queries (Figure~\ref{fig:skiplist-0-100} and~\ref{fig:tree-0-100}), and updates (Figure~\ref{fig:skiplist-50-0} and~\ref{fig:tree-50-0}). 
For the 100\% contains workload (Figure~\ref{fig:skiplist-0-0} and~\ref{fig:tree-0-0}), the Bundles data-structures surpass all competitors by small margins. As this workload does not use EEMARQ's range queries mechanism or VBR's locality, EEMARQ does not have an advantage when compared against its competitors.
For the 100\% range queries workload (Figure~\ref{fig:skiplist-0-100} and~\ref{fig:tree-0-100}), the EEMARQ data-structures surpass their competitors by high margins, as they impose no indirection, and the executing threads perform no rollbacks (memory is not reclaimed during the execution).
For the 100\% updates workload (Figure~\ref{fig:skiplist-50-0} and~\ref{fig:tree-50-0}), the EEMARQ data-structures are comparable to their best competitors (the EBR-RQ skip list and the EBR-RQ and vCAS trees). It seems that these results are heavily affected by the different skip list and BST implementations, as the Unsafe BST does not surpass the rest of the algorithms for this workload.

In Figure~\ref{fig:throughput-isolating-rq-100} we isolate range queries (as in Figure~\ref{fig:skiplist-0-100} and~\ref{fig:tree-0-100}) show the total range queries throughput when the range queries are shorter (100 instead of 1000). The EEMARQ skip list surpasses its competitors by a high margin for short range queries as well (Figure~\ref{fig:skiplist-0-100-100}), but the EEMARQ tree surpasses its next best competitor (vCAS) by a small margin (Figure~\ref{fig:tree-0-100-100}).

\begin{figure*}[!ht]
\centering
      \begin{subfigure}{0.63\textwidth}

         \includegraphics[width=\textwidth]{figures/legend.png}
      \end{subfigure}
  
     \begin{subfigure}{0.45\textwidth}
    \includegraphics[width=\textwidth]{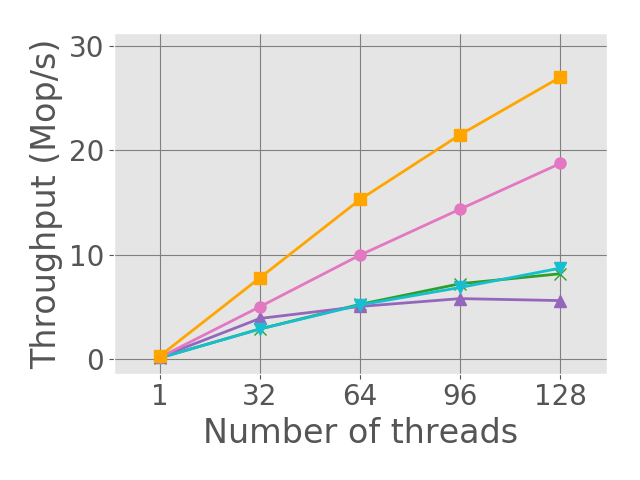}
	\caption{Skiplist. 100\% range queries}
	        \label{fig:skiplist-0-100-100}

    \end{subfigure}
 \hfill
     \begin{subfigure}{0.45\textwidth}
    \includegraphics[width=\textwidth]{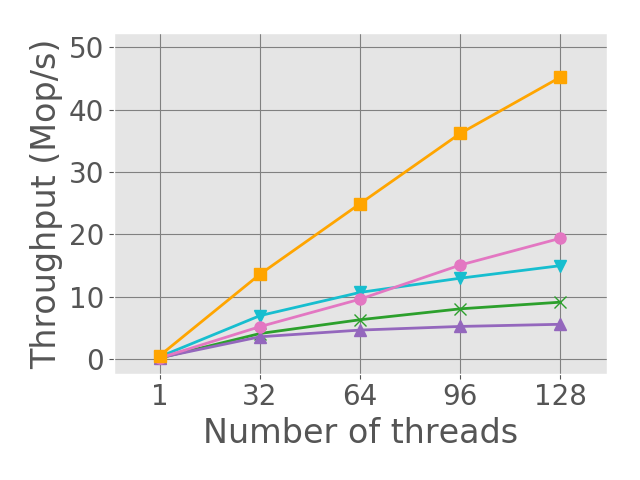}
	\caption{BST. 100\% range queries}
	        \label{fig:tree-0-100-100}

   \end{subfigure}

\caption{Range queries throughput for the skip list (\ref{fig:skiplist-0-100-100}) and the tree (\ref{fig:tree-0-100-100}). The key range is 1M. The range query size is 100. Y axis: throughput in million operations per second. X axis: \#threads.} \label{fig:throughput-isolating-rq-100}

\end{figure*}

\section{Correctness} \label{sec-correctness}

In this Section we prove that the linked-list implementation, as presented in Algorithms~\ref{pseudo-list} and~\ref{pseudo-range} (without taking reclamation or the fast index integration into account) is linearizable and lock-free.
We start by setting preliminaries in Section~\ref{sec-preliminaries}. 
We introduce some basic definitions and lemmas in Section~\ref{sec-basic-definitions}, continue with a full linearizability proof in Section~\ref{sec-linearization-points}, and prove that the implementation is lock-free in Section~\ref{sec-lock-free}.

\subsection{Preliminaries} \label{sec-preliminaries}

We use the basic asynchronous shared memory model from~\cite{herlihy1991wait}. The system consists of a set of executing threads, communicating through the shared memory. The shared memory is accessed by hardware-provided atomic operations, such as reads, writes, compare-and-swap (CAS) and wide-compare-and-swap (WCAS). The CAS operation receives an address of a memory word, an expected value and a new value. If the address content is equal to the expected value, it replaces it with the new value and returns true. Otherwise, it returns false. The WCAS does the same for two adjacent words in memory. 
%
We say that a concurrent implementation is \textit{lock-free} if, given that the executing threads are run for a sufficiently long time, at least one thread makes progress. 
In addition, we use the standard correctness condition of \emph{linearizability} from~\cite{herlihy1990linearizability}. Briefly, any operation must take effect at some point (referred to as its \emph{linearizability point}) between its invocation and response.

A mapping data structure represents a map from a set of unique keys to their respective values. Each key-value pair is represented by a node, consisting typically of immutable key and value, and some other mutable fields (e.g., \textit{next} pointers).
The data-structure has an entry point (e.g., a linked-list head~\cite{harris2001pragmatic} or a tree root~\cite{natarajan2014fast}).
The physical insertion of a node is the act of making it reachable from the entry point, and the physical deletion of a node is its unlinking from the data-structure (making it unreachable from the entry point). 
Many concurrent data-structure implementations~\cite{wei2021constant,harris2001pragmatic,natarajan2014fast} also include a level of logical insertion and deletion between the physical ones. I.e., a reachable node does not necessarily represent a key-value pair which is currently in the map. 

\subsection{Basic Definitions and Lemmas} \label{sec-basic-definitions}

We define reachability using Definition~\ref{definition-reachability} and~\ref{definition-reachable}:

\begin{definition} \label{definition-reachability}
We say that a node $m$ is \textit{reachable} from a node $n$ if there exist nodes $n_0,\ldots,n_k$ such that $n=n_0$, $m=n_k$, and for every $0 \leq i < k$, $n_i$'s $next$ pointer points to $n_{i+1}$.
\end{definition}

\begin{definition} \textbf{(reachable nodes)} \label{definition-reachable}
We say that a node $n$ is \textit{reachable} if it is reachable from \emph{head}.
\end{definition}

According to Definition~\ref{definition-reachable}, both \emph{head} and \emph{tail} are reachable at initialization.
However, while \emph{head} remains reachable throughout the execution, \emph{tail} may be removed, and replaced by another node with an $\infty$ key. Only the original node with an $\infty$ key is refereed to as \emph{tail}.

Note that the reachable path from Definition~\ref{definition-reachability} may consist of either marked or flagged pointers. Moreover, and as opposed to previous implementations~\cite{harris2001pragmatic,natarajan2014fast,herlihy2020art}, the marking of a node's next pointer does not serve as its logical deletion from the list. 

Before moving on to our linearizability proof, we must first define the term of being logically in the list. 
This term is presented in Definition~\ref{definition-logically-in}, after we define some other basic terms and prove several supporting lemmas.

\begin{definition} \label{definition-node-states}
Let $n$ be a node. We denote $n$'s possible states as follows:
\begin{itemize}
    \item We say that $n$ is \emph{pending} if $n$'s timestamp is $\bot$, and there does not exist an execution of the \emph{getTS()} command in line~\ref{alg-insert-update-ts},~\ref{alg-find-update-pred-ts},~\ref{alg-find-update-curr-ts},~\ref{alg-trim-update-curr-ts},~\ref{alg-trim-update-succ-ts}, or \ref{alg-trim-update-new-curr-ts} of Algorithm~\ref{pseudo-list}, or in line~\ref{alg-range-query-smaller-curr-succ-ts-update} or~\ref{alg-range-query-bigger-curr-succ-ts-update} of Algorithm~\ref{pseudo-range}, for which (1) the referenced node is $n$, (2) the following CAS execution is successful\footnote{Note that the node stops being pending before the CAS is executed}.
    \item We say that $n$ is \emph{marked} if its next pointer is marked.
    \item We say that $n$ is \emph{flagged} if its next pointer is flagged.
    \item We say that $n$ is \emph{active} if it is in none of the above states.
\end{itemize}
\end{definition}

\begin{definition} \textbf{(infant nodes)} \label{definition-infant}
We say that a node $n$ is an \textit{infant} if (1) it is not one of the two sentinel nodes , (2) it has never been the third input value to a successful CAS in line~\ref{alg-insert-update-pred} of Algorithm~\ref{pseudo-list}, and (3) it has never been the third input value to a successful CAS in line~\ref{alg-trim-cas} of Algorithm~\ref{pseudo-list}.
\end{definition}

\begin{definition} \textbf{(unlinked nodes)} \label{definition-unlinked}
Let $n_1$ be a node, traversed during the loop in lines~\ref{alg-trim-while}-\ref{alg-trim-curr-gets-curr-next}, and let $n_2$ be the node, allocated in line~\ref{alg-trim-alloc}. If $n_2$ eventually becomes active, then $n_1$ is considered as \textit{unlinked} once $n_2$ becomes active.
\end{definition}

\begin{lemma} \label{lemma-immutable-ts}
Once a node's timestamp is not $\bot$, it is immutable.
\end{lemma}

\begin{proof}
Timestamps are updated in lines~\ref{alg-insert-update-ts},~\ref{alg-find-update-pred-ts},~\ref{alg-find-update-curr-ts},~\ref{alg-trim-update-curr-ts},~\ref{alg-trim-update-succ-ts} and~\ref{alg-trim-update-new-curr-ts} of Algorithm~\ref{pseudo-list}, and in lines~\ref{alg-range-query-smaller-curr-succ-ts-update} and~\ref{alg-range-query-bigger-curr-succ-ts-update} of Algorithm~\ref{pseudo-range}.
In all cases, the expected value is $\bot$ and the new written value is not $\bot$. Therefore, once a node's timestamp is not $\bot$, it is immutable.
\end{proof}

\begin{lemma} \label{lemma-deleted-states}
Marked and flagged next pointers are immutable.
\end{lemma}

\begin{proof}
Both the \emph{mark()} and \emph{flag()} methods, called in line~\ref{alg-remove-mark} and~\ref{alg-trim-flag} of Algorithm~\ref{pseudo-list}, respectively, change the pointer only if it is neither marked nor flagged (see Figure~\ref{pseudo-methods}).
The two remaining scenarios of pointer updates are in lines~\ref{alg-insert-update-pred} and~\ref{alg-trim-cas} of Algorithm~\ref{pseudo-list}. In both cases, the expected and new pointer values are neither marked nor flagged.
Therefore, marked and flagged pointer never change.
\end{proof}

\begin{lemma} \label{lemma-node-states}
Let $n$ be a node. Then $n$'s state can only go through the following transitions: 
\begin{itemize}
    \item From being pending to being active.
    \item From being active to being marked.
    \item From being active to being flagged.
\end{itemize}
\end{lemma}

\begin{proof}
By Lemma~\ref{lemma-immutable-ts}, once a node's timestamp is not $\bot$, it is immutable.
In addition, once the respective \emph{getTS()} command from Definition~\ref{definition-node-states} is executed, a node cannot be pending anymore.  
Therefore, active, marked, and flagged nodes never become pending.
In addition, by Lemma~\ref{lemma-deleted-states}, marked and flagged pointers are immutable. In particular, a marked node cannot become active or flagged, and a flagged node cannot become active or marked.
It still remains to show that a pending node cannot become marked or flagged. 
Nodes are marked in line~\ref{alg-remove-mark} of Algorithm~\ref{pseudo-list}, after being returned as the second output parameter from the \emph{find()} call in line~\ref{alg-remove-find}. Before this \emph{find()} call returns, the node's timestamp is updated in line~\ref{alg-find-update-curr-ts} (if the node is not already active).
In a similar way, right before a node is flagged in line~\ref{alg-trim-flag} of Algorithm~\ref{pseudo-list}, its timestamp is updated in line~\ref{alg-trim-update-curr-ts}, for making sure that it is no longer pending (the update fails if the node is already active). 
\end{proof}

\begin{lemma} \label{lemma-invariants} Let $n_1$ and $n_2$ be two different nodes. Then:
\begin{enumerate}
    \item \label{lemma-invariants-infant} If $n_1$'s next pointer points to $n_2$, than $n_2$ is not an infant.
    
    \item \label{lemma-invariants-pending-infant} Suppose that $n_1$ is the third input parameter to a successful CAS execution, either in line~\ref{alg-insert-update-pred} or~\ref{alg-trim-cas} of Algorithm~\ref{pseudo-list}. Then, starting from its allocation and until this CAS execution, it holds that (1) $n_1$ is an infant, and (2) $n_1$ is pending.
    
    \item \label{lemma-invariants-still-reachable} If $n_2$ was reachable from $n_1$ at a certain point during execution, and $n_2$ is either pending or active, then $n_2$ is still reachable from $n_1$.
    
    \item \label{lemma-invariants-active-reachable} If $n_1$ is not an infant, and is either pending or active, then it is reachable.
    
    \item \label{lemma-invariants-succ-pending-reachable} If $n_1$ is not an infant, $n_2$ is $n_1$'s successor, and $n_2$ is pending, then $n_1$ is reachable.
    
    \item \label{lemma-invariants-not-reachable} If $n_1$ was not pending and not reachable at some point, then it is still not reachable.
    
    \item \label{lemma-invariants-pending-still-successor} If $n_2$ was $n_1$'s successor at a certain point during execution, and either $n_1$ or $n_2$ are pending, then $n_2$ is still $n_1$'s successor.
    
    \item \label{lemma-invariants-smaller-not-reachable} If $n_1$'s key is smaller than or equal to $n_2$'s key, then $n_1$ is not reachable from $n_2$. 
    
    \item \label{lemma-invariants-still-reachable-before-trim} Assume $n_1$ is not an infant. If either $n_1$ has not been traversed during the loop in lines~\ref{alg-trim-while}-\ref{alg-trim-curr-gets-curr-next} yet, or that the following CAS in line~\ref{alg-trim-cas} (during the same \emph{trim()} execution) has not been successfully executed yet, $n_1$ is reachable.
    
    \item \label{lemma-invariants-unlinked} If $n_1$ is unlinked, then it is not reachable.
    
    \item \label{lemma-invariants-tail} If $n_1$ is not an infant then there exists a node with a key $\infty$ which is reachable from $n_1$.
    
    \item \label{lemma-invariants-no-null-dereference} If a pointer is assigned into a local variable in Algorithm~\ref{pseudo-list} and is later dereferenced, then it is not null.
    
\end{enumerate}
\end{lemma}

\begin{proof}
We are going to prove the lemma by induction on the length of the execution. At the initial stage, head and tail are the only nodes in the list, having $-\infty$ and $\infty$ keys (respectively), and head's next pointer points to tail. By Definition~\ref{definition-infant}, tail is not an infant, and Invariant~\ref{lemma-invariants-infant} holds. Invariants~\ref{lemma-invariants-pending-infant},~\ref{lemma-invariants-succ-pending-reachable},~\ref{lemma-invariants-not-reachable},~\ref{lemma-invariants-pending-still-successor},~\ref{lemma-invariants-unlinked} and~\ref{lemma-invariants-no-null-dereference} vacuously hold, and Invariants~\ref{lemma-invariants-still-reachable},~\ref{lemma-invariants-active-reachable},~\ref{lemma-invariants-smaller-not-reachable},~\ref{lemma-invariants-still-reachable-before-trim} and~\ref{lemma-invariants-tail} hold since tail is head's successor and both nodes are reachable by Definition~\ref{definition-reachable}.
Now, assume that all of the invariants hold at a certain point during the execution, and let $s$ be the next execution step, executed by a thread $t$.

\begin{enumerate}
    \item Assume by contradiction that after executing $s$, there exist two different nodes $n_1$ and $n_2$, such that $n_1$'s next pointer points to $n_2$, and $n_2$ is an infant. By the induction hypothesis, either $n_2$ is not an infant before executing $s$ (and thus, by Definition~\ref{definition-infant}, it is also not an infant after executing $s$ -- a contradiction), or $n_1$'s next pointer does not point to $n_2$ right before executing $s$. Therefore, $s$ has to be a pointer initialization (line~\ref{alg-insert-update-next} or~\ref{alg-trim-init-next} in Algorithm~\ref{pseudo-list}) or update (line~\ref{alg-insert-update-pred} or~\ref{alg-trim-cas} in Algorithm~\ref{pseudo-list}). In the first two cases, $n_2$ already have a predecessor prior to $s$, and is not an infant by the induction hypothesis -- a contradiction. In addition, by Definition~\ref{definition-infant}, after executing line~\ref{alg-insert-update-pred} or~\ref{alg-trim-cas}, $n_2$ is no longer an infant, which also derives a contradiction. Therefore, the invariant holds.
    
    \item We first prove that $n_1$ is still an infant. By Definition~\ref{definition-infant}, it suffices to assume by contradiction that $n_1$ is not an infant after $s$, and that $t$'s next step is the execution of line~\ref{alg-insert-update-pred} or~\ref{alg-trim-cas} of Algorithm~\ref{pseudo-list}, having $n_1$ as the third input parameter to the respective successful CAS. $n_1$ is either allocated in line~\ref{alg-insert-alloc} or~\ref{alg-trim-alloc}, respectively. By Definition~\ref{definition-infant}, $n_1$ is an infant after allocated. Since $n_1$ is no longer an infant, there exist a former successful CAS execution, either in line~\ref{alg-insert-update-pred} or~\ref{alg-trim-cas}, for which $n_1$ is the third input parameter. This implies that $n_1$ was allocated twice -- a contradiction.
    
    Now, assume by contradiction that $n_1$ is not pending before either line~\ref{alg-insert-update-pred} or~\ref{alg-trim-cas} of Algorithm~\ref{pseudo-list} is executed, having $n_1$ as the third input parameter to the respective successful CAS. 
    $n_1$ is either allocated in line~\ref{alg-insert-alloc} or~\ref{alg-trim-alloc}, respectively. By Definition~\ref{definition-node-states}, $n_1$ is pending after allocated. Since $n_1$ is guaranteed to still be an infant, its timestamp was not updated in line~\ref{alg-insert-update-ts} or~\ref{alg-trim-update-new-curr-ts}. Therefore, it was either updated in line~\ref{alg-find-update-pred-ts},~\ref{alg-find-update-curr-ts},~\ref{alg-trim-update-curr-ts} or ~\ref{alg-trim-update-succ-ts} of Algorithm~\ref{pseudo-list}, or in line~\ref{alg-range-query-smaller-curr-succ-ts-update} or~\ref{alg-range-query-bigger-curr-succ-ts-update} of Algorithm~\ref{pseudo-range}. In all cases, $n_1$ already had a predecessor at this point, and by the induction hypothesis (Invariant~\ref{lemma-invariants-infant}), it is no longer an infant -- a contradiction.
    Therefore, $n_1$ is neither an infant nor pending, and the invariant holds.
    
    \item Assume by contradiction that $n_1$ was reachable from $n_2$ at a certain point during the execution, and that after $s$, $n_1$ is either pending or active, and not reachable from $n_2$. By the induction hypothesis, $n_2$ is reachable from $n_1$ before $s$. 
    Note that $s$ cannot be one of the pointer initializations (in line~\ref{alg-insert-update-next} and~\ref{alg-trim-init-next}), as the respective infants did not have any former successors.
    Therefore, $s$ must be a successful CAS execution, either in line~\ref{alg-insert-update-pred} or~\ref{alg-trim-cas} of Algorithm~\ref{pseudo-list}.
    
    If $s$ is the execution of line~\ref{alg-insert-update-pred}, then right before $s$, $pred$ must be reachable from $n_1$ and $n_2$ must be reachable from $curr$. Since $curr$ is still reachable from $pred$ after executing $s$, $n_2$ is also still reachable from $n_1$ -- a contradiction.
    
    Otherwise, if $s$ is the execution of line~\ref{alg-trim-cas}, then right before $s$, $pred$ must be reachable from $n_1$ and $n_2$ must be reachable from $victim$.
    Since $n_2$ is not marked or flagged, by Lemma~\ref{lemma-node-states} it is also guaranteed that it has not been marked or flagged before executing $s$. Therefore, $n_2$ is not among the marked nodes, traversed during the loop in lines~\ref{alg-trim-while}-\ref{alg-trim-curr-gets-curr-next}. In addition, it cannot be the flagged $curr$ node (flagged not later than the execution of line~\ref{alg-trim-flag}). Since marked and flagged pointers are immutable (see Lemma~\ref{lemma-deleted-states}), $n_2$ is reachable from $succ$, which is $curr$'s successor, read in line~\ref{alg-trim-read-succ}. As $newCurr$'s successor is set to be $succ$ in line~\ref{alg-trim-init-next}, by the induction hypothesis, $n_2$ is reachable from $newCurr$, right before executing $s$. Since $s$ does not update $newCurr$'s next pointer, $n_2$ is also reachable from $newCurr$ after executing $s$. Finally, since $newCurr$ is $pred$'s successor after executing $s$, $n_2$ is still reachable from $n_1$ -- a contradiction.
    
    \item Using Definition~\ref{definition-reachable} and Invariant~\ref{lemma-invariants-still-reachable}, it suffices to show that $n_1$ is reachable when it stops being an infant. I.e., when it is the third input parameter to a successful CAS, either in line~\ref{alg-insert-update-pred} or~\ref{alg-trim-cas}. In both cases, $pred$ is not an infant by Invariant~\ref{lemma-invariants-infant}, and is either active or pending~\footnote{Actually, $pred$ cannot be pending at this point.} by Lemma~\ref{lemma-deleted-states} (as its next pointer is mutable). Therefore, by the induction hypothesis, $pred$ is reachable. By Definition~\ref{definition-reachable}, $n_1$ is also reachable after executing the respective CAS.
    
    \item Assume by contradiction that after executing $s$, $n_1$ is not an infant, $n_2$ is $n_1$'s successor, $n_2$ is pending, and $n_1$ is not reachable. Note that by Invariant~\ref{lemma-invariants-infant}, $n_2$ is not an infant, and by Invariant~\ref{lemma-invariants-active-reachable}, it is still reachable after $s$.
    If $n_1$ is either pending or active, then by Invariant~\ref{lemma-invariants-active-reachable}, it is reachable -- a contradiction. Therefore, $n_1$ must either be marked or flagged. By Lemma~\ref{lemma-deleted-states}, $n_2$ is $n_1$'s successor before $s$ and thus, by the induction hypothesis, $n_1$ is reachable before $s$. I.e., $s$ updates a next pointer.
    
    If $s$ is the execution of line~\ref{alg-insert-update-pred}, then $n_1$ must be reachable from the node referenced by the $pred$ variable before $s$, and not reachable from it after $s$. By definition~\ref{definition-reachability}, $n_1$ must be reachable from the node referenced by the $curr$ variable. As the latter is still reachable from the former after the pointer change (via the node referenced by the $newNode$ variable), we get a contradiction.
    
    Otherwise, if $s$ is the execution of line~\ref{alg-trim-cas}, then $n_1$ must be reachable from the node referenced by the $pred$ variable before $s$, and not reachable from it after $s$. I.e., while $n_1$ is not reachable from the $newCurr$ variable, $n_2$ must be reachable from it (as $n_2$ remains reachable after the pointer change, and $s$ does not change other pointers). The only possible scenario is the one in which $n_2$ is referenced by the $succ$ variable, as the predecessors of all following nodes remain reachable after this pointer change. Since $n_2$ is no longer pending at this stage (its timestamp is updated no later than the execution of line~\ref{alg-trim-update-succ-ts}), we get a contradiction.
    
    Therefore, if $n_1$ is not an infant, $n_2$ is $n_1$'s successor, and $n_2$ is pending, then $n_1$ is reachable.
    
    \item Assume by contradiction that $n_1$ was not pending and not reachable at some point before executing $s$, and that it is reachable after $s$. By the induction hypothesis, it has not been reachable until right before executing $s$. 
    By Definition~\ref{definition-infant}, a node cannot become an infant after not being an infant. I.e., by Definition~\ref{definition-reachable}, $s$ physically inserts $n_1$ into the list. W.l.o.g., assume that for every other node $n'$, formerly not reachable (while not being an infant) and then reachable after $s$, $n'$ is reachable from $n_1$ after $s$. I.e., there are no such nodes that precede $n_1$ after $s$. This means that $n_1$'s reachable predecessor after $s$~\footnote{By Definition~\ref{definition-reachable}, $n_1$ has a single reachable predecessor} was also reachable before $s$, and that $s$ updated its next pointer to point to $n_1$.
    
    For convenience, let us denote this predecessor with $n_0$. Since $n_0$ is reachable prior to $s$, by Invariant~\ref{lemma-invariants-infant} it is not an infant. Consequently, by Invariant~\ref{lemma-invariants-pending-infant}, $s$ cannot be the execution of line~\ref{alg-insert-update-next} or~\ref{alg-trim-init-next}. It remains to show that it also cannot be the execution of line~\ref{alg-insert-update-pred} or~\ref{alg-trim-cas}.
    
    In both cases, by Invariant~\ref{lemma-invariants-pending-infant}, $n_0$'s successor after $s$ is an infant right before $s$, and we get a contradiction.
    Therefore, a former unreachable (while not being an infant) node never becomes reachable.

    \item Assume that $n_2$ was $n_1$'s successor at a certain point during the execution, and that either $n_1$ or $n_2$ are pending after $s$. If $n_2$ is $n_1$'s successor after $s$, then we are done. Otherwise, by Lemma~\ref{lemma-node-states}, either $n_1$ or $n_2$ are pending before $s$, and by the induction hypothesis, $n_2$ is $n_1$'s successor right before $s$. Therefore, $s$ must change $n_1$'s successor, either in line~\ref{alg-insert-update-pred} or~\ref{alg-trim-cas}.
    
    If $n_1$'s next pointer is updated in line~\ref{alg-insert-update-pred}, then both $n_1$ and $n_2$ are not pending, as their timestamps were updated no later than the execution of lines~\ref{alg-find-update-pred-ts} and~\ref{alg-find-update-curr-ts} (respectively), during the \emph{find()} execution, invoked in line~\ref{alg-insert-find}.
    
    If $n_1$'s next pointer is updated in line~\ref{alg-trim-cas}, then $n_1$ is not pending, as its timestamp is updated no later than the execution of line~\ref{alg-find-update-pred-ts}, during the calling \emph{find()} execution.
    In addition, during this calling \emph{find()} execution, a pointer to $n_2$ was written into the $predNext$ variable, either in line~\ref{alg-find-head-next} or~\ref{alg-find-get-pred-next}. Assuming $n_2$ is still pending, by the induction hypothesis, when the condition in line~\ref{alg-find-if-not-adjacent} was checked, $n_2$ was still $n_1$'s successor. I.e., the $curr$ variable was updated at least once during the loop in lines~\ref{alg-find-while-is-marked}-\ref{alg-find-curr-gets-curr-next}. This means that $n_2$ was either marked or flagged when the condition in line~\ref{alg-find-while-is-marked} was checked -- a contradiction to Lemma~\ref{lemma-node-states}.
    Therefore, if either $n_1$ or $n_2$ are pending after $s$, then $n_2$ is still $n_1$'s successor after $s$, and the invariant holds.
    
    \item Assume by contradiction that $n_1$'s key is smaller than or equal to $n_2$'s key, and that $n_1$ is reachable from $n_2$. Recall that keys are immutable. Therefore, by the induction hypothesis, $n_1$ was not reachable from $n_2$ before $s$. 
    Then $s$ must be a pointer update. 
    
    Obviously, $s$ cannot be the execution of line~\ref{alg-insert-update-next}; The $curr$ node, returned as output from the \emph{find()} method, has a key which is at least $newNode$'s key (checked either in line~\ref{alg-find-key-bigger} or~\ref{alg-find-key-bigger-after-trim}). Moreover, since the condition checked in line~\ref{alg-insert-return-curr-val} does not hold for $curr$, its key is strictly bigger than $newNode$'s key. By the induction hypothesis, all nodes reachable from $curr$ have bigger keys, so executing line~\ref{alg-insert-update-next} cannot violate the condition.
    
    In a similar way, $s$ cannot be the execution of line~\ref{alg-trim-init-next}; By the induction hypothesis, $succ$'s key is strictly bigger than $curr$'s key, and thus it is strictly bigger than $newCurr$'s key. By the induction hypothesis, all nodes reachable from $succ$ have bigger keys, so executing line~\ref{alg-trim-init-next} cannot violate the condition.
    
    It still remains to show that $s$ cannot be the execution of line~\ref{alg-insert-update-pred} or~\ref{alg-trim-cas}.
    First, assume by contradiction that $s$ is the execution of line~\ref{alg-insert-update-pred}. By the induction hypothesis, the invariant holds prior to $s$. Therefore, $pred$ must be reachable from $n_1$ and $n_2$ must be reachable from $newNode$. As explained above, the invariant holds for $newNode$ and all of its successive nodes, for any possible path. In addition, by the induction hypothesis, the invariant holds for $pred$ and all of the nodes preceding it. Since $pred$ is the node returned from the \emph{find()} method, its key is strictly smaller than $newNode$'s key. Therefore, the violating step cannot be the execution of line~\ref{alg-insert-update-pred}.
    
    The remaining possible step is the execution of line~\ref{alg-trim-cas}. Again, by the induction hypothesis, $pred$ must be reachable from $n_1$ and $n_2$ must be reachable from $newCurr$. As explained above, the invariant holds for $newCurr$ and all of its successive nodes, for any possible path. In addition, the invariant holds for $pred$ and all of the nodes preceding it. Since $curr$ is reachable from $pred$, its key is strictly bigger by the induction hypothesis. Consequently, $newCurr$'s key is also strictly bigger than $pred$'s. Therefore, the violating step cannot be the execution of line~\ref{alg-trim-cas}, and we get a contradiction.
    Since there are no other possible cases, the invariant holds.
    
    \item If $n_1$ is either pending or active, then by Invariant~\ref{lemma-invariants-active-reachable}, it is reachable. Otherwise, by the induction hypothesis, it is still reachable before $s$. Assume by contradiction that it is not reachable after $s$, and that either $s$ is not the execution of line~\ref{alg-trim-cas}, or $n_1$ was not traversed during the loop in lines~\ref{alg-trim-while}-\ref{alg-trim-curr-gets-curr-next} of the same \emph{trim()} execution. 
    
    If $s$ is not the execution of line~\ref{alg-trim-cas}, then it must be a next pointer change. I.e., it must be a successful CAS execution in line~\ref{alg-insert-update-pred}. If the node, referenced by the $pred$ variable, is reachable from $n_1$ before $s$, then by Invariant~\ref{lemma-invariants-smaller-not-reachable}, this pointer change does not affect $n_1$'s reachability. Otherwise, prior to $s$, $n_1$ must be reachable from the node referenced by the $curr$ variable. In this case, $n_1$ remains reachable after $s$ -- a contradiction.
    
    Otherwise, $n_1$ was not traversed during the loop in lines~\ref{alg-trim-while}-\ref{alg-trim-curr-gets-curr-next}, but $s$ is still the execution of line~\ref{alg-trim-cas}. Let $n_0$ be the node referenced by the $pred$ variable. If $n_1$'s key is not bigger then $n_0$'s key, then by Invariant~\ref{lemma-invariants-smaller-not-reachable}, $n_1$'s reachability is not affected by this CAS. 
    
    By Lemma~\ref{lemma-deleted-states}, all nodes, traversed during the loop in  lines~\ref{alg-trim-while}-\ref{alg-trim-curr-gets-curr-next}, are immutable after the execution of line~\ref{alg-trim-flag}. Recall that by the induction hypothesis, $n_1$ is still reachable before $s$. Therefore, if $n_1$'s key is bigger than $n_0$'s key, then it also must be reachable from the node referenced by the $succ$ variable. I.e., $n_1$ is still reachable after $s$ -- a contradiction. Any possible case lead to a contradiction and thus, the invariant holds.
    
    \item By Invariant~\ref{lemma-invariants-pending-infant}, $n_2$ (as defined in Definition~\ref{definition-unlinked}) becomes active only after the execution of line~\ref{alg-trim-cas}. At this point, $n_1$ is either marked or flagged, and by Lemma~\ref{lemma-node-states}, it is not pending when line~\ref{alg-trim-cas} is executed. By Invariant~\ref{lemma-invariants-not-reachable}, it suffices to show that $n_1$ is not reachable after the execution of line~\ref{alg-trim-cas}, during the respective \emph{trim()} execution from Definition~\ref{definition-unlinked}.
    
    Assume by contradiction that $s$ is the execution of this update, and that $n_1$ is reachable after $s$.
    Let $n_0$ be the node referenced by the $pred$ local variable. By Invariant~\ref{lemma-invariants-smaller-not-reachable}, $n_0$'s key is strictly smaller than $n_1$'s key, and $n_2$'s key is at least $n_1$'s key.
    As $n_0$'s timestamp had been updated no later then the execution of line~\ref{alg-find-update-pred-ts}, during the calling \emph{find()} method, $n_0$ is not pending.
    By Invariant~\ref{lemma-invariants-pending-infant}, $n_0$ is not an infant.
    By Definition~\ref{definition-infant} and Invariant~\ref{lemma-invariants-pending-infant}, $n_2$ is pending after executing line~\ref{alg-trim-cas}.
    Therefore, by Invariant~\ref{lemma-invariants-active-reachable} and~\ref{lemma-invariants-succ-pending-reachable}, both $n_0$ and $n_2$ are reachable after executing line~\ref{alg-trim-cas}. By Invariant~\ref{lemma-invariants-smaller-not-reachable}, $n_1$ cannot be reachable at this point -- a contradiction.
    Therefore, if $n_1$ is unlinked, then it is not reachable.
    
    \item Assume by contradiction that after $s$, $n_1$ is not an infant and there does not exist a node with a key $\infty$ which is reachable from $n_1$.
    By the induction hypothesis, the invariant holds before $s$. I.e., either $n_1$ is an infant before $s$, or there exists a node with a key $\infty$ which is reachable from $n_1$.
    
    If $n_1$ is an infant before $s$, then by Definition~\ref{definition-infant}, $s$ must be the execution of line~\ref{alg-insert-update-pred} or~\ref{alg-trim-cas}, having $n_1$ referenced by the $newNode$ or $newCurr$ variable, respectively.
    By Invariant~\ref{lemma-invariants-infant}, $n_1$'s successor is not an infant. By the induction hypothesis, a node with a key $\infty$ is reachable from $n_1$'s successor, and by Definition~\ref{definition-reachability}, it is also reachable from $n_1$ -- a contradiction.
    
    Therefore, there exists a node with a key $\infty$ which is reachable from $n_1$ before $s$. 
    W.l.o.g., assume that $n_1$ is the node with the maximal key, for which there does not exist a reachable node with a key $\infty$. By Definition~\ref{definition-reachability}, $n_1$'s key is smaller than $\infty$.
    Let us denote the respective node (with a key $\infty$), reachable from $n_1$ before $s$, with $n_0$. 
    In addition, let us denote $n_1$'s successor before $s$, with $n_2$ (By the induction hypothesis, $n_1$ necessarily has a successor before $s$).
    By the choice of $n_1$ and Invariant~\ref{lemma-invariants-smaller-not-reachable}, a node with a key $\infty$ is reachable from $n_2$ after $s$.
    By Definition~\ref{definition-reachability}, $s$ updates $n_1$'s next pointer. Since $n_1$ is not an infant, by Invariant~\ref{lemma-invariants-pending-infant}, $s$ is the execution of line~\ref{alg-insert-update-pred} or~\ref{alg-trim-cas}, having $n_1$ referenced by the $pred$ variable.
    
    If $s$ is the execution of line~\ref{alg-insert-update-pred}, then $n_2$ is still reachable from $n_1$ -- a contradiction. Therefore, $s$ is the execution of line~\ref{alg-trim-cas}.
    Let $n_3$ be the node referenced by the $succ$ variable.
    By Invariant~\ref{lemma-invariants-smaller-not-reachable}, $s$ does not affect the reachability of any node from $n_3$. I.e., by the induction hypothesis, a node with a key $\infty$ is reachable from $n_3$ after $s$. As $n_3$ is reachable from $n_1$ after $s$ (by Invariant~\ref{lemma-invariants-pending-infant} and~\ref{lemma-invariants-pending-still-successor}), we get a contradiction to Definition~\ref{definition-reachability}.
    Therefore, $succ$ points to NULL.
    By the induction hypothesis, this means that the last node referenced by the $curr$ variable, has a key $\infty$.
    As the node referenced by the $newCurr$ variable has the same key, a node with a key $\infty$ is reachable from $n_1$ after $s$ -- a contradiction. 
    
    Therefore, if $n_1$ is not an infant, there exists a node with a key $\infty$ which is reachable from $n_1$.
    
    \item By the induction hypothesis, it suffices to show that whenever the $pred$, $curr$ and $succ$ variables are assigned with a new value (which is later dereferenced), it is not null.
    
    If $s$ assigns new values into the $pred$ and $curr$ variables in line~\ref{alg-insert-find},~\ref{alg-remove-find} or~\ref{alg-contains-find}, then this value is not null by the induction hypothesis.
    
    If $s$ assigns a new value into the $pred$ variable in line~\ref{alg-find-head}, then this value is a pointer to $head$, which is obviously not null.
    
    If $s$ assigns a new value into the $curr$ variable in line~\ref{alg-find-curr-gets-head-next}, then by Invariant~\ref{lemma-invariants-tail}, this value is not null.
    
    If $s$ assigns a new value into the $curr$ variable in line~\ref{alg-find-curr-gets-curr-next}, then by Lemma~\ref{lemma-deleted-states}, this value is the same value for which the condition in line~\ref{alg-find-if-tail} did not hold. I.e., this value is not null.
    
    If $s$ assigns a new value into the $pred$ variable in line~\ref{alg-find-pred-gets-curr}, then by the induction hypothesis, this value is not null.
    
    If $s$ assigns a new value into the $curr$ variable in line~\ref{alg-find-curr-gets-pred-next} or~\ref{alg-find-curr-gets-pred-next-after-trim}, then since the condition, checked in line~\ref{alg-find-key-bigger}, does not hold for the node referenced by the $pred$ variable, and by Invariant~\ref{lemma-invariants-tail}, the assigned value is not null.
    
    If $s$ assigns a new value into the $curr$ variable in line~\ref{alg-trim-curr-gets-curr-next}, then the previous node, referenced by the $curr$ variable, is marked. I.e., it cannot have an $\infty$ key. By Invariant~\ref{lemma-invariants-tail}, the assigned value is not null.
    
    Finally, if $s$ assigns a new value into the $succ$ variable in line~\ref{alg-trim-read-succ}, then it is dereferenced iff it is not null.
    
\end{enumerate}
\end{proof}


Before setting linearization points, we still need to define the term of being logically in the list. Next, we are going to prove that (1) an \emph{insert()} operation logically inserts a new node iff there is no other node with the same key, logically in the list at the operation's linearization point, (2) a \emph{remove()} operation either logically removes a node, or does nothing in case there is no node with the input key, which is logically in the list at the operation's linearization point, (3) a \emph{contains()} returns true iff there exists a node with the operation's input key which is logically in the list at its linearization point, and (4) a \emph{rangeQuery()} operation returns the sequence of all nodes, having keys in the given range, and which are logically in the list during the operation's linearization point.

\begin{definition} \textbf{(logically in the list)} \label{definition-logically-in}
We say that a node $n$ is \textit{logically in the list} if it is not pending and not unlinked.
\end{definition}

According to Definition~\ref{definition-logically-in}, unreachable nodes may still be considered as logically in the list, as long as they are not unlinked (see Definition~\ref{definition-unlinked}). Allegedly, this may result in foiling linearizability, as a key may appear twice in the list: an unreachable (but not yet unlinked) node and a newly inserted one might have the same key. However, this scenario is impossible, as we prove in Lemmas~\ref{lemma-reachable-not-between},~\ref{lemma-logically-in-not-between} and~\ref{lemma-key-once} below.

\begin{lemma} \label{lemma-reachable-not-between}
Let $n_1$ and $n_2$ be two nodes, and let $k_1$ and $k_2$ be their keys, respectively. If both nodes are logically in the list, reachable, and $n_2$ is $n_1$'s successor, then there does not exist a different node $n_3$, with a key $k_1 \leq k_3 \leq k_2$, which is logically in the list.
\end{lemma}

\begin{proof}
Assume by contradiction that $n_3$ exists. 

By Definition~\ref{definition-reachability},~\ref{definition-reachable}, and Lemma~\ref{lemma-invariants} (Invariant~\ref{lemma-invariants-smaller-not-reachable}), it is impossible that all three nodes are reachable. Therefore, $n_3$ is not reachable.

By Lemma~\ref{lemma-invariants} (Invariant~\ref{lemma-invariants-still-reachable-before-trim}), $n_3$ has already been traversed during the loop in lines~\ref{alg-trim-while}-\ref{alg-trim-curr-gets-curr-next}, followed by a successful CAS execution in line~\ref{alg-trim-cas}. Let $n_4$ be the node referenced by $pred$ and let $n_5$ be the node referenced by $newCurr$ at this point. By Lemma~\ref{lemma-invariants} (Invariant~\ref{lemma-invariants-smaller-not-reachable}), $n_4$'s key is smaller than $n_3$'s key, and $n_5$'s key is at least $n_3$'s key.
Since $n_3$ is still logically in the list, by Definition~\ref{definition-unlinked}, $n_5$ is still pending. Therefore, by Lemma~\ref{lemma-invariants} (Invariant~\ref{lemma-invariants-succ-pending-reachable}), $n_4$ is reachable (and consequently, $n_5$ is reachable as well). 
Now, by Lemma~\ref{lemma-invariants} (Invariant~\ref{lemma-invariants-smaller-not-reachable}), $n_4$'s key must be smaller than $n_2$'s key, and $n_5$'s key must be bigger than $n_1$'s key. As all four nodes are reachable, it must hold that $n_4$ is $n_1$ and $n_5$ is $n_2$. I.e., $n_2$ is pending, and by Definition~\ref{definition-logically-in}, is not logically in the list -- a contradiction. Therefore, $n_3$ does not exist.
\end{proof}

\begin{lemma} \label{lemma-logically-in-not-between}
Let $n_1$ and $n_2$ be two nodes, and let $k_1$ and $k_2$ be their keys, respectively. If both nodes are logically in the list, and $n_2$ is $n_1$'s successor, then there does not exist a different node $n_3$, with a key $k_1 \leq k_3 \leq k_2$, which is logically in the list.
\end{lemma}

\begin{proof}
Assume by contradiction that $n_3$ exists. By Lemma~\ref{lemma-reachable-not-between}, if $n_1$ is reachable, then we are done (as $n_2$ is also reachable, by Definition~\ref{definition-reachable}). Therefore, $n_1$ is not reachable.

By Lemma~\ref{lemma-invariants} (Invariant~\ref{lemma-invariants-still-reachable-before-trim}), $n_1$ has already been traversed during the loop in lines~\ref{alg-trim-while}-\ref{alg-trim-curr-gets-curr-next}, followed by a successful CAS execution in line~\ref{alg-trim-cas}. 

By Lemma~\ref{lemma-invariants} (Invariant~\ref{lemma-invariants-active-reachable}), $n_1$ is either marked or flagged. Therefore, by Lemma~\ref{lemma-deleted-states}, $n_2$ was also $n_1$'s successor right before it stopped being reachable.
Assume by contradiction that $n_3$ was logically in the list while $n_1$ was still reachable. Then by Lemma~\ref{lemma-reachable-not-between}, $n_2$ was not logically in the list when $n_1$ stopped being reachable (as $n_2$ necessarily was $n_1$'s successor and reachable at this point). In particular, $n_2$ was pending, and by Definition~\ref{definition-node-states}, it was neither marked nor flagged at this point, which means that it was not traversed during the loop in lines~\ref{alg-trim-while}-\ref{alg-trim-curr-gets-curr-next} (during which $n_1$ was traversed).
I.e., $n_2$ cannot be pending, as its timestamp was updated no later than the execution of line~\ref{alg-trim-update-succ-ts} -- a contradiction.
Therefore, $n_3$ started being logically in the list after $n_1$ stopped being reachable.
Moreover, by Lemma~\ref{lemma-invariants} (Invariant~\ref{lemma-invariants-active-reachable} and~\ref{lemma-invariants-not-reachable}), $n_3$ became reachable after $n_1$ stopped being reachable.

Going back to $n_1$'s traversal in lines~\ref{alg-trim-while}-\ref{alg-trim-curr-gets-curr-next}, guaranteed by Definition~\ref{definition-unlinked}, let $n_0$ be the node referenced by $pred$, let $n_4$ be the node referenced by $newCurr$, and let $k_0$ and $k_4$ be their keys, respectively. By Lemma~\ref{lemma-invariants} (Invariant~\ref{lemma-invariants-smaller-not-reachable}), $k_0 < k_1 \leq k_4$. Since $n_1$ is still logically in the list, by Definition~\ref{definition-unlinked} and~\ref{definition-logically-in}, $n_4$ is pending.
By Lemma~\ref{lemma-invariants} (Invariant~\ref{lemma-invariants-active-reachable}) $n_4$ became reachable after the successful CAS in line~\ref{alg-trim-cas} was executed. Since $n_1$ was reachable prior to this CAS (by Lemma~\ref{lemma-invariants}, Invariant~\ref{lemma-invariants-still-reachable-before-trim}), $n_4$ became reachable before $n_3$.

By Lemma~\ref{lemma-invariants} (Invariant~\ref{lemma-invariants-succ-pending-reachable}), both $n_0$ and $n_4$ are still reachable. 
As $n_3$ is currently logically in the list, by Lemma~\ref{lemma-invariants} (Invariant~\ref{lemma-invariants-still-reachable-before-trim}), it must have been reachable at some point. Since $k_0 < k_1 \leq k_3$, by Lemma~\ref{lemma-invariants} (Invariant~\ref{lemma-invariants-smaller-not-reachable}), while being reachable, $n_3$ was reachable from $n_4$. However, $n_3 \neq n_4$, as $n_3$ is currently logically in the list and $n_4$ is still pending. By Lemma~\ref{lemma-invariants} (Invariant~\ref{lemma-invariants-smaller-not-reachable}), $k_4 < k_3$. Consequently, $k_4 < k_2$. This means that $n_2$ was not traversed during the loop in lines~\ref{alg-trim-while}-\ref{alg-trim-curr-gets-curr-next} (as all traversed nodes must have a key which is not bigger than $k_4$, by Lemma~\ref{lemma-invariants}, Invariant~\ref{lemma-invariants-smaller-not-reachable}). Actually, Since $n_2$ is $n_1$'s successor, it must have been referenced by the $succ$ variable, which means that it was $n_4$'s successor when line~\ref{alg-trim-cas} was executed. By Lemma~\ref{lemma-invariants} (Invariant~\ref{lemma-invariants-pending-still-successor}), $n_2$ is still $n_4$'s successor. By Definition~\ref{definition-reachability}, $n_3$ is reachable from $n_2$ -- a contradiction to Lemma~\ref{lemma-invariants} (Invariant~\ref{lemma-invariants-smaller-not-reachable}).
Therefore, $n_3$ does not exist.
\end{proof}

\begin{lemma} \label{lemma-key-once}
Let $n_1$ and $n_2$ be two different nodes with the same key. Then either $n_1$ or $n_2$ is not logically in the list.
\end{lemma}

\begin{proof}
Assume by contradiction that both $n_1$ and $n_2$ are logically in the list. 
If either $n_1$ or $n_2$ are reachable, then we get a contradiction to Lemma~\ref{lemma-reachable-not-between}.
Otherwise, if either $n_1$ or $n_2$ have a predecessor which is logically in the list, we get a contradiction to Lemma~\ref{lemma-logically-in-not-between}.
Therefore, both are not reachable, and do not have a predecessor which is logically in the list.

By Lemma~\ref{lemma-invariants} (Invariant~\ref{lemma-invariants-still-reachable-before-trim}), both nodes have already been traversed during the loop in lines~\ref{alg-trim-while}-\ref{alg-trim-curr-gets-curr-next}, followed by a successful CAS execution in line~\ref{alg-trim-cas}. However, they were not traversed during the same \emph{trim()} invocation, as that would imply they were reachable from each other (a contradiction to Lemma~\ref{lemma-invariants}, Invariant~\ref{lemma-invariants-smaller-not-reachable}).
W.l.o.g., assume that the CAS, executed in line~\ref{alg-trim-cas} in the scope of the \emph{trim()} invocation, associated with $n_1$, was executed first.
Let $n_0$ be the node referenced by $pred$ and let $n_3$ be the node referenced by $newCurr$ at this point. 
As $n_1$ is still logically in the list, by Definition~\ref{definition-unlinked}, $n_3$ is pending. Therefore, by Lemma~\ref{lemma-invariants} (Invariant~\ref{lemma-invariants-succ-pending-reachable} and~\ref{lemma-invariants-pending-still-successor}), both $n_0$ and $n_3$ are still reachable, and $n_3$ is still $n_0$'s successor.
By Lemma~\ref{lemma-invariants} (Invariant~\ref{lemma-invariants-smaller-not-reachable}), $n_0$'s key is smaller than $n_1$'s key, and $n_3$'s key is at least $n_1$'s key.

In a similar way to $n_1$, by Lemma~\ref{lemma-invariants} (Invariant~\ref{lemma-invariants-still-reachable-before-trim}), $n_2$ has already been traversed during the loop in lines~\ref{alg-trim-while}-\ref{alg-trim-curr-gets-curr-next}, followed by a successful CAS execution in line~\ref{alg-trim-cas}. Let $n_4$ be the node referenced by $pred$ and let $n_5$ be the node referenced by $newCurr$ at this point. 
Since $n_1$ is logically in the list, by Definition~\ref{definition-unlinked}, $n_5$ is pending. Therefore, by Lemma~\ref{lemma-invariants} (Invariant~\ref{lemma-invariants-succ-pending-reachable} and~\ref{lemma-invariants-pending-still-successor}), both $n_4$ and $n_5$ are reachable, and $n_5$ is $n_4$'s successor.
By Lemma~\ref{lemma-invariants} (Invariant~\ref{lemma-invariants-smaller-not-reachable}), $n_4$'s key is smaller than $n_2$'s key, and $n_5$'s key is at least $n_2$'s key.

By Lemma~\ref{lemma-invariants} (Invariant~\ref{lemma-invariants-still-reachable-before-trim}), $n_0$ and $n_4$ are the same node, and $n_3$ and $n_5$ are the same node -- a contradiction to the fact that $n_1$ and $n_2$ were traversed during separate \emph{trim()} invocations.

Therefore, for every two different nodes, $n_1$ and $n_2$, with the same key, either $n_1$ or $n_2$ is not logically in the list.
\end{proof}

\subsection{Linearizability} \label{sec-linearization-points}

Before handling the list operations, we set a linearization point per \emph{find()} execution, as the rest of the operations heavily rely on this method. This point is guaranteed by Lemma~\ref{lemma-find-linearization} below.

\begin{lemma} \label{lemma-find-linearization}
Let $n_1$, $n_2$ be the two nodes returned from a \emph{find(key)} call. Then:
\begin{enumerate}
    \item $n_1$'s key is smaller than the input key.
    \item $n_2$'s key is not smaller than the input key
    \item There exists a point during the method execution in which both nodes are reachable and logically in the list and $n_2$ is $n_1$'s successor.
    \item During the guaranteed point, there does not exist a node $n_3 \neq n_2$, with a key bigger than $n_1$'s key and not bigger than $n_2$'s key, and which is logically in the list.
\end{enumerate}
\end{lemma}

\begin{proof}
Let $n_1$ and $n_2$ be the two nodes returned from a \emph{find(key)} call.
\begin{enumerate}
    \item A pointer to $n_1$ was assigned into the $pred$ variable, either in line~\ref{alg-find-head} or~\ref{alg-find-pred-gets-curr}. If it was assigned in line~\ref{alg-find-head}, then $n_1$ is the head sentinel and its key is obviously smaller than the input key. Otherwise, the condition in line~\ref{alg-find-key-bigger}, checked right before the assignment, did not hold for $n_1$. I.e., $n_1$'s key must be smaller than the input key.  

    \item A pointer to $n_2$ was assigned into the $curr$ variable, either in line~\ref{alg-find-curr-gets-head-next},~\ref{alg-find-curr-gets-curr-next},~\ref{alg-find-curr-gets-pred-next}, or~\ref{alg-find-curr-gets-pred-next-after-trim}.
    If the last such assignment was in line~\ref{alg-find-curr-gets-head-next},~\ref{alg-find-curr-gets-curr-next} or~\ref{alg-find-curr-gets-pred-next}, then the condition checked in line~\ref{alg-find-key-bigger} holds for $n_2$. Otherwise, the condition checked in line~\ref{alg-find-key-bigger-after-trim} does not hold for $n_2$. In both cases, $n_2$'s key is not smaller than the input key.

    \item We are now going to set the guaranteed point, with respect to the last time $n_2$ is assigned into the $curr$ variable.

    If $n_2$ is assigned into the $curr$ variable for the last time in line~\ref{alg-find-curr-gets-head-next}, then $n_1$ is the head sentinel node, which is not pending and always reachable (by Definition~\ref{definition-reachable}), and thus (By Lemma~\ref{lemma-invariants}, Invariant~\ref{lemma-invariants-unlinked}), is logically in the list. 
    In addition, $n_2$ is neither marked nor flagged at this point, as the condition checked in line~\ref{alg-find-while-is-marked} does not hold.
    If $n_2$ is active when line~\ref{alg-find-head-next} is executed, then it is reachable by Lemma~\ref{lemma-invariants} (Invariant~\ref{lemma-invariants-active-reachable}), and by the same lemma (Invariant~\ref{lemma-invariants-unlinked}), it is not unlinked. By Definition~\ref{definition-logically-in}, it is logically in the list at this point, and we are done. 
    Otherwise, by Lemma~\ref{lemma-invariants} (Invariant~\ref{lemma-invariants-pending-still-successor}), it is guaranteed that when $n_2$ becomes active, it is still $n_1$'s successor. I.e., it is reachable, thus neither an infant (by Lemma~\ref{lemma-invariants}, Invariant~\ref{lemma-invariants-infant}) nor unlinked (by Lemma~\ref{lemma-invariants}, Invariant~\ref{lemma-invariants-unlinked}), and by Definition~\ref{definition-logically-in}, it is logically in the list.
    Therefore, \textbf{if $n_2$ is assigned into the $curr$ variable for the last time in line~\ref{alg-find-curr-gets-head-next}, then the guaranteed point is the first time $n_2$ is active after the last execution of line~\ref{alg-find-head-next}}.

    The next scenario is when $n_2$ is assigned into the $curr$ variable for the last time in line~\ref{alg-find-curr-gets-pred-next}. Since the condition from line~\ref{alg-find-pred-marked-flagged} does not hold, $n_1$ is either pending or active at this point. In addition, since the condition from line~\ref{alg-find-while-is-marked} does not hold for $n_2$ in the next loop iteration, $n_2$ is also pending or active.
    Therefore, when $n_1$'s successor is read in line~\ref{alg-find-get-pred-next}, by Lemma~\ref{lemma-invariants} (Invariant~\ref{lemma-invariants-active-reachable}), both nodes are reachable.
    If both nodes are active at this point, then by Lemma~\ref{lemma-invariants} (Invariant~\ref{lemma-invariants-infant} and~\ref{lemma-invariants-unlinked}) and Definition~\ref{definition-logically-in}, we are done.
    Otherwise, $n_1$ becomes active no later than the execution of line~\ref{alg-find-update-pred-ts}, and $n_2$ becomes active no later than the execution of line~\ref{alg-find-update-curr-ts}.
    By Lemma~\ref{lemma-invariants} (Invariant~\ref{lemma-invariants-pending-still-successor}), once the latter becomes active, $n_2$ is still $n_1$'s successor. In order to show that both nodes are logically in the list at this point, by Lemma~\ref{lemma-invariants} (Invariant~\ref{lemma-invariants-infant} and~\ref{lemma-invariants-unlinked}) and Definition~\ref{definition-logically-in}, it still remains to show that both nodes are reachable at this point. If $n_1$ is the latter one becoming active, then by Lemma~\ref{lemma-invariants} (Invariant~\ref{lemma-invariants-active-reachable}) and the fact that $n_2$ is $n_1$'s successor, both nodes are reachable at this point.
    Otherwise, if $n_2$ is the latter one, then by Lemma~\ref{lemma-invariants} (Invariant~\ref{lemma-invariants-succ-pending-reachable}), $n_1$ is still reachable when $n_2$ becomes active. By Definition~\ref{definition-reachability} and~\ref{definition-reachable}, $n_2$ is reachable at this point as well. By Definition~\ref{definition-logically-in}, both nodes are logically in the list at this point.
    Therefore, \textbf{if $n_2$ is assigned into the $curr$ variable for the last time in line~\ref{alg-find-curr-gets-pred-next}, then the guaranteed point is the first time both $n_1$ and $n_2$ are not pending after the last execution of line~\ref{alg-find-get-pred-next}}.

    The last scenario is when $n_2$ is assigned into the $curr$ variable for the last time in line~\ref{alg-find-curr-gets-pred-next-after-trim}. In this scenario, $n_1$'s next pointer is read for the last time in line~\ref{alg-find-pred-next-after-trim}. Since the conditions in lines~\ref{alg-find-pred-marked-flagged-after-trim} and~\ref{alg-find-key-bigger-after-trim} do not hold, both $n_1$ and $n_2$ are not marked or flagged. Moreover, $n_1$ is not pending, as its timestamp is updated no later than the execution of line~\ref{alg-find-update-pred-ts}.
    In addition, by Lemma~\ref{lemma-invariants}, both are not infants (by Invariant~\ref{lemma-invariants-infant}) and thus, are reachable (by Invariant~\ref{lemma-invariants-active-reachable}).
    By Lemma~\ref{lemma-invariants} (Invariant~\ref{lemma-invariants-succ-pending-reachable}), $n_1$ is still reachable once $n_2$ becomes active.
    By Definition~\ref{definition-reachability} and~\ref{definition-reachable}, $n_2$ is reachable at this point as well.
    By Lemma~\ref{lemma-invariants} (Invariant~\ref{lemma-invariants-infant} and~\ref{lemma-invariants-unlinked}) and Definition~\ref{definition-logically-in}, both nodes are logically in the list at this point.
    Therefore, \textbf{if $n_2$ is assigned into the $curr$ variable for the last time in line~\ref{alg-find-curr-gets-pred-next-after-trim}, then the guaranteed point is the first time $n_2$ is not pending after the last execution of line~\ref{alg-find-pred-next-after-trim}}.

    Note that $n_2$ cannot be assigned into the $curr$ variable for the last time in line~\ref{alg-find-curr-gets-curr-next}. By Lemma~\ref{lemma-invariants} (Invariant~\ref{lemma-invariants-smaller-not-reachable}), and since the executing thread traverses at least one node during the loop in lines~\ref{alg-find-while-is-marked}-\ref{alg-find-curr-gets-curr-next}, $n_2$'s key is strictly bigger than the key of the node, referenced be the $predNext$ variable. Therefore, they cannot be equal, and the condition from line~\ref{alg-find-if-not-adjacent} holds. I.e., a node pointer is written into the $curr$ variable, either in line~\ref{alg-find-curr-gets-pred-next-after-trim}, or in the next traversal trial (in case the method restarts before line~\ref{alg-find-curr-gets-pred-next-after-trim} is executed) -- a contradiction to the fact that a pointer to the second output parameter was assigned into the $curr$ variable for the last time in line~\ref{alg-find-curr-gets-curr-next}.
    
    \item Directly holds by Lemma~\ref{lemma-reachable-not-between}.

\end{enumerate}
\end{proof}

Given Lemma~\ref{lemma-find-linearization}, we are now able to set linearization points per list operation.

\subsubsection{The insert() Operation}

We separate between successful executions (returning NO\_VAL in line~\ref{alg-insert-return-no-val}) and unsuccessful executions (returning $curr$'s value in line~\ref{alg-insert-return-curr-val}).

\begin{definition} \label{definition-successful-insert-linearization} \textbf{(Successful inserts - linearization points)}
Let $n$ be the node referenced by the $newNode$ variable when a successful \emph{insert()} execution returns. Then the operation's linearization point is set to the step that turns $n$ into an active node.
\end{definition}

\begin{definition} \label{definition-unsuccessful-insert-linearization} \textbf{(Unsuccessful inserts - linearization points)}
An unsuccessful \emph{insert()} execution is linearized at the point guaranteed by Lemma~\ref{lemma-find-linearization}, with respect to the last \emph{find()} execution (called in line~\ref{alg-insert-find}).
\end{definition}

We are now going to prove, using Lemmas~\ref{lemma-insert-linearization-between},~\ref{lemma-successful-insert-linearization-correctness}, and~\ref{lemma-unsuccessful-insert-linearization-correctness} below, that Definitions~\ref{definition-successful-insert-linearization}-\ref{definition-unsuccessful-insert-linearization} above indeed define adequate linearization points. I.e., that each linearization point occurs between the invocation and termination of the operation, and that the operation indeed takes affect at its linearization point.

\begin{lemma} \label{lemma-insert-linearization-between}
The linearization points, guaranteed by Definitions~\ref{definition-successful-insert-linearization} and~\ref{definition-unsuccessful-insert-linearization}, occur between the invocation and termination of the respective execution.
\end{lemma}

\begin{proof}
The linearization point guaranteed by Definition~\ref{definition-unsuccessful-insert-linearization} occurs in the scope of the \emph{find()} execution (by Lemma~\ref{lemma-find-linearization}), and consequently, it occurs between the invocation and response of the unsuccessful \emph{insert()} execution. 
The linearization point guaranteed by Definition~\ref{definition-successful-insert-linearization} occurs after the node's allocation (done in the scope of the execution), and no later than the execution of line~\ref{alg-insert-update-ts} (by Definition~\ref{definition-node-states}). Therefore, it occurs between the invocation and response of the successful \emph{insert()} execution.
\end{proof}

\begin{lemma} \label{lemma-successful-insert-linearization-correctness}
Let $n$ be the node referenced by the $newNode$ variable when a successful \emph{insert()} execution returns, let $k$ be its key, and let $s$ be the step guaranteed by Definition~\ref{definition-successful-insert-linearization}. Then: 
\begin{enumerate}
    \item Immediately after $s$, $n$ is logically in the list.
    \item Immediately before $s$, there is no node with a key $k$, which is logically in the list.
\end{enumerate}
\end{lemma}

\begin{proof}
Recall that $n$ becomes active at $s$ (and in particular, it is not pending).
\begin{enumerate}
    \item By Lemma~\ref{lemma-invariants} (Invariant~\ref{lemma-invariants-active-reachable}), $n$ is reachable immediately before and after $s$. Therefore, by Lemma~\ref{lemma-invariants} (Invariant~\ref{lemma-invariants-unlinked}), it is not unlinked, and by Definition~\ref{definition-logically-in}, it is logically in the list immediately after $s$.

    \item By Lemma~\ref{lemma-key-once}, immediately after $s$, there is no node with a key $k$, which is logically in the list. Assume by contradiction that there exists such a node, denoted $n_0$, immediately before $s$.
    I.e., by Definition~\ref{definition-logically-in}, $s$ unlinks $n_0$.
    By Definition~\ref{definition-unlinked}, $n$ was allocated during a \emph{trim()} execution -- a contradiction.
    Therefore, there is no node with a key $k$, which is logically in the list immediately before $s$.
\end{enumerate}
\end{proof}

\begin{lemma} \label{lemma-unsuccessful-insert-linearization-correctness}
Let $k$ be the input key to an unsuccessful \emph{insert()} execution.
Then at the linearization point, defined by Definition~\ref{definition-unsuccessful-insert-linearization}, there exists a node with a key $k$, which is logically in the list.
\end{lemma}

\begin{proof}
By Lemma~\ref{lemma-find-linearization}, the node referenced by the $curr$ variable is indeed logically in the list at this point. In addition, since the operation returns in line~\ref{alg-insert-return-curr-val}, its key is indeed $k$.
\end{proof}

\subsubsection{The remove() Operation}

We separate between successful executions (returning $curr$'s value in line~\ref{alg-remove-return-curr-val}), and unsuccessful executions (returning NO\_VAL in line~\ref{alg-remove-return-no-val}).

\begin{definition} \label{definition-successful-remove-linearization} \textbf{(Successful removals - linearization points)}
Let $n$ be the node referenced by the $curr$ variable when a successful \emph{remove()} execution returns. Then the operation's linearization point is set to the step that unlinks $n$.
\end{definition}

\begin{definition} \label{definition-unsuccessful-remove-linearization} \textbf{(Unsuccessful removals - linearization points)}
An unsuccessful \emph{remove()} execution is linearized at the point guaranteed by Lemma~\ref{lemma-find-linearization}, with respect to the last \emph{find()} execution (called in line~\ref{alg-remove-find}).
\end{definition}

We are now going to prove, using Lemmas~\ref{lemma-remove-linearization-between},~\ref{lemma-successful-remove-linearization-correctness}, and~\ref{lemma-unsuccessful-remove-linearization-correctness} below, that Definitions~\ref{definition-successful-remove-linearization}-\ref{definition-unsuccessful-remove-linearization} above indeed define adequate linearization points. I.e., that each linearization point occurs between the invocation and termination of the operation, and that the operation indeed takes affect at its linearization point.

\begin{lemma} \label{lemma-remove-linearization-between}
The linearization points, guaranteed by Definitions~\ref{definition-successful-remove-linearization} and~\ref{definition-unsuccessful-remove-linearization}, occur between the invocation and termination of the respective execution.
\end{lemma}

\begin{proof}
The linearization point guaranteed by Definition~\ref{definition-unsuccessful-remove-linearization} occurs in the scope of the \emph{find()} execution (by Lemma~\ref{lemma-find-linearization}), and consequently, it occurs between the invocation and response of the unsuccessful \emph{remove()} execution. 

We are now going to show that at the linearization point, guaranteed by Definition~\ref{definition-successful-remove-linearization}, occurs no later then the termination of the \emph{find()} method, called in line~\ref{alg-remove-trim}.
Let $n_1$,$n_2$ be the two nodes, returned from this \emph{find()} execution (and not used by the calling \emph{remove()} operation).
By Lemma~\ref{lemma-find-linearization}, $n_1$'s key is smaller than $n$'s key, $n_2$'s key is at least $n$'s key, and there exists a point during the \emph{find()} execution, in which $n_2$ is $n_1$'s successor, and both nodes are reachable and logically in the list.

First, assume by contradiction that $n_2$ is $n$. As $n$ is already marked when the \emph{find()} method is invoked, it is impossible that the \emph{find()} method returns it (see the conditions in lines~\ref{alg-find-while-is-marked} and~\ref{alg-find-key-bigger-after-trim}) -- a contradiction. Therefore, by Lemma~\ref{lemma-invariants} (Invariant~\ref{lemma-invariants-smaller-not-reachable}), $n$ is already not reachable at the point, guaranteed by Lemma~\ref{lemma-find-linearization}.
By Lemma~\ref{lemma-find-linearization}, $n$ is necessarily not logically in the list at this point.
Therefore, both linearization points occur between the invocation and response of the \emph{remove()} execution.
\end{proof}

\begin{lemma} \label{lemma-successful-remove-linearization-correctness}
Let $n$ be the node, marked during a successful \emph{remove()} execution, let $k$ be its key, and let $s$ be the step guaranteed by Definition~\ref{definition-successful-remove-linearization}. Then: 
\begin{enumerate}
    \item $n$ is not marked during any other \emph{remove()} execution.
    \item After $s$, $n$ is not logically in the list.
    \item Immediately before $s$, $n$ is logically in the list.
\end{enumerate}
\end{lemma}

\begin{proof}
Recall that a node is necessarily marked during a successful \emph{remove()} operation (in line~\ref{alg-remove-mark}).
\begin{enumerate}
    \item By Lemma~\ref{lemma-deleted-states}, a node can only be marked once. Therefore, $n$ is not marked during any other \emph{remove()} execution.

    \item As $s$ unlinks $n$, by Definition~\ref{definition-logically-in}, immediately after $s$, $n$ is not logically in the list. Moreover, at any point after $s$, $n$ remains unlinked by Definition~\ref{definition-unlinked}, and thus, not logically in the list.
    
    \item Since $n$ is not unlinked before $s$, by Definition~\ref{definition-logically-in}, it remains to show that $n$ is not an infant immediately before $s$. By Definition~\ref{definition-unlinked}, $n$ must have a predecessor prior to $s$. Therefore, by Lemma~\ref{lemma-invariants} (Invariant~\ref{lemma-invariants-infant}), $n$ is not an infant before $s$.  
\end{enumerate}
\end{proof}

\begin{lemma} \label{lemma-unsuccessful-remove-linearization-correctness}
Let $k$ be the input key to an unsuccessful \emph{remove()} execution.
Then at the linearization point, defined by Definition~\ref{definition-unsuccessful-remove-linearization}, there does not exist a node with a key $k$, which is logically in the list.
\end{lemma}

\begin{proof}
By Lemma~\ref{lemma-find-linearization}, the node referenced by the $curr$ variable is logically in the list at this point. Since the condition in line~\ref{alg-remove-return-no-val} holds, by Lemma~\ref{lemma-find-linearization}, its key is bigger than $k$. By Lemma~\ref{lemma-find-linearization}, this node, along with the node referenced by the $pred$ variable, are logically in the list, and $k$ is between their respective keys. By Lemma~\ref{lemma-find-linearization}, there does not exist a node with a key $k$, which is logically in the list at this point.
\end{proof} 

\subsubsection{The contains() Operation}

As opposed to the \emph{insert()} and \emph{remove()} operations, there is no need to separate between \emph{contains()} invocations that return NO\_VAL in line~\ref{alg-contains-return-no-val} and an actual value in line~\ref{alg-contains-return-curr-val}. Both types of executions are linearized at the same point, as defined in Definition~\ref{definition-contains-linearization} below.

\begin{definition} \label{definition-contains-linearization} \textbf{(contains - linearization points)}
A \emph{contains()} execution is linearized at the point guaranteed by Lemma~\ref{lemma-find-linearization}, with respect to the \emph{find()} execution (called in line~\ref{alg-contains-find}).
\end{definition}

We are now going to prove, using Lemmas~\ref{lemma-contains-linearization-between} and~\ref{lemma-contains-linearization-correctness} below, that Definition~\ref{definition-contains-linearization} above indeed defines adequate linearization points. I.e., that each linearization point occurs between the invocation and termination of the operation, and that the operation indeed takes affect at its linearization point.

\begin{lemma} \label{lemma-contains-linearization-between}
The linearization points, guaranteed by Definition~\ref{definition-contains-linearization}, occur between the invocation and termination of the respective execution.
\end{lemma}

\begin{proof}
By Lemma~\ref{lemma-find-linearization}, the linearization point, guaranteed by Definition~\ref{definition-contains-linearization}, occurs in the scope of the \emph{find()} execution, called in line~\ref{alg-contains-find}. Consequently, it occurs between the invocation and response of the \emph{contains()} execution.
\end{proof}

\begin{lemma} \label{lemma-contains-linearization-correctness}
Let $k$ be the input key to a \emph{contains()} execution.
Then at the linearization point, defined by Definition~\ref{definition-contains-linearization}, it holds that:
\begin{enumerate}
    \item If the operation eventually returns in line~\ref{alg-contains-return-no-val}, then there does not exist a node with a key $k$, which is logically in the list.
    \item If the operation eventually returns in line~\ref{alg-contains-return-curr-val}, then there exists a node with a key $k$, which is logically in the list, and its value is returned in line~\ref{alg-contains-return-curr-val}. 
\end{enumerate}
\end{lemma}

\begin{proof}
Let $n_1$ and $n_2$ be the nodes, referenced by the $pred$ and $curr$ variables (respectively) when the operation returns.
By Lemma~\ref{lemma-find-linearization}, both $n_1$ and $n_2$ are logically in the list at this point. In addition, $n_1$'s key is smaller than $k$ and $n_2$'s key is at least $k$.
\begin{enumerate}
    \item Since the condition in line~\ref{alg-contains-return-no-val} holds, by Lemma~\ref{lemma-find-linearization}, $n_2$'s key is bigger than $k$. Therefore, by Lemma~\ref{lemma-find-linearization}, there does not exist a node with a key $k$, which is logically in the list at this point.
    \item By Lemma~\ref{lemma-find-linearization}, $n_2$ is logically in the list at this point. Since the condition in line~\ref{alg-contains-return-no-val} does not hold, $n_2$'s key is $k$. In addition, $n_2$'s value is value returned in line~\ref{alg-contains-return-curr-val}. 
\end{enumerate}
\end{proof} 

\subsubsection{The rangeQuery() Operation}

For proving that the \emph{rangeQuery()} operation is linearizable, we first need to prove some basic claims. As these claims are irrelevant to the other operations' linearizability, we are going to prove them below. 
For convenience, we use the term \textit{epoch} when relating to the global timestamps clock's value. In addition, we are going to define the terms of birth epoch and retire epoch in this context\footnote{Note that to this end, we ignore memory reclamation. Therefore, there is no problem with overloading the birth epoch term.}.

\begin{definition} \textbf{(birth epoch and retire epoch)} \label{definition-birth-retire-epoch}
Let $n$ be a node. We say that $n$'s \textit{birth epoch} is the first epoch in which $n$ was logically in the list, and that $n$'s \textit{retire epoch} is the first epoch in which $n$ stopped being logically in the list.
\end{definition}

\begin{lemma} \label{lemma-birth-epoch}
Suppose that $n$'s timestamp is $t$, and that $t \neq \bot$. Then $t$ is $n$'s birth epoch. 
\end{lemma}

\begin{proof}
By Definition~\ref{definition-node-states} and~\ref{definition-logically-in}, $n$ became logically in the list when $t$ was read from the global epoch clock. I.e., by Definition~\ref{definition-birth-retire-epoch}, $t$ is indeed $n$'s birth epoch.
\end{proof}

\begin{lemma} \label{lemma-retire-epoch}
Suppose that $n$ is an unlinked node, and let $n_2$ be the node guaranteed by Definition~\ref{definition-unlinked}. Then $n_2$'s timestamp is $n$'s retire epoch. 
\end{lemma}

\begin{proof}
By Definition~\ref{definition-node-states} and~\ref{definition-unlinked}, $n$ was unlinked from the list when $n_2$'s timestamp was read from the global epoch clock. I.e., by Definition~\ref{definition-birth-retire-epoch}, $n_2$'s timestamp is indeed $n$'s retire epoch.
\end{proof}

\begin{lemma} \label{lemma-continuous-succ}
Let $n_1$ and $n_2$ be two nodes, and let $t_1$ and $t_2$ be their birth epochs, respectively. Let $t_{max}$ be the maximum between $t_1$ and $t_2$. If $n_2$ is $n_1$'s successor during an epoch $t \geq t_{max}$, then $n_2$ has been $n_1$'s successor since epoch $t_{max}$.
\end{lemma}

\begin{proof}
By Definition~\ref{definition-infant} and Lemma~\ref{lemma-invariants} (Invariants~\ref{lemma-invariants-infant} and~\ref{lemma-invariants-pending-infant}), indeed $n_2$ started to be $n_1$'s successor when one of them was pending.

In addition, the other one was not pending:
\begin{enumerate}
    \item If $n_2$ started to be $n_1$'s successor in line~\ref{alg-insert-update-next} of Algorithm~\ref{pseudo-list}, then $n_1$ was pending (by Lemma~\ref{lemma-invariants}, Invariant~\ref{lemma-invariants-pending-infant}) and $n_2$ was not (its timestamp was updated no later than a former execution of line~\ref{alg-find-update-curr-ts} in Algorithm~\ref{pseudo-list}).
    
    \item If $n_2$ started to be $n_1$'s successor in line~\ref{alg-insert-update-pred} of Algorithm~\ref{pseudo-list}, then $n_2$ was pending (by Lemma~\ref{lemma-invariants}, Invariant~\ref{lemma-invariants-pending-infant}) and $n_1$ was not (its timestamp was updated no later than a former execution of line~\ref{alg-find-update-pred-ts} in Algorithm~\ref{pseudo-list}).
    
    \item If $n_2$ started to be $n_1$'s successor in line~\ref{alg-trim-init-next} of Algorithm~\ref{pseudo-list}, then $n_1$ was pending (by Lemma~\ref{lemma-invariants}, Invariant~\ref{lemma-invariants-pending-infant}) and $n_2$ was not (its timestamp was updated no later than the execution of line~\ref{alg-trim-update-succ-ts} in Algorithm~\ref{pseudo-list}).
    
    \item If $n_2$ started to be $n_1$'s successor in line~\ref{alg-trim-cas} of Algorithm~\ref{pseudo-list}, then $n_2$ was pending (by Lemma~\ref{lemma-invariants}, Invariant~\ref{lemma-invariants-pending-infant}) and $n_1$ was not (its timestamp was updated no later than a former execution of line~\ref{alg-find-update-pred-ts} in Algorithm~\ref{pseudo-list}).
\end{enumerate}

By Lemma~\ref{lemma-invariants} (Invariant~\ref{lemma-invariants-pending-still-successor}), $n_2$ remained $n_1$'s successor until both were not pending.
After this point, as explained above, they could not have re-become successive nodes again. Therefore, $n_2$ has been $n_1$'s successor since epoch $t_{max}$.
\end{proof}

\begin{definition} \textbf{(p-predecessor and p-successor)} \label{definition-p-successor}
We say that a node $n_2$ is the \textit{p-successor} of a node $n_1$, if $n_1$'s $prior$ pointer points to $n_2$. In this case, we also say that $n_1$ is a \textit{p-predecessor} of $n_2$.
\end{definition}

\begin{lemma} \label{lemma-p-successor}
Let $n_1$ and $n_2$ be two nodes, let $t_1$ and $t_2$ be their timestamps, and let $k_1$ and $k_2$ be their keys, respectively.
If $n_2$ is $n_1$'s p-successor, then:
\begin{enumerate}
    \item $t_2 \neq \bot$.
    \item If $t_1 \neq \bot$, then $t_2 \leq t_1$.
    \item If $t_1 \neq \bot$, then $n_2$'s retire epoch is at least $t_1$.
    \item If $t_1 \neq \bot$ and $n_2$'s key is smaller than $n_1$'s key, then $n_2$'s retire epoch is $t_1$.
    \item Let $n_0$ be $n_1$'s first predecessor, let $t_0$ be its birth epoch, and let $k_0$ be its key. In addition, let $t_{max}$ be the maximum between $t_0$ and $t_2$. Then:
    \begin{enumerate}
        \item $n_2$ became $n_0$'s predecessor no later than epoch $t_{max}$.
        \item $n_2$ had been $n_0$'s successor until $n_1$ was its successor.
        \item If either $t_1$ is $\bot$, or both $t_1 \neq \bot$ and $t_1 > t_{max}$, then since the global epoch clock showed $t_{max}$ for the last time, and while $n_1$ was not logically in the list yet, no node with a key between $k_0$ and $k_2$ had been logically in the list.
    \end{enumerate}
\end{enumerate}
\end{lemma}

\begin{proof}
Recall that by Lemma~\ref{lemma-immutable-ts}, timestamps are immutable once set. In addition, $prior$ fields are set upon initialization (either in line~\ref{alg-insert-update-prior} or~\ref{alg-trim-update-prior}), and are immutable as well.
\begin{enumerate}
    \item It suffices to show that $t_2 \neq \bot$ when $n_1$'s $prior$ pointer was initialized. If $n_2$ became $n_1$'s p-successor in line~\ref{alg-insert-update-prior}, then a pointer to $n_2$ was returned as the second output parameter from the \emph{find()} execution in line~\ref{alg-insert-find}. Therefore, $n_2$'s timestamp stopped being $\bot$ no later than the respective execution of line~\ref{alg-find-update-curr-ts}. 
    
    Otherwise, $n_2$ became $n_1$'s p-successor in line~\ref{alg-trim-update-prior}. In this case, $n_2$ is marked, and by Lemma~\ref{lemma-node-states}, $t_2 \neq \bot$.
    
    \item When $n_1$'s $prior$ is set, $t_2 \neq \bot$. In addition, by Lemma~\ref{lemma-invariants} (Invariant~\ref{lemma-invariants-pending-infant}), $t_1 = \bot$ at this stage. As the global epoch clock can only increase over time, if $t_1 \neq \bot$, then $t_2 \leq t_1$.
    
    \item If $n_2$ became $n_1$'s p-successor in line~\ref{alg-insert-update-prior}, then by Definition~\ref{definition-reachability} and Lemma~\ref{lemma-invariants} (Invariants~\ref{lemma-invariants-pending-infant},~\ref{lemma-invariants-active-reachable} and~\ref{lemma-invariants-pending-still-successor}), $n_2$ was still reachable when $n_1$ became active. By Lemma~\ref{lemma-invariants} (Invariant~\ref{lemma-invariants-unlinked}), $n_2$ was still logically in the list at this point. By Definition~\ref{definition-birth-retire-epoch}, $n_2$'s retire epoch is at least $t_1$.
    
    Otherwise, if $n_2$ became $n_1$'s p-successor in line~\ref{alg-trim-update-prior}, then by Lemma~\ref{lemma-retire-epoch}, $n_2$'s retire epoch is $t_1$.
    
    \item If $n_2$'s key is smaller than $n_1$'s key, then by Lemma~\ref{lemma-invariants} (Invariant~\ref{lemma-invariants-smaller-not-reachable}) and Lemma~\ref{lemma-find-linearization}, $n_2$ necessarily became $n_1$'s p-successor in line~\ref{alg-trim-update-prior}. By Lemma~\ref{lemma-retire-epoch}, $n_2$'s retire epoch is $t_1$ in this case.
    
    \item First, note that both $n_0$'s and $n_2$'s timestamps are not $\bot$: as $n_2$ is $n_1$'s p-successor, its timestamp is not $\bot$ (as already proven above). In addition, by Definition~\ref{definition-infant} and Lemma~\ref{lemma-invariants} (Invariants~\ref{lemma-invariants-infant} and~\ref{lemma-invariants-pending-infant}), $n_0$'s timestamp was updated no later than a former execution of line~\ref{alg-find-update-pred-ts}.
    
    Second, note that if $t_1 \neq \bot$, then $t_1 \geq t_{max}$ (as both $n_0$'s and $n_2$'s timestamps were updated during former \emph{find()} executions).
    
    \begin{enumerate}
        \item $n_2$ was $n_0$'s successor just before $n_1$. Therefore, by Lemma~\ref{lemma-birth-epoch} and~\ref{lemma-continuous-succ}, $n_2$ became $n_0$'s predecessor no later than epoch $t_{max}$.
        
        \item $n_2$ was $n_0$'s successor just before $n_1$. By Lemma~\ref{lemma-continuous-succ}, $n_2$ had been $n_0$'s successor until $n_1$ was its successor.
        \item Assume that $t_1 \neq \bot$ and $t_1 > t_{max}$. By Lemma~\ref{lemma-continuous-succ}, $n_2$ was $n_0$'s successor when the global epoch clock showed $t_{max}$ for the last time, and until $n_1$ became reachable.
        
        By Lemma~\ref{lemma-invariants} (Invariants~\ref{lemma-invariants-pending-infant} and~\ref{lemma-invariants-active-reachable}), $n_1$ became reachable either after the execution of line~\ref{alg-insert-update-pred} or~\ref{alg-trim-cas} in Algorithm~\ref{pseudo-list}.
        
        If $t_1$ became reachable after the execution of line~\ref{alg-insert-update-pred}, then both $n_0$ and $n_2$ were logically in the list at this point, and by Lemma~\ref{lemma-logically-in-not-between}, no node with a key between $k_0$ and $k_2$ had been logically in the list since the global epoch clock showed $t_{max}$ for the last time, and until $n_1$ became reachable. Assume by contradiction that since $n_1$ became reachable, and before it was logically in the list, there had been a logically in the list node, $n_3$, with a key $k_0 < k_3 < k_2$. By Definition~\ref{definition-logically-in}, $n_3$ is not $n_1$. In addition, by Lemma~\ref{lemma-invariants} (Invariant~\ref{lemma-invariants-smaller-not-reachable}), $n_3$ was not reachable since $n_1$ became reachable, and before it was logically in the list.
        Since $n_3$ was logically in the list, by Lemma~\ref{lemma-invariants} (Invariant~\ref{lemma-invariants-still-reachable-before-trim}), $n_3$ must have been reachable from $n_0$ while still being reachable. I.e, $n_3$ stopped being reachable before $n_2$ became $n_0$'s successor, and thus, by Lemma~\ref{lemma-reachable-not-between}, stooped being logically in the list before $n_1$ became reachable -- a contradiction.
        
        Therefore, $t_1$ became reachable after the execution of line~\ref{alg-trim-cas}. By Definition~\ref{definition-unlinked}, $n_2$ remained logically in the list until $n_1$ was logically in the list, and $k_1 \geq k_2$. Assume by contradiction that since $n_1$ became reachable, and before it was logically in the list, there had been a logically in the list node, $n_3$, with a key $k_0 < k_3 < k_2$. By Definition~\ref{definition-logically-in}, $n_3$ is not $n_1$. In addition, by Lemma~\ref{lemma-invariants} (Invariant~\ref{lemma-invariants-smaller-not-reachable}), $n_3$ was not reachable since $n_1$ became reachable, and before it was logically in the list.
        Since $n_3$ was logically in the list, by Lemma~\ref{lemma-invariants} (Invariant~\ref{lemma-invariants-still-reachable-before-trim}), $n_3$ must have been reachable from $n_0$ while still being reachable. I.e, $n_3$ stopped being reachable before $n_2$ became $n_0$'s successor, and thus, by Lemma~\ref{lemma-reachable-not-between}, stooped being logically in the list before $n_1$ became reachable -- a contradiction.
        
        Therefore, since the global epoch clock showed $t_{max}$ for the last time, and until $n_1$ was logically in the list for the first time, no node with a key between $k_0$ and $k_2$ had been logically in the list.
    \end{enumerate}
\end{enumerate}
\end{proof}

\begin{definition} \textbf{(p-reachability)} \label{definition-p-reachability}
We say that a node $m$ is \textit{p-reachable} from a node $n$ if there exist nodes $n_0,\ldots,n_k$ such that $n=n_0$, $m=n_k$, and for every $0 \leq i < k$, $n_{i+1}$ is $n_i$'s p-successor.
\end{definition}

\begin{lemma} \label{lemma-tail-p-reachable}
Let $n$ be a node, and assume that $n$'s $prior$ pointer is initialized. Then $tail$\footnote{The original tail sentinel node, and not just any node with an $\infty$ key} is p-reachable from $n$.
\end{lemma}

\begin{proof}
Recall that $head$'s prior is never initialized, so the lemma holds for it vacuously.
The lemma is proven by Induction on the execution length.
For the base case, $tail$ is p-reachable from itself by Definition~\ref{definition-p-reachability}.
For the induction step, It suffices to show that the lemma holds for $n$'s (single) p-successor. 
If $n$'s p-successor is $tail$, then we are done.

Otherwise, if $n$'s p-successor is set in line~\ref{alg-insert-update-prior}, then $n$'s p-successor is also its successor. By Lemma~\ref{lemma-invariants} (Invariant~\ref{lemma-invariants-smaller-not-reachable}), $n$'s successor is not $head$. In addition, by Lemma~\ref{lemma-invariants} (Invariant~\ref{lemma-invariants-pending-infant}), its prior pointer is already set. Therefore, by the induction hypothesis and Definition~\ref{definition-p-reachability}, $tail$ is p-reachable from $n$.

Otherwise, if $n$'s p-successor is set in line~\ref{alg-trim-update-prior}, then $n$'s p-successor is marked. I.e., it is not $head$, and by Lemma~\ref{lemma-invariants} (Invariant~\ref{lemma-invariants-pending-infant}), its prior pointer is already set. Therefore, by the induction hypothesis and Definition~\ref{definition-p-reachability}, $tail$ is p-reachable from $n$ in this case as well.
\end{proof}

Given Lemma~\ref{lemma-tail-p-reachable}, we can now prove that we never dereference a null pointer during the \emph{rangeQuery()} operation.

\begin{lemma} \label{lemma-range-query-never-null}
A null pointer is never assigned into the $pred$, $curr$ and $succ$ pointers during a \emph{rangeQuery()} execution.
\end{lemma}

\begin{proof}
We are going to prove the lemma by induction on the execution length. The lemma holds vacuously when the operation is invoked.
For the induction step:
\begin{itemize}
    \item The lemma holds after executing line~\ref{alg-range-query-find}, by Lemma~\ref{lemma-invariants} (Invariant~\ref{lemma-invariants-no-null-dereference}).
    
    \item The lemma holds after executing line~\ref{alg-range-query-pred-gets-prior}, by Lemma~\ref{lemma-tail-p-reachable}. Note that if a pointer to $tail$ is assigned into the $pred$ variable before the execution of  line~\ref{alg-range-query-pred-gets-prior}, then the loop breaks (as $tail$'s timestamp is necessarily not bigger than $ts$).
    
    \item The lemma holds after executing line~\ref{alg-range-query-curr-gets-pred}, by the induction hypothesis.
    
    \item The lemma holds after executing line~\ref{alg-range-query-succ-gets-curr-next-smaller}, by Lemma~\ref{lemma-invariants} (Invariant~\ref{lemma-invariants-tail}). Note that if a pointer to $tail$ is assigned into the $curr$ variable before the execution of  line~\ref{alg-range-query-pred-gets-prior}, then the loop breaks (as $tail$'s key is necessarily bigger than $low$).
    
    \item The lemma holds after executing line~\ref{alg-range-query-smaller-curr-succ-gets-prior}, by Lemma~\ref{lemma-tail-p-reachable}. Note that if a pointer to $tail$ is assigned into the $succ$ variable before the execution of line~\ref{alg-range-query-smaller-curr-succ-gets-prior}, then the loop breaks (as $tail$'s timestamp is necessarily not bigger than $ts$).
    
    \item The lemma holds after executing line~\ref{alg-range-query-smaller-curr-gets-succ}, by the induction hypothesis.
    
    \item The lemma holds after executing line~\ref{alg-range-query-succ-gets-curr-next-bigger}, by Lemma~\ref{lemma-invariants} (Invariant~\ref{lemma-invariants-tail}). Note that if a pointer to $tail$ is assigned into the $curr$ variable before the execution of  line~\ref{alg-range-query-succ-gets-curr-next-bigger}, then the loop breaks (as $tail$'s key is necessarily bigger than $high$).
    
    \item The lemma holds after executing line~\ref{alg-range-query-bigger-curr-succ-gets-prior}, by Lemma~\ref{lemma-tail-p-reachable}. Note that if a pointer to $tail$ is assigned into the $succ$ variable before the execution of line~\ref{alg-range-query-bigger-curr-succ-gets-prior}, then the loop breaks (as $tail$'s timestamp is necessarily not bigger than $ts$).
    
    \item The lemma holds after executing line~\ref{alg-range-query-bigger-curr-gets-succ}, by the induction hypothesis.
    
\end{itemize}
\end{proof}

\begin{corollary} \label{corollary-dereference}
Null pointers are never dereferenced during the \emph{rangeQuery()} operation.
\end{corollary}

\begin{lemma} \label{lemma-range-pred-first-loop}
Let $n$ be the node referenced by the $pred$ variable, when the loop in lines~\ref{alg-range-query-while-find}-\ref{alg-range-query-new-ts} of Algorithm~\ref{pseudo-range} breaks in (line~\ref{alg-range-query-break}). Then $n$ was logically in the list when epoch $ts$ ended.
\end{lemma}

\begin{proof}

First, note that $n$'s timestamp is not $\bot$, as it is either the $head$ sentinel node, returned as the first output parameter, from the \emph{find()} call in line~\ref{alg-range-query-find} (i.e., its timestamp was updated no later than the respective execution of line~\ref{alg-find-update-pred-ts}), or has a p-predecessor (and therefore, has a non-$\bot$ timestamp, by Lemma~\ref{lemma-p-successor}).

Now, as the condition, checked in line~\ref{alg-range-query-while-pred}, does not hold for $n$, $n$'s timestamp is at most $ts$. By Lemma~\ref{lemma-birth-epoch}, $n$ was logically in the list before epoch $ts$ ended. 

If $n$ is the $head$ sentinel node, then it obviously was logically in the list when epoch $ts$ ended.

Otherwise, if a pointer to $n$ was returned as the first output parameter, from the \emph{find()} call in line~\ref{alg-range-query-find} (i.e., the condition in line~\ref{alg-range-query-while-pred} was checked only once), then by Lemma~\ref{lemma-find-linearization}, $n$ was still logically in the list during an epoch which is bigger than $ts$ (as $ts$ is always smaller than the current epoch). Therefore, $n$ was logically in the list when epoch $ts$ ended.

Otherwise, $n$ has a p-predecessor with a timestamp which is bigger than $ts$ (as the condition in line~\ref{alg-range-query-while-pred} held for this p-predecessor). By Lemma~\ref{lemma-p-successor}, $n$'s retire epoch is bigger then $ts$. Therefore, $n$ was logically in the list when epoch $ts$ ended.
\end{proof}

\begin{lemma} \label{lemma-succ-curr-succ}
Let $n_2$ be the last node assigned into the $succ$ variable during a loop iteration in lines~\ref{alg-range-query-while-curr-smaller}-\ref{alg-range-query-smaller-curr-gets-succ}, and let $n_1$ be the node referenced by the $curr$ variable during this assignment.
In addition, let $k_1$ be $n_1$'s key, and let $k_2$ be $n_2$'s key.
Then when epoch $ts$ ended:
\begin{enumerate}
    \item Both nodes were logically in the list.
    \item There was no node with a key $k_1 < k_3 < k_2$, that was logically in the list when epoch $ts$ ended.
\end{enumerate}
\end{lemma}

\begin{proof}
We are going to prove the lemma by induction on the number of loop iterations.
For the base case, before the loop started, by Lemma~\ref{lemma-range-pred-first-loop}, $n_1$ was logically in the list when epoch $ts$ ended, and the rest holds vacuously.

For the induction step, first assume that a pointer to $n_2$ was assigned into the $succ$ variable in line~\ref{alg-range-query-succ-gets-curr-next-smaller}.
By the induction hypothesis, $n_1$ was logically in the list when epoch $ts$ ended (either by Lemma~\ref{lemma-range-pred-first-loop}, or by the induction hypothesis). In addition, as the condition from line~\ref{alg-range-query-smaller-curr-while-succ} did not hold for $n_2$, by Lemma~\ref{lemma-birth-epoch}, $n_2$'s birth epoch is at most $ts$. Since $n_2$ is $n_1$'s successor during the current epoch, which is obviously bigger than $ts$, by Lemma~\ref{lemma-continuous-succ}, $n_2$ was also $n_1$'s successor when epoch $ts$ ended. Assume by contradiction that $n_2$ was not logically in the list when epoch $ts$ ended. Then by Lemma~\ref{lemma-invariants} (Invariant~\ref{lemma-invariants-unlinked}), it was also not reachable. By Definition~\ref{definition-reachability}, $n_1$ was not reachable as well at this point. By Definition~\ref{definition-unlinked}, $n_1$ became unlinked when $n_2$ became unlinked -- a contradiction.
Therefore, $n_2$ was also logically in the list when epoch $ts$ ended. By Lemma~\ref{lemma-logically-in-not-between}, there was no node with a key $k_1 < k_3 < k_2$, that was logically in the list when epoch $ts$ ended.

Now, assume that a pointer to $n_2$ was assigned into the $succ$ variable in line~\ref{alg-range-query-smaller-curr-succ-gets-prior}. Then since the condition in line~\ref{alg-range-query-smaller-curr-while-succ} did not hold for $n_2$, $n_2$'s birth epoch is at most $ts$.
By the induction hypothesis, $n_1$ was logically in the list when epoch $ts$ ended. Specifically, $n_1$'s birth epoch is at most $ts$ as well.

Let $n_3$ be the previous node referenced by the $succ$ variable. Then $n_3$'s timestamp is bigger than $ts$. 
By Lemma~\ref{lemma-p-successor}, $n_2$ had been $n_1$'s successor when epoch $ts$ ended, and there was no node with a key $k_1 < k_3 < k_2$, that was logically in the list when epoch $ts$ ended.

It still remains to show that $n_2$ was logically in the list when epoch $ts$ ended.
Assume by contradiction that $n_2$ was not logically in the list when epoch $ts$ ended. Then by Lemma~\ref{lemma-invariants} (Invariant~\ref{lemma-invariants-unlinked}), it was also not reachable. By Definition~\ref{definition-reachability}, $n_1$ was not reachable as well at this point. By Lemma~\ref{lemma-invariants} (Invariant~\ref{lemma-invariants-active-reachable}), it was neither pending nor active -- a contradiction to Lemma~\ref{lemma-deleted-states} (as its next pointer was later updated to point to $n_3$).
Therefore, $n_2$ was logically in the list when epoch $ts$ ended.
\end{proof}

\begin{lemma} \label{lemma-range-curr-smaller-loop}
Let $n_1$ and $n_2$ be two nodes, referenced by the $curr$ variable consecutively, during the loop in lines~\ref{alg-range-query-while-curr-smaller}-\ref{alg-range-query-smaller-curr-gets-succ}, and let $k_1$ and $k_2$ be their respective keys. Then when epoch $ts$ ended, both nodes were logically in the list. In addition, for every node $n_3$ with a key $k_1 < k_3 < k_2$, $n_3$ was not logically in the list when epoch $ts$ ended.
\end{lemma}

\begin{proof}
The lemma holds directly from Lemma~\ref{lemma-succ-curr-succ}.
\end{proof}

\begin{lemma} \label{lemma-range-curr-bigger-loop}
Let $n_1$ and $n_2$ be two nodes, referenced by the $curr$ variable consecutively, during the loop in lines~\ref{alg-range-query-while-curr-bigger}-\ref{alg-range-query-bigger-curr-gets-succ}, and let $k_1$ and $k_2$ be their respective keys. Then when epoch $ts$ ended, both nodes were logically in the list. In addition, for every node $n_3$ with a key $k_1 < k_3 < k_2$, $n_3$ was not logically in the list when epoch $ts$ ended.
\end{lemma}

\begin{proof}
W.l.o.g., the lemma holds from Lemma~\ref{lemma-succ-curr-succ} as well.
\end{proof}

\begin{lemma} \label{lemma-range-first-in-range}
Let $n$ be the node, referenced by the $curr$ variable when line~\ref{alg-range-query-init-count} is executed. Then among all nodes that are logically in the list when epoch $ts$ ends, $n$ has the minimal key which is at least $low$.
\end{lemma}

\begin{proof}
Since the condition in line~\ref{alg-range-query-while-curr-smaller} holds for $n$, $n$'s key is at least $low$. In addition, by Lemma~\ref{lemma-range-curr-smaller-loop}, $n$ was logically in the list when epoch $ts$ ended. 
If $n$'s key is $low$, then the lemma holds by Lemma~\ref{lemma-key-once}.
Otherwise, since $n$'s key is bigger than $low$, $n$ is not the first node referenced by the $curr$ variable (see line~\ref{alg-range-query-curr-gets-pred}).
By Lemma~\ref{lemma-range-curr-smaller-loop}, the previous node referenced by the $curr$ variable was also logically in the list, and also $n$'s predecessor, when epoch $ts$ ended.
Since the condition in line~\ref{alg-range-query-while-curr-smaller} did not hold for that node, its key is necessarily smaller than $low$. By Lemma~\ref{lemma-logically-in-not-between}, among all nodes that were logically in the list when epoch $ts$ ended, $n$ had the minimal key which was at least $low$. 
\end{proof}

\begin{lemma} \label{lemma-range-key-in-range}
For every key $k$, written into the output array in line~\ref{alg-range-query-write-key}, it holds that $low \leq k \leq high$, and $k$ belongs to a node that was logically in the list when $ts$ ended.
\end{lemma}

\begin{proof}
Let $k$ be a key, written into the output array in line~\ref{alg-range-query-write-key}, and let $n$ be the node referenced by the $curr$ variable at this point. Since the condition in line~\ref{alg-range-query-while-curr-bigger} holds for $k$, $k \leq high$. In addition, by Lemma~\ref{lemma-invariants} (Invariant~\ref{lemma-invariants-smaller-not-reachable}),~\ref{lemma-range-curr-bigger-loop}, and~\ref{lemma-range-first-in-range}, $low \leq k$.
Finally, by Lemma~\ref{lemma-range-curr-bigger-loop}, $n$ was logically in the list when epoch $ts$ ended. 
\end{proof}

\begin{lemma} \label{lemma-range-no-missing-keys}
Let $n$ be a node that was logically in the list when epoch $ts$ ended, with a key between $low$ and $high$. Then $n$'s key and value were written into the output array in lines~\ref{alg-range-query-write-key}-\ref{alg-range-query-write-value}.
\end{lemma}

\begin{proof}
Assume by contradiction that $n$'s key and value were not written into the output array in lines~\ref{alg-range-query-write-key}-\ref{alg-range-query-write-value}. W.l.o.g., let $n$ be the node with the minimal key, whose key and value were not written. By Lemma~\ref{lemma-range-first-in-range}, $n$ cannot be the node referenced by the $curr$ variable when line~\ref{alg-range-query-init-count} is executed. Let $n_0$ be the node with the maximal node in range which is smaller than $n$'s key.
By the choice of $n$, $n$ was referenced by the $curr$ variable during the loop in lines~\ref{alg-range-query-while-curr-bigger}-\ref{alg-range-query-bigger-curr-gets-succ}.
As $n_0$'s key is smaller than $n$'s key, and $n$'s key is in range, let $n_1$ be the next node, referenced by the $curr$ variable. By the choice of $n$, $n_1$'s key is bigger than $n$'s key. By Lemma~\ref{lemma-range-curr-bigger-loop}, $n$ was not logically in the loop when epoch $ts$ ended -- a contradiction.
\end{proof}

We can now define a linearization point per \emph{rangeQuery()} execution.

\begin{definition} \label{definition-range-query-linearization} \textbf{(rangeQuery - linearization points)}
A \emph{rangeQuery()} execution is linearized when epoch $ts$\footnote{the final $ts$ value} terminates.
\end{definition}

We are now going to prove, using Lemmas~\ref{lemma-range-query-linearization-between} and~\ref{lemma-range-query-linearization-correctness} below, that Definition~\ref{definition-range-query-linearization} above indeed defines adequate linearization points. I.e., that each linearization point occurs between the invocation and termination of the operation, and that the operation indeed takes affect at its linearization point.

\begin{lemma} \label{lemma-range-query-linearization-between}
The linearization point, guaranteed by Definition~\ref{definition-range-query-linearization}, occurs between the invocation and termination of the respective execution.
\end{lemma}

\begin{proof}
If $ts$ is read in line~\ref{alg-range-query-faa}, then epoch $ts$ terminates in this line, which is between the invocation and termination of the respective execution.
Otherwise, $ts$ is the epoch before the current epoch, when line~\ref{alg-range-query-new-ts} is executed for the last time. As the global epoch was incremented at least once during this execution, epoch $ts$ necessarily ended between the invocation and termination of the execution.
\end{proof}

\begin{lemma} \label{lemma-range-query-linearization-correctness}
Let $low$ and $high$ be the input parameters to a \emph{rangeQuery()} execution, and let $S$ be the set of keys, returned using the output array.
Then at the linearization point, defined by Definition~\ref{definition-range-query-linearization}, it holds that $S$ contains exactly the keys of all of the nodes which are logically in the list, with keys between $low$ and $high$.
\end{lemma}

\begin{proof}
The Lemma immediately holds by Lemma~\ref{lemma-range-key-in-range} and~\ref{lemma-range-no-missing-keys}.
\end{proof}

We conclude with the following theorem:

\begin{theorem} \label{theorem-linearizable}
The list implementation, presented in Algorithm~\ref{pseudo-list} and~\ref{pseudo-range}, is a linearizable map implementation.
\end{theorem}

\subsection{Lock-Freedom} \label{sec-lock-free}

We are going to prove that during every possible execution, at least one of the executing threads terminates its operation.
To derive a contradiction, we assume the existence of an execution $E$, in which, starting from a certain point, no thread terminates its operation. 
Let $s$ be the last step executed before this point.
As the total number of threads is finite and bounded, we assume, w.l.o.g., that no operation is invoked after $s$.

As the CAS execution in line~\ref{alg-insert-update-pred} of Algorithm~\ref{pseudo-list} is successfully executed at most once per \emph{insert()} execution, we can assume that the next pointer of a finite number of nodes are updated in line~\ref{alg-insert-update-pred} after $s$. 
In a similar way, as nodes can only be marked once, during a \emph{remove()} execution, we can assume that a finite number of nodes become marked after $s$. 
In addition, by Definition~\ref{definition-unlinked} and Lemma~\ref{lemma-invariants} (Invariant~\ref{lemma-invariants-still-reachable-before-trim} and~\ref{lemma-invariants-unlinked}), there is a one-to-one function between the set of flagged nodes to the set of marked nodes. Therefore, we can also assume that a finite number of nodes become flagged after $s$. 
Finally, by Definition~\ref{definition-unlinked} and Lemma~\ref{lemma-invariants} (Invariant~\ref{lemma-invariants-still-reachable-before-trim} and~\ref{lemma-invariants-unlinked}), there is also a one-to-one function between the set of pointer updates in line~\ref{alg-trim-cas} to the set of marked nodes.
Although the number of pointer changes in lines~\ref{alg-insert-update-next},~\ref{alg-insert-update-prior},~\ref{alg-trim-init-next}, and~\ref{alg-trim-update-prior} (in Algorithm~\ref{pseudo-list}) may still be unbounded, by Lemma~\ref{lemma-invariants} (Invariant~\ref{lemma-invariants-pending-infant}), these pointers belong to infant nodes, which are not reachable.
Therefore, w.l.o.g., we can assume that all of the list pointers (i.e., pointers that belong to nodes that are reachable from the $head$ sentinel node) are immutable after $s$.

Since the list pointers are immutable after $s$, by Lemma~\ref{lemma-invariants} (Invariant~\ref{lemma-invariants-tail}), the loop in lines~\ref{alg-find-while-true}-\ref{alg-find-curr-gets-pred-next} of Algorithm~\ref{pseudo-list} should always terminates. Then, if the condition in line~\ref{alg-find-if-not-adjacent} does not hold, the \emph{find()} execution terminates, and otherwise, the \emph{trim()} method is invoked.

Assume by contradiction that the \emph{trim()} method may indeed be invoked at this stage.
During the \emph{trim()} execution, by Lemma~\ref{lemma-invariants} (Invariant~\ref{lemma-invariants-tail}), the loop in lines~\ref{alg-trim-while}-\ref{alg-trim-curr-gets-curr-next} must terminate, and the condition in line~\ref{alg-trim-flag} must hold (as otherwise, the CAS in line~\ref{alg-trim-cas} must be successful due to the list immutability -- a contradiction).
This means that the node referenced by the $curr$ is not flagged, and the executing thread did not flag it. As this node is necessarily not marked (the condition in line~\ref{alg-trim-while} did not hold for it), we get a contradiction to the list's immutability.

Therefore, every \emph{find()} invocation eventually terminates after $s$. This implies that every \emph{contains()} invocation eventually terminates after $s$ as well.

Moving on to the \emph{rangeQuery()} operation, by Lemma~\ref{lemma-invariants} (Invariant~\ref{lemma-invariants-tail}) and Lemma~\ref{lemma-tail-p-reachable}, the loops in lines~\ref{alg-range-query-while-pred}-\ref{alg-range-query-pred-gets-prior},~\ref{alg-range-query-while-curr-smaller}-\ref{alg-range-query-smaller-curr-gets-succ}, and~\ref{alg-range-query-while-curr-bigger}-\ref{alg-range-query-bigger-curr-gets-succ} must terminate.
It still remains to show that the loop in lines~\ref{alg-range-query-while-find}-\ref{alg-range-query-new-ts} terminates.
The \emph{find()} invocation in line~\ref{alg-range-query-find} must terminate. In addition, by Lemma~\ref{lemma-find-linearization}, every new loop iteration searches for a smaller key. At a worst case scenario, the \emph{find()} invocation in line~\ref{alg-range-query-find} eventually finds $head$, which obviously has a timestamp which is smaller than $ts$, and a key which is smaller than $low$, and the loop terminates.
Therefore, every \emph{rangeQuery()} invocation eventually terminates after $s$ as well.

Now, assume an \emph{insert()} execution after $s$. The \emph{find()} execution in line~\ref{alg-insert-find} must terminate, and the operation must return in line~\ref{alg-insert-return-curr-val}. Otherwise, by the list immutability, the CAS in line~\ref{alg-insert-update-pred} should be successful -- a contradiction to the list immutability.

Finally, assume a \emph{remove()} execution after $s$. By Lemma~\ref{lemma-find-linearization} and the list's immutability, the second output parameter, returned from the \emph{find()} invocation in line~\ref{alg-remove-find}, should not be marked or flagged. Therefore, the execution must return in line~\ref{alg-remove-return-no-val}, as otherwise, the marking in line~\ref{alg-remove-mark} must succeed (which is a contradiction to the list's immutability).

Since every possible scenario derives a contradiction, we conclude with Theorem~\ref{theorem-lock-free}.

\begin{theorem} \label{theorem-lock-free}
The list implementation, presented in Algorithm~\ref{pseudo-list} and~\ref{pseudo-range}, is a lock-free map implementation.
\end{theorem}

\section{VBR Integration - Correctness} \label{sec-vbr-correctness}

As discussed in Section~\ref{sec-mvcc-vbr}, we had to change VBR in order to integrate it correctly into our list implementation. 
More specifically, let $t$ be an executing thread, and assume that $t$ currently has a pointer to a node $n$, with a birth epoch $b$\footnote{As opposed to Appendix~\ref{sec-correctness}, here we must separate between a node's timestamp (related to the global timestamps clock) an birth epoch (related to the global epoch clock, used for VBR).} (already read by $t$). In addition, let $e$ be the last reclamation epoch $t$ has recorded. Then $t$ must follow these three guidelines:
\begin{enumerate}
    \item \label{guideline-ts} If $t$'s next step accesses one of $n$'s fields (e.g., its timestamp, key, value, next pointer or prior pointer), then it must be followed by a read of $t$'s birth epoch. If it is bigger than $e$, then $t$ must rollback to its previous checkpoint.
    \item \label{guideline-next} If $t$'s next step dereferences $n$'s next pointer (assuming it is not null), then $t$ must additionally check that $n$'s successor's birth epoch is not bigger than $n$'s next pointer version. If it is bigger, then $t$ must rollback to its previous checkpoint.
    \item \label{guideline-prior} If $t$'s next step dereferences $n$'s prior pointer (assuming it is not null), then $t$ must additionally check that $n$'s p-successor's birth epoch (see Definition~\ref{definition-p-successor} in Appendix~\ref{sec-correctness}) is not bigger than $b$. If it is bigger, then $t$ must rollback to its previous checkpoint.
\end{enumerate}

\begin{wrapfigure}{R}{0.4\textwidth}
\begin{minipage}{0.4\textwidth}
\begin{algorithm}[H]
\caption{The VBR-integrated Node}
\label{pseudo-vbr}
\begin{algorithmic}

\State \textbf{class} Node \label{alg-vbr-class-node}
\Indent
    \State $\langle$Long ts, Long birth\_epoch$\rangle$
    \State K key
    \State V value
    \State $\langle$Node* next, Long next\_version$\rangle$
    \State Node* prior
    
\EndIndent

\end{algorithmic}
\end{algorithm}
\end{minipage}
\end{wrapfigure}

Recall that nodes are retired after being unlinked (see Definition~\ref{definition-unlinked} in Appendix~\ref{sec-correctness}). I.e., a retired node is never reachable, but may still be p-reachable from a reachable node.
The original VBR paper~\cite{DBLP:journals/corr/abs-2107-13843} rigorously proves that VBR maintains the original implementation's linearizability and lock freedom guarantee.  
Given this proof, we are going to prove that, given our list implementation and the new VBR variant, stale values are always ignored. In addition, we prove that the new VBR variant does not foil the original VBR's lock-freedom.
These two guarantees will be proven by showing equivalence between the two reading mechanisms.
Note that, as updates and checkpoint installations remain unchanged, we do not need to handle them.
For convenience, let us denote the new VBR variant with VBR$_N$, and assume both schemes are integrated into the same lock-free linearizable implementation. 
Note that a rollback during a VBR-integrated execution does not necessarily imply a rollback during the equivalent VBR$_N$-integrated execution. E.g., a thread $t$ may hold a pointer to a node $n$, with a birth epoch which is smaller than the current epoch. If the global epoch changes, a further read of one of $n$'s fields in the VBR-integrated execution, would result in a rollback. However, since $n$'s birth epoch is still smaller than the current epoch, this read would not result in a rollback during the equivalent VBR$_N$-integrated execution.
On the other hand, a rollback during a VBR$_N$-integrated execution does imply a rollback during the equivalent VBR-integrated execution, as we prove in Lemma~\ref{lemma-rollback-equivalence}:

\begin{lemma} \label{lemma-rollback-equivalence} 
If a rollback is performed in the VBR$_N$-integrated execution, then a rollback is also performed in the VBR-integrated execution.
\end{lemma}

\begin{proof}
We prove the lemma with respect to the three VBR$_N$ guidelines.
Let $t$ be an executing thread, and assume that $t$ currently has a pointer to a node $n$, with a timestamp $ts$ and a birth epoch $b$ (already read by $t$). In addition, let $e$ be the last reclamation epoch $t$ has recorded (obviously, $b \leq e$).
Recall that in VBR, the global epoch is re-read upon each (either mutable or immutable) field access, followed by a rollback whenever the global epoch has changed.

Assume that $t$ rollbacks due to guideline~\ref{guideline-ts}. Then after reading one of $n$'s fields, it re-reads $n$'s birth epoch, which is bigger than $b$. By the VBR original invariants, this necessarily implies an epoch change since the last time $t$ read $n$'s birth epoch. Therefore, $t$ would read the global epoch and rollback in the VBR-integrated execution as well (if it has not done so already).

Assume that $t$ rollbacks due to guideline~\ref{guideline-next}. Then after reading $n$'s next pointer, $t$ re-reads $n$'s birth epoch, which is still $b$, followed by a read of $n$'s successor's birth epoch, which is bigger than $e$. 
As birth epochs are set according to the current global epoch, this implies that the global epoch has necessarily changed.
Therefore, $t$ would read the global epoch and rollback in the VBR-integrated execution as well (if it has not done so already).

Assume that $t$ rollbacks due to guideline~\ref{guideline-prior}. Then after reading $n$'s prior pointer, $t$ re-reads $n$'s birth epoch, which is still $b$, followed by a read of $n$'s p-successor's birth epoch, which is bigger than $b$. 
As nodes are always allocated after their p-successors already have a birth epoch (see Lemma~\ref{lemma-invariants}, Invariant~\ref{lemma-invariants-infant} and~\ref{lemma-invariants-pending-infant}), this implies that the referenced node is not $n$'s p-successor. I.e., the global epoch has necessarily changed. Therefore, $t$ would read the global epoch and rollback in the VBR-integrated execution as well (if it has not done so already).
\end{proof}

Lemma~\ref{lemma-rollback-equivalence} directly implies that VBR$_N$ maintains lock-freedom. It still remains to show that stale values are always ignored during VBR$_N$-integrated executions. 

\begin{lemma} \label{lemma-stale-values} 
During a VBR$_N$-integrated execution, a rollback is performed upon each read of a stale value.
\end{lemma}

\begin{proof}
We prove the lemma with respect to our list implementation (as presented in Algorithm~\ref{pseudo-list} and~\ref{pseudo-range}).

Assume by contradiction that at some point, an executing thread $t$ reads a stale value and does not perform a following rollback, for the first time during the execution.

If the stale value is a node's field (i.e., a node's timestamp, key, value, next pointer or prior pointer), and the node has already been reclaimed since $t$ accessed it for the first time, then by Guideline~\ref{guideline-ts}, $t$ must read the node's birth epoch. As $t$ must have recorded the global epoch before accessing the node for the first time, this node's birth epoch is necessarily bigger than $t$'s recorded epoch. I.e., by Guideline~\ref{guideline-ts}, $t$ must perform a rollback -- a contradiction.

Otherwise, if $t$ accesses a node $n$ for the first time, then it must access it either via a next pointer or a prior pointer of some node $n_0$.

First, assume that $n_0$ is already reclaimed by the time $t$ reads the next pointer. By Guideline~\ref{guideline-ts}, and as already explained above, $t$ must perform a rollback after reading the pointer.
Therefore, we can assume that $n_0$ has not been reclaimed at this point.

Now, assume that $n$ is accessed via $n_0$'s next pointer, and that $n$ is not $n_0$'s original successor. Let $n_1$ be $n_0$'s successor  (i.e., $n$ and $n_1$ were allocated from the same address).
Since $n_0$'s next pointer changes during an epoch which is at least the one recorded by $t$, by Lemma~\ref{lemma-deleted-states} and Lemma~\ref{lemma-invariants} (Invariant~\ref{lemma-invariants-active-reachable}), it is still reachable when $n_1$ stops being its successor. I.e., by Definition~\ref{definition-reachability}, $n_1$ is still reachable at this point as well. By Lemma~\ref{lemma-invariants} (Invariant~\ref{lemma-invariants-unlinked}), $n_1$'s retire epoch is at least the global epoch when $t$ reads $n_0$'s next pointer. I.e., $n$'s birth epoch must be bigger than the last epoch, recorded by $t$. By Guideline~\ref{guideline-next}, $t$ must perform a rollback -- a contradiction.

The remaining possible scenario is that $n$ is accessed via $n_0$'s prior pointer. Let $n_1$ be $n_0$'s p-successor  (i.e., $n$ and $n_1$ were allocated from the same address).
By Lemma~\ref{lemma-invariants}, (Invariant~\ref{lemma-invariants-infant} and~\ref{lemma-invariants-pending-infant}), $n_1$'s birth epoch is at most $n_0$'s birth epoch. Therefore, since $n$ and $n_1$ are allocated from the same address, $n$'s birth epoch is strictly bigger than $n_0$'s birth epoch. By Guideline~\ref{guideline-prior}, $t$ must perform a rollback -- a contradiction.

As there are no other types of accesses, and every possible scenario derives a contradiction, threads always rollback after reading stale values.
\end{proof}

Theorem~\ref{theorem-vbr-n} directly derives from Theorem~\ref{theorem-linearizable}, Theorem~\ref{theorem-lock-free}, Lemma~\ref{lemma-rollback-equivalence}, and Lemma~\ref{lemma-stale-values}:

\begin{theorem} \label{theorem-vbr-n}
The implementation presented in Algorithm~\ref{pseudo-list} and~\ref{pseudo-range}, integrated with $VBR_N$, is a safe, lock-free, and linearizable map implementation.
\end{theorem}
\section{Fast Index Integration - Correctness} \label{sec-index-correctness}

After showing how to integrate our list with a safe and lock-free memory reclamation scheme (see Section~\ref{sec-mvcc-vbr} and Appendix~\ref{sec-vbr-correctness}), in this section we prove that adding a fast index (as described in Section~\ref{sec-index}), does not foil lock-freedom or linearizability.

Integrating the code from Figure~\ref{pseudo-index} obviously does not foil lock-freedom, as the loop in lines~\ref{alg-index-while}-\ref{alg-index-break} stops after a finite and predefined number of iterations.
Therefore, we only need to show that the integration does not foil linearizability.

\begin{lemma} \label{lemma-index-pred}
Let $n$ be the node referenced by the $pred$ variable whe the code in Figure~\ref{pseudo-index} terminates. Then $n$'s key is smaller than $key$, and when line~\ref{alg-index-pred-next} was executed for the last time, $n$ was logically in the list, and neither marked nor flagged.
\end{lemma}

\begin{proof}
The lemma obviously holds if the condition in line~\ref{alg-index-if-attempts} holds, since $head$'s key is necessarily smaller than the input key, and it is always logically in the list.

Otherwise, a pointer to $n$ was returned from the \emph{findPred()} call in line~\ref{alg-index-find-pred}.
As the condition in line~\ref{alg-index-continue} does not hold for $n$, its key is indeed smaller than $key$.

Assume by contradiction that $n$ is not the node, encountered during this \emph{findPred()} execution.
Since the condition in line~\ref{alg-index-rollback} did not hold for $n$, it is guaranteed that $n$ was allocated before this \emph{findPred()} invocation -- a contradiction.

It is still possible that the index is not up to date. I.e., a different node was allocated from $n$'s address before $n$, and its respective index node has not been removed from the index yet.
However, regardless of this scenario, since the conditions in lines~\ref{alg-index-continue} and~\ref{alg-index-if-marked} did not hold for $n$, it was an active node (see Definition~\ref{definition-node-states} in Appendix~\ref{sec-correctness}) before line~\ref{alg-index-pred-next} was executed for the last time. By Definition~\ref{definition-logically-in}, $n$ was logically in the list when line~\ref{alg-index-pred-next} was executed for the last time. In addition, since the condition in line~\ref{alg-index-if-marked} did not hold for $n$, by Lemma~\ref{lemma-node-states}, it was neither marked nor flagged when line~\ref{alg-index-pred-next} was executed for the last time.
\end{proof}

Theorem~\ref{theorem-index-correctness} directly derives from Theorem~\ref{theorem-linearizable},
Theorem~\ref{theorem-lock-free},
Theorem~\ref{theorem-vbr-n} and Lemma~\ref{lemma-index-pred}:

\begin{theorem} \label{theorem-index-correctness}
The implementation presented in Algorithm~\ref{pseudo-list} and~\ref{pseudo-range}, integrated with $VBR_N$, and with an external index (as described in Figure~\ref{pseudo-index}), is a safe, lock-free, and linearizable map implementation.
\end{theorem}

\end{document}